\documentclass[11pt]{article}

\usepackage{hyperref,titlesec}
\usepackage{amsfonts}
\usepackage{amsmath,amsthm}
\usepackage{graphicx}
\usepackage{epstopdf}
\usepackage{algorithmic}
\usepackage{bm}
\usepackage{dsfont}
\usepackage{subfigure}
\usepackage{amsfonts}
\usepackage{enumerate}
\usepackage{amssymb,eurosym}
\usepackage{enumitem}
\usepackage{multirow}
\usepackage{blkarray,pifont}
\usepackage{cleveref}

\newtheorem{thm}{Theorem}[section] 
\newtheorem{defn}[thm]{Definition} 
\newtheorem{example}[thm]{Example}
\newtheorem{pro}[thm]{Proposition}
\newtheorem{cor}[thm]{Corollary}
\newtheorem{asp}[thm]{Assumption}

\theoremstyle{remark}
\newtheorem{remark}{Remark}

\newcommand{\E}{\mathbb{E}}

\newcommand{\tM}{\tilde{M}}
\newcommand{\cM}{\check{M}}

\newcommand{\ba}{\bm{\alpha}}

\newcommand{\ts}{\sigma}
\newcommand{\tp}{{P}}
\newcommand{\tf}{\tilde{\mathcal{F}}}
\newcommand{\tn}{\eta}
\newcommand{\ms}{\check{s}}

\begin{document}

\title{Stochastic Switching Games}

\author{{Liangchen Li}
	and {Mike Ludkovski}
		\thanks{Ludkovski is partially supported by NSF-CDSE 1521743.}}

\date{Department of Statistics and Applied Probability, \\ University of California, Santa Barbara, Santa Barbara, CA, USA 93106-3110 \\ \url{{li,ludkovski}@pstat.ucsb.edu} }

	\maketitle
	
	\begin{abstract}
		We study nonzero-sum stochastic switching games. Two players compete for market dominance through controlling (via timing options) the discrete-state market regime $M$. Switching decisions are driven by a continuous stochastic factor $X$ that modulates instantaneous revenue rates and switching costs. This generates a competitive feedback between the short-term fluctuations due to $X$ and the medium-term advantages based on $M$.
We construct threshold-type Feedback Nash Equilibria which characterize stationary strategies describing long-run dynamic equilibrium market organization. Two sequential approximation schemes link the switching equilibrium to (i) constrained optimal switching; (ii) multi-stage timing games. We provide illustrations using an Ornstein-Uhlenbeck $X$ that leads to a recurrent equilibrium $M^\ast$ and a Geometric Brownian Motion $X$ that makes $M^\ast$ eventually ``absorbed'' as one player eventually gains permanent advantage. Explicit computations and comparative statics regarding the emergent macroscopic market equilibrium are also provided.
	\end{abstract}
	
	

\section{Introduction}

Dynamic competition under uncertainty has motivated an ongoing stream of literature in Operations Research and Industrial Organization. It provides a natural generalization from the classical one-agent optimization problems and hence is applicable in a large variety of applied settings. In a typical setup, firms aim to maximize their profits while being exposed to exogenous stochastic shocks and competitive effects. The latter competition is generally indirect, e.g.~through price, and hence leads to non-zero-sum game formulation, as opposed to the adversarial, zero-sum case. Some examples include: (i) capacity expansion \cite{boyer2012dynamic}; (ii) technological innovation and adoption \cite{HuismanKort04}; (iii) market entry \cite{HubertsDawid15}.

In this paper we focus on a specific class of \emph{switching games} which are inspired by the two time-scale feature of many markets. Indeed, dynamic competition is often driven by infinitesimal shocks that determine the rapidly fluctuating short-run market conditions. These fluctuations yield ``local'' advantage to players. In turn, the firms convert such short-term effect into a more durable gain through market dominance. For example, advantageous investment costs in the present can be converted into longer-term capacity gains that will yield a larger market share. Likewise, cheaper R\&D costs can be locked-in through technological edge. The overall model then links exogenous stochastic shocks with the endogenized firms' decisions to obtain the dynamic equilibrium for the market organization.

This narrative is formalized by considering a ``microeconomic'' stochastic factor $(X_t)$ that drives market conditions vis-a-vis the ``macroeconomic'' market state $(M_t)$ that determines market power and relative profits of the firms. We then aim to solve for the competitive \emph{equilibrium} that determines $(M_t)$ from the evolution of $X$. Specifically, $(M_t)$ is taken to be fully controlled by the players through discrete/lumpy actions. In turn the actions are interpreted as controls available to the players and are affected by the local fluctuations captured by $(X_t)$. As a result, we consider a dynamic stochastic game with a micro/macro structure: in the short-term firms treat $M_t$ as fixed and optimize their next action like in a timing pre-emption game. In the long term, $(M_t)$ is a stationary market regime with endogenous market dynamics.

The two time-scales link the immediate competitive advantage identified by $X_t$ and the long-run market organization $M_t$. In effect, the latter ``tracks'' $X_t$: after a sustained period of advantageous market conditions we expect the respective firm to become dominant.  A lag is naturally introduced due to the fixed investment costs that create an opportunity cost. The novelty of our setup is to fully integrate this well-known idea within a non-cooperative game model. Thus, the players continuously compete for market dominance, with their actions influenced both by the short-term stochastic fluctuations driving $X_t$ and the strategic preemption/competition effects.
We focus on a duopoly that naturally mirrors the two-sided nature of the up/down market conditions represented by the one-dimensional $X$ and clarifies the competitive effects of responding to rival actions.

Our methodological interest in this model stems from three different directions. First, it extends our previous work \cite{aid2017capacity} on multi-stage timing games. In that version, the number of controls available to the players was a priori restricted; here we consider the more plausible situation of an infinitely-repeated game. One economic motivation is the capacity expansion problem under a growing stochastic environment (e.g.~demand) $X$. This yields a non-stationary model but the ultimate number of aggregate investments is unbounded, and must be modeled by a switching game. Second, we are interested in a \emph{stationary} switching game, where the market undergoes cyclical shocks (in the sense of $X$ being a recurrent Markov process).  We wish to find the endogenous dynamic equilibrium that will mirror this cyclicality through the strategically adjusted market regime. Describing such a recurrent stochastic investment-timing competition naturally links to switching duopoly games. Third, our model is motivated by the desire for tractability while allowing for dynamic cross-effects due to competition and stochastic shocks. The repeated stationary nature of the competition allows to remove the time-variable but still maintain the stochastic dynamics. As a result, the equilibrium structure is intuitive (summarized through switching thresholds) yet brings novel insights.

The switching game we consider naturally merges the single-agent switching models \cite{BrekkeOksendal94,HamadeneJeanblanc07,CarmonaLudkovski08,bayraktar2010one,ly2007explicit,PhamZhou09,JohnsonZervos10} and the non-zero-sum stopping games \cite{HamadeneZhang10,FerrariMoriarty15,Attard18,MartyrMoriarty17,RiedelSteg17}. Optimal switching arises in a number of important applications, and captures the market structure of a stochastic process $(X_t)$ linked with a controlled regime process $(M_t)$. Most relevant are the multi-mode models \cite{Pham07,PhamZhou09,HuTang10}. In particular, we leverage the related analytical results about the variational inequalities satisfied by the value functions and the construction of threshold-type strategies. We also make extensive use of the finite-control approximation \cite{bayraktar2010one} that offers an intuitive way to establish the Dynamic Programming Principle. From the game-theoretic perspective, switching controls are best approached as multi-stage stopping, making the respective analysis an essential building block. In turn a non-zero-sum stopping game requires handling the solution of an optimal stopping problem with an exit constraint \cite{leung2015optimal,MartyrMoriarty17,aid2017capacity}. Together, the above two strands provide the main intuition for the dynamics of the switching game:
\begin{itemize}
  \item As a multi-stage sequence of non-zero-sum stopping games;
  \item As a fixed point of associated best-response (constrained) optimal switching problems.
\end{itemize}
Through these lens, one may consider further versions by modifying the type of controls available to the players: path-dependent non-Markovian strategies (viewed via BSDEs); impulse controls (see \cite{aid2016nonzero}); singular controls (see e.g.~\cite{ChevalierVath13} for optimal switching with singular controls) and so on.

\section{Problem Formulation}\label{sec:PF}
We consider two firms, dubbed player $i,j\in\{1,2\}, i\neq j$, competing on the same market. The macro market state is described by a discrete-state process $(M_t)$  and represents the relative market dominance of each player. The domain of $(M_t)$ is a finite set $\mathcal{M}$; for simplicity we consider integer-valued $M_t$ and $\mathcal{M} = \{ \underline{m}, \underline{m}+1, \ldots, \overline{m}\}$. The players exercise switching-type controls to \textit{enhance} their market dominance; thus, Player $1$ can \emph{increase} $M_t$ by +1, and Player 2 can \emph{decrease} $M_t$ by -1. To exercise a switch, player $i$ must pay a cost $K^i(X_t, M_t)$.
Note that a switch by Player 1, followed by a switch by Player 2 completely neutralize each other and bring the market to its original state. The interpretation of $M_t$ as a relative dominance can be motivated by taking $M_t =M^1_t - M^2_t$, where $M^i_t \in \mathbb{N}$ represents the production capacity, or technology level of firm $i$. Thus, players repeatedly make competing investments to increase their capacity; investments by Player 1 raise $M_t$ and those by Player~2 lower it. The assumption that $(M_t)$ only moves between adjacent regimes $\Delta M = \pm 1$ can be easily relaxed, see \cref{rem:general-switching}.

To capture the local fluctuating market condition, we introduce an {exogenous} diffusion process $\left(X_t\right)_{t\geq0}$ on a probability space $(\Omega, \mathcal{F}, \mathbb{P})$, satisfying the stochastic differential equation (SDE)
\begin{align}\label{eq:PF_sde}
dX_t = \mu(X_t) \, dt + \sigma(X_t)\, dW_t,
\end{align}
where $\left(W_t\right)_{t\geq0}$ is a standard Brownian motion under $\mathbb{P}$. Denote by $\mathcal{D}:=(\underline{d}, \bar{d})$, with $-\infty\leq \underline{d}<\bar{d}\leq+\infty$, the domain of $(X_t)$ and $\mathbb{F}:=\left(\mathcal{F}_t\right)_{t\geq 0}$ the natural filtration generated by $(X_t)$. The coefficients $\mu: \mathcal{D}\to \mathbb{R}$ and $\sigma: \mathcal{D} \to \mathbb{R}_{+}$ are assumed to ensure a unique strong solution to \eqref{eq:PF_sde}. Moreover, we assume the boundaries of $\mathcal{D}$ are \textit{natural}, and $X$ is \textit{regular} in $\mathcal{D}$\footnote{Informally this means that starting at any $x \in \mathcal{D}$, $X$ will reach any $y \in \mathcal{D}$ with positive probability, and $\underline{d}, \bar{d}$ cannot be reached in finite time; see Ch.~2 of \cite{borodin2012handbook} for detailed exposition.}.

\begin{remark}
  More generally, one may take the coefficients $\mu,\sigma$ in \eqref{eq:PF_sde} to depend on $M_t$ which presents no technical difficulties beyond heavier notation. We revisit this extension in Section \ref{sec:Conc}.
\end{remark}

Players aim to maximize their expected future (discounted) profits on $[0,\infty)$ defined through revenue rates $\pi^i$'s that are driven by $\left(X_t, M_t\right)$. The integrated total profit is then given by $\int_0^\infty  e^{-r s} \pi^i(X_s, M_s)ds$ minus the net present value of switching costs, see \cref{def:game payoff} below. To match the intuition about the role of $(M_t)$, we postulate that:
(\textit{i}) Player 1 (resp.~P2) is dominant when $M_t>0$ (resp. $M_t < 0$);
(\textit{ii}) Player 1 (resp.~P2) prefers \textit{higher} (resp. \textit{lower}) $X_t$. The last assumption creates a positive feedback effect between $X$ and $M$: as $X$ rises, Player 1 gets more motivated to enhance her market dominance, eventually triggering her to act and make $M_t$ higher too; when $X$ falls sufficiently Player 2 gains short-term advantage and moves $M_t$ towards her preferred negative direction.

By way of illustration, we consider the following two representative examples:
\begin{example}[Mean-reverting Advantage]\label{eg:PF_MA}
	Local market fluctuations are mean-reverting, modeled by an Ornstein-Uhlenbeck process
	\begin{align*}
	dX_t=\mu(\theta-X_t)dt+\sigma dW_t,
	\end{align*}
	with $\mathcal{D}=\mathbb{R}$, $\mu,\sigma>0$ and $\theta\in\mathbb{R}$. Thus, the long-run market is stationary and market organization is expected to undergo a cyclical behavior as $X$ stochastically oscillates around $\theta$.  The players receive \textit{constant} profit rates based on deterministic profit ladders  $\pi^i_m$ that are independent of $X_t$ with
	\begin{align*}
	\pi^1_m < \pi^1_{m+1} \qquad\text{and}\qquad \pi^2_m > \pi^2_{m+1} \quad \forall m.
	\end{align*}
	Thus, Player~1 maximizes her revenue when $M_t$ is high and Player~2 when $M_t$ is low. In complement, the present market conditions $X_t$ affect the switching or investment costs.  Thus, when $X_t$ is high/low, $K^1$ is low/high ($K^2$ is high/low). Economically this could be interpreted as $X$ representing exchange rate, with dollar-denominated investment costs both for the domestic firm P1 and foreign firm P2.  For concreteness, we suppose the switching costs are exponential in $X_t$:
	\begin{align*}
	K^i(x, m) = c^i(m) + \alpha^i(m)e^{\beta^i(m)\cdot x}, \qquad i = 1, 2,
	\end{align*}
	where $c^i(m),\alpha^i(m) >0$, $\beta^1(m)<0$ and $\beta^2(m) >0$.
\end{example}

\begin{example}[Long-run Advantage]\label{eg:PF_LA} In the second example we suppose that in the long-run one player will possess the competitive advantage and become dominant. However, in the medium-term fluctuations $X$ creates uncertainty in $M$. This is captured by using a Geometric Brownian Motion (GBM) for $X$
	\begin{align}\label{eq:gbm}
	dX_t=\mu X_t \, dt+\sigma X_t \, dW_t,
	\end{align}
	with $\mathcal{D}=(0, +\infty)$, $\mu\in\mathbb{R}$ and $\sigma>0$. The players receive profit rates according to predetermined profit ladders $\pi^i_m$ as well; for the sake of diversity we use linear switching costs,
	\begin{align*}
	K^i(x, m) = \left[c^i(m) + \beta^i(m) \cdot x\right]_+,\qquad i = 1, 2
	\end{align*}
	where $\beta^1(m) < 0 $ and $\beta^2(m) >0$.
\end{example}

\subsection{Admissible Strategies}
As explained, we view $X$ as an exogenous stochastic risk factor(s), and $M$ as a fully endogenous macro market regime controlled via the joint strategy profile of the two players. Because $M_t$ affects the profitability of both players and brings a negative externality---as Player 1 profit rate increases, Player 2 revenue falls---we use a noncooperative game paradigm, with the players investing while taking into account that the rival will in the future switch  back to her unfavored states. In line with the discrete nature of $M$ we postulate the players adopt \textit{timing strategies}. Thus, the strategic interaction between them leads to a  non-zero-sum stochastic switching game.

To define game strategies, we need to introduce some technical constructs needed to precise closed-loop equilibrium. Informally, closed-loop strategies are based on the history of $(X_t)$ and the history of past switches. Because multiple switches are possible, $(M_t)$ is not sufficient on its own for that purpose.

We denote strategy of player $i$ by $\ba^i:=\{\tau^i(n):n \geq 1 \}$ where $\tau^i$ are certain stopping times. Admissibility of $\tau^i(n)$ is defined recursively, based on the initial state $(x,m)$. Let $(\mathcal{F}_t)_{t\geq0}$ be the natural filtration generated by $(X_t)$. Set $\ts_0=0$, $X_0=x$, $\tM_0=m$. For $n\geq 1$, we require $\tau^i(n)$ to be $\tf^{(n)}$-adapted and set
\begin{subequations}\label{eq:PF_construct}
	\begin{align}
	&\tf^{(n)}_t = \mathcal{F}_t \bigvee \sigma\big\{(\ts_k, \tp_k, \tM_k), k<n\big\}\label{eq:PF_constructF},\\
	&\sigma_{n} =  \tau^1(n) \wedge \tau^2(n),\\
	& \tp_{n} = 1 \cdot \mathds{1}_{\{\tau^1(n) < \tau^2(n)\}} + 2 \cdot \mathds{1}_{\{\tau^1(n) > \tau^2(n)\}} +\mathcal{H}_n \cdot\mathds{1}_{\{\tau^1(n) = \tau^2(n)\}},\label{eq:PF_constructP}\\
	&\tM_{n}=\tM_{n-1} +1 \cdot \mathds{1}_{\{\tp_{n}=1\}}- 1 \cdot \mathds{1}_{\{\tp_{n}=2\}}.\label{eq:PF_constructM1}
	\end{align}
\end{subequations}
The meaning of $n=1,\ldots,$ is the counter for the overall ``round'' of the switching game, with $\sigma_n$ recording the corresponding $n$-th switch time, $\tp_n$ the identity of the player who makes the $n$-th switch, and $\tM_n$ the macro market state after $n$ total switches are exercised.

In \eqref{eq:PF_constructP} we address scenarios that both players intend to switch at the same time by letting $\mathcal{H}_n$ denote the identity of the resulting leader. As a simple example, $\mathcal{H}_n\equiv 1$ if Player 1 has the instantaneous priority to switch. In general, resolving $\mathcal{H}_n$ requires consideration of auxiliary discrete-stage game \cite{grasselli2013priority,RiedelSteg17} that happens instantaneously at $\underline{\tau}_m$ on the event $\{\tau^1(n) = \tau^2(n)\}$; the latter could involve mixed strategies, i.e.~there is an additional random variable $\omega_t$ that determines the value of $\mathcal{H}_n$. This is another reason why we must explicitly augment the history of $\tp_k$ to the history of $(X_t)$ in \eqref{eq:PF_constructF}.

\begin{defn} (\textit{Admissible Strategies}) \label{def:admstrat}
The set of admissible closed-loop strategies $\mathcal{A}$ is $\ba^i:=\{\tau^i(n):n \geq 1 \}$ where $\tau^i(n)$ is adapted to $(\tf^{(n)}_t)$ with $\ts_n$, $\tp_n$, $\tM_n$ constructed in \eqref{eq:PF_construct}, and satisfying
	\begin{itemize}
		\item bounded domain of $(M_t)$: $\tau^1(n)=+\infty$ if $\tM_n=\overline{m}$, $\tau^2(n)=+\infty$ if $\tM_n=\underline{m}$;
		\item ordered in time: $\tau^i(n) \geq \ts_{n-1}, i\in\{1, 2\}, \quad \forall n\geq 1$;
		\item defined for all times: $\lim_{n\rightarrow\infty}\ts_n = +\infty$.
	\end{itemize}
\end{defn}

Note that strategies  in $\mathcal{A}$ are of Closed Loop Perfect State (CLPS) type, see a detailed exposition in  \cite[Ch.~3]{carmona2016lectures}. The three admissibility conditions state that only one player can act when $(M_t)$ is at its minimum/maximum value and also rule out a clustering of switches in finite time. The latter restriction $\lim_{n\rightarrow\infty}\ts_n = +\infty$ is mild, as it would be sub-optimal to make infinite switches, as soon as there are some switching costs. Note that \cref{def:admstrat} is joint over the profile $(\ba^1,\ba^2)$ and also depends on the initial condition. In the sequel we suppress this dependence for lighter notation.

Let us revisit the construction \eqref{eq:PF_construct} with $(\ba^1,\ba^2)$ denoting these players' strategies.  Given the strategy profile $(\ba^1,\ba^2)$,  the evolution of $(M_t)$ is admitted as
\begin{equation}\label{eq:PF_constructM2}
M_t := \tM_{\tn(t)}, \quad\text{with}\quad \tn(t) = \max\{n\geq0: \sigma_n\leq t\}.
\end{equation}
It is entirely possible and feasible that one player acts immediately, $\tau^{i}(n) = \sigma_{n-1}$, in which case $\sigma_{n} = \sigma_{n-1}$, hence $(M_t)$ formally undergoes multiple changes simultaneously. Furthermore, we describe the sequence of switching times realized by \textit{each player}, denoted by $\sigma^i_k$, $i\in\{1,2\}$, $k\geq1$ as
\begin{align}\label{eq:PF_constructIk}
\sigma^i_k:= \sigma_{\tn(i, k)},\quad\text{with}\quad \tn(i, k)=\min\{n\geq1:\sum^n_{l=1}\mathds{1}_{\{P_l=i\}}=k\}.
\end{align}

\subsection{Game Payoffs}
The game payoffs $J^i$'s received by these players are the total net present value (NPV) of future profits, namely the discounted expected future cashflows, minus the discounted switching costs.

\begin{defn}\label{def:game payoff}
	(\textit{Game Payoffs}) Given a strategy profile $(\ba^1,\ba^2)$, the NPV of future profits received by player $i$ is
	\begin{align}\label{eq:PF_game payoff}
	J^i_{m}(x; \bm{\alpha}^1, \bm{\alpha}^2) & := \E\left[\int_0^{\infty}e^{-rt}\! \pi^i\big(X_t,\tM_{\tn(t)}\big)dt-\sum_{n=1}^\infty \mathds{1}_{\{\tp_n=i\}}e^{-r\ts_n}\cdot K^i\big(X_{\ts_n}, \tM_{n-1}\big) \,\middle|\,X_0=x,M_0=m\right],\nonumber\\
 & = \E\left[\int_0^{\infty}e^{-rt}\! \pi^i\big(X_t,M_t\big)dt-\sum_{k } e^{-r \sigma^i_k} K^i\big(X_{\sigma^i_k}, { \tM_{\tn(i, k)-1}}\big) \,\middle|\,X_0=x,M_0=m\right]
	\end{align}
	where $r>0$ is the constant discount rate.
\end{defn}

Let us introduce the static discounted future cashflows
\begin{align}\label{eq:MS_dcf}
D^i_m(x):=\mathbb{E}\left[\int^\infty_0 e^{-rt}\pi^i(X_t, m)dt\,\middle|\,X_0=x\right],
\end{align}
which are assumed to satisfy the growth condition $D^i_m(x) \leq C(1+|x|)$ for $i\in\{1,2\}$ and all $m\in\mathcal{M}$. Because switching costs are non-negative, game payoffs are also of linear growth since they are dominated by $D^i$'s, in particular $J^1_m(x) \leq D^1_{\overline{m}}(x)$ while $J^2_m(x) \leq D^2_{\underline{m}}(x)$.

\subsection{Nash Equilibrium: Fixed Points of Best-response}

To describe players' behavior in this nonzero-sum game we utilize  Markov Nash equilibria (MNE).
\begin{defn}\label{defn:ne}(\textit{Nash Equilibrium})
	Let $J^i_m(x,\cdot)$ denote the game payoff received by player $i$ with $X_0=x, M_0 = m$. The strategy profile $(\bm{\alpha^{1,\ast}},\,\bm{\alpha^{2,\ast}})\in\mathcal{A}$ is said to be a \textit{Nash equilibrium} of the switching game if for any $x \in \mathcal{D}, m\in\mathcal{M}$ and
strategy $\ba^i(x)$ of player $i$ such that $(\ba^i(x), \ba^{j,\ast})$ is admissible
	\begin{equation}\label{eq:PF_ne}
	J^i_m(x;\,\ba^i(x),\ba^{j,\ast})\leq J^i_m(x;\,\ba^{i, \ast},\ba^{j,\ast}).
	\end{equation}
	 The corresponding $V^i_m(x) := J^i_m(x;\,\ba^{i,\ast},\ba^{j,\ast})$ is then named the \textit{equilibrium payoff} of player $i$.
\end{defn}

The Nash equilibrium criterion \eqref{eq:PF_ne} characterizes equilibrium strategies as a fixed point of each player's \textit{best-response} to her rival's strategy. Specifically, given an arbitrary rival's strategy $\bm{\alpha}^j$ define the resulting best-response payoff of player $i$
\begin{align}\label{eq:ES_BestRinf}
\widetilde{V}^i_m(x\,;\bm{\alpha}^j):= \sup_{\{\ba^i:(\ba^i, \ba^j)\in\mathcal{A}\} }J^i_{m}(x; \ba^i, \bm{\alpha}^j),\qquad x\in\mathcal{D},m\in\mathcal{M}.
\end{align}
Because (taking Player 1 as an example) game payoffs satisfy $ D^1_{\overline{m}}(x)  \geq \widetilde{V}^1_m(x; \bm{\alpha}^2) \geq J^1_m(x;\bar{\ba}^1, \ba^2) \geq D^1_{\underline{m}}(x),$
such best-response values are always well-defined. Equilibrium payoffs then satisfy:
\begin{align}\label{eq:PF_ChNE}
V^i_m(x)=\widetilde{V}^i_m(x\,;\bm{\alpha}^{j,\ast}), \quad i\in\{1, 2\}, \qquad j\neq i.
\end{align}

\section{Constructing an Equilibrium}
We now focus on a special class of strategies, which are \textit{stationary} and of \textit{threshold-type}, and allow us to explicitly construct a MNE. To do so, two key properties are needed. First, one must show that this class of strategies is closed under the best-response map~\eqref{eq:ES_BestRinf}. Second, a verification theorem is needed to show that the resulting fixed point of \eqref{eq:PF_ChNE}, defined through a system of equations, is indeed a MNE of the game. The programme starts in Section \ref{sec:SFPS} where we define threshold-type strategies and then characterize the best-response to such strategies as a solution to a system of coupled optimal stopping problems. Next, in Section \ref{sec:BVT} we  state the verification theorem which provides a system of nonlinear equations for the equilibrium threshold vectors $\bm{s}^{1,\ast},\bm{s}^{2,\ast}$. Lastly in Section \ref{sec:SMG_summary} we study the emerging equilibrium macro state $M^*$.

\subsection{Stationary and Threshold-type Strategies}\label{sec:SFPS}
The time-stationary Markovian strategies, also known as Feedback Perfect State (FPS) type defined in \cite[Ch.~3]{carmona2016lectures} depend only on the current $X_t$ and $\tM_{\eta(t)}$. Following the idea of a similar construction  in \cite{aid2016nonzero}, we define a strategy of player $i\in\{1, 2\}$ by $\ba^i:=\left(\Gamma^i_m\right)_{m\in\mathcal{M}}$, where $\Gamma^i_m$'s are fixed subsets of $\mathcal{D}$. Given a strategy profile $(\ba^1, \ba^2)$, a sequence of switches is uniquely determined as follows:
\begin{itemize}
	\item [---] when $M_t=m$,  player $i$ adopts the (feedback) switching region $\Gamma^i_m$: player $i$ exercises a switch (changes $M^i_t$ by $\pm 1$) at the first hitting time $\tau^i_m$ of $(X_t)$ to $\Gamma^i_m$ (with the convention that the hitting time of an empty set is $\infty$);
	\item [---] if both players want to switch, Player 1 has the priority.
\end{itemize}
Admissibility of the strategy profile $(\ba^1, \ba^2)$ in \cref{def:admstrat} now reduces to
	\begin{itemize}
	\item [---] $\Gamma^1_{\overline{m}}=\Gamma^2_{\underline{m}} = \emptyset$ ($M_t \in \mathcal{M}$)
	\item [---] $\Gamma^1_m\cap \Gamma^2_{m+1}=\emptyset$ for $m<\overline{m}$ and $\Gamma^1_{m-1}\cap \Gamma^2_{m}=\emptyset$ for $m>\underline{m}$. This rules out simultaneous switching loops; for instance if there were an $x \in \Gamma^1_m\cap \Gamma^2_{m+1}$ then starting in regime $m$, we would have that P1 switches up to $m+1$, but them immediately P2 switches back down to $m$, generating an infinite sequence of instantaneous switches.
\end{itemize}

Relying on the resulting Markov structure of threshold-type strategies, we revisit the formal game evolution which can now be constructed using independent auxiliary copies $\tilde{X}^{(n)}, n=1,\ldots,$ of the strong Markov $X$. Below, $X^x$ denotes the $X$-process started at $X_0 = x$.
Let $x\in\mathcal{D}$, $m\in\mathcal{M}$, and a strategy profile $(\ba^1, \ba^2)\in\mathcal{A}$. Set $\sigma_0=0$, $X_0=x$ and $\tM_0=m$. For $n\geq 0$, define
\begin{subequations}
\begin{align}
&\tilde{X}^{(n)}_t = X^x_{\sigma_n+t},\qquad\text{for } t\geq 0,\label{eq:STS_constructX}\\
&\tilde{\tau}^{i, n} = \inf\{s\geq 0\,:\,\tilde{X}^{(n)}_s\in \Gamma^i_{\tM_{n}}\},\qquad i\in\{1,2\},\label{eq:STS_constructTau}\\
&\sigma_{n+1} = \sigma_n + \tilde{\tau}^{1,n} \wedge \tilde{\tau}^{2,n},\label{eq:STS_constructsig}\\
& P_{n+1} = 1 \cdot \mathds{1}_{\{\tilde{\tau}^{1,n} < \tilde{\tau}^{2,n}\}} + 2 \cdot \mathds{1}_{\{\tilde{\tau}^{1,n} > \tilde{\tau}^{2,n}\}} + \mathcal{H}_{n+1} \mathds{1}_{\{\tilde{\tau}^{1,n} = \tilde{\tau}^{2,n}\}},\\
&\tM_{n+1}=\tM_n +1 \cdot \mathds{1}_{\{\tp_{n+1}=1\}}- 1 \cdot \mathds{1}_{\{\tp_{n+1}=2\}}.\label{eq:STS_constructM1}
\end{align}
\end{subequations}
Then the evolution of $(M_t)$ and the sequence of switching times of each player $(\sigma^i_k)_{k\geq1}$ are obtained as in \eqref{eq:PF_constructM2} and \eqref{eq:PF_constructIk}. The strong Markov property of $X$ implies that each $\tilde{X}^{(n)}$ can be considered as a fresh (independent) copy of $X$ starting at $\tilde{X}^{(n)}_0 = X^x_{\sigma_n}$. Consequently, given these players' strategies, the pair $(X_t, M_t)$ is Markovian.

Recall that Player 1 is in favor of high $X_t$ and large $M_t$, while Player 2 prefers the opposite; it is therefore natural to assume that P1 switches up when $X$ becomes high enough and P2 switches down when $X$ becomes low enough.
\begin{defn}\label{def:thstrat}
	(\textit{Threshold-type Strategies}) Let $\bm{s}^i:=(s^i_m)_{m\in\mathcal{M}}$ be a vector which characterizes subsets of $\mathcal{D}$ for the switching regions $\Gamma^i_m$ of player $i\in\{1,2\}$ according to
	\begin{align}
	\Gamma^1_m \equiv \Gamma^1_m( \bm{s}^1) :=[s^1_m, \overline{d}),\qquad\text{and}\qquad\Gamma^2_m \equiv \Gamma^2_m( \bm{s}^2) :=(\underline{d},s^2_m].
	\end{align}
	A strategy associated to $(\Gamma^i_m)_{m\in\mathcal{M}}$ is called of threshold-type and denoted by $\bm{s}^i$.
\end{defn}

\begin{figure}[bh!]
	\centering
	\includegraphics[width=0.7\textwidth,height=2.5in]{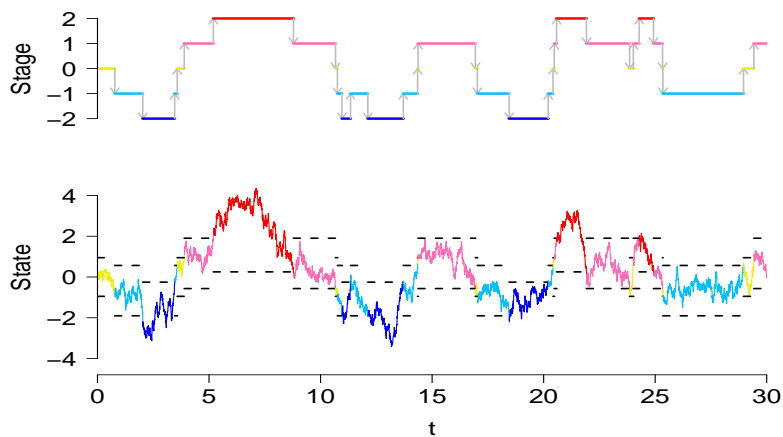}
	\caption{\label{fig:MS_Sample}A trajectory of $X$ and equilibrium $M^\ast$ starting at
		$X_0=0$, $M^\ast_0=0$. Here $X$ is an Ornstein-Uhlenbeck process and  $\mathcal{M}=\{-2, -1, 0,1, 2\}$. The equilibrium strategies are of \textit{threshold-type};  the dashed lines in the bottom plot indicate the respective switching thresholds $s^i_m$. Thus, the switching times $\sigma_n$ are hitting times of the above.}
\end{figure}

In \cref{fig:MS_Sample} we sketch the emerging equilibrium based on threshold-type strategies associated to one of our case studies. The players make a switch whenever the process $(X_t)$ hits the threshold $s^{1,\ast}_m$ from below or $s^{2,\ast}_m$ from above when at stage $m$, see the dashed lines in the bottom plot. The switching times $\sigma^i_k$ are described through the respective hitting times.  The top panel shows the resulting macro stage $(M^\ast_t)$ driven by $\sigma^i_k$'s along one realized trajectory of the local market fluctuations $(X_t)$. These players are ``at equal strength'' in the beginning, $M^\ast_0 = 0$; as $(X_t)$ drops, it enters Player 2's switching region first ($\tau^2(1)<\tau^1(1)$) leading her to exercise a switch and change $M^\ast_{\sigma^2_1} = -1$. The players then recursively wait for $(X_t)$ to hit either the threshold $s^{1,\ast}_{-1}$  or $s^{2,\ast}_{-1}$ ($\tau^1(2)\wedge\tau^2(2)$), to make further switches.

Note that in the above definition we require $\Gamma^i_m$ to be connected, so that they are fully characterized by their boundary $s^i_m$. In turn, threshold-type strategies allow to move from looking at the unstructured (in the sense of optimization) switching strategies defined by general $\Gamma^i_m$ to searching for equilibria parametrized by the $|\mathcal{M}|$-vectors $\bm{s}^1,\bm{s}^2$. In particular, this reduces the search for MNE to a $2|\mathcal M|$-dimensional setting where numerical resolution becomes possible. Towards this goal, the main aim in this section is to constructively find such threshold equilibria.

\begin{remark}
	(\textit{Boundary Stages}) Recall that admissible strategies defined in \cref{def:admstrat} imply Player~1 (resp. Player~2)
	cannot make any switches at stage $\overline{m}$ (resp. at stage $\underline{m}$). In terms of threshold-type strategies this is equivalent to simply taking $s^1_{\overline{m}}=\overline{d}$ and $s^2_{\underline{m}}=\underline{d}$ which can be viewed as a constraint on possible admissible controls.
\end{remark}

Given a \textit{threshold-type} strategy $\ba^j \equiv \bm{s}^j$ of player~$j$, we expect the best-response strategy of player~$i$ to be consistently of threshold-type (see \cref{cor:BestRVT}). The Dynamic Programming Principle (DPP) implies that her corresponding value function, $\widetilde{V}^i( \cdot; \bm{s}^j)$ defined in \eqref{eq:ES_BestRinf}, solves a system of coupled stopping problems (see \cite{bayraktar2010one}). Namely letting $\underline{\tau}_m := \tau^1_{m}\wedge\tau^2_{m}$ (with $\tau^j_m$ pre-specified hitting times according to $\bm{s}^j$), we expect that
\begin{align}\notag
\widetilde{V}^i_m(x\,; \bm{s}^j) & = \sup_{\tau^i_m\in\mathcal{T}}\E_x\Bigg[\int_0^{\underline{\tau}_m}e^{-rt}\pi^i_m(X_t)dt+e^{-r\underline{\tau}_m}\mathds{1}_{\{\tau^1_m>\tau^2_m\}}
\Big(\widetilde{V}^i_{m-1}(X_{\tau^2_m}\,;\bm{s}^j)-\mathds{1}_{\{i=2\}}K^i_m(X_{\tau^2_m})\Big)\\ \notag
&\qquad\qquad\qquad\qquad+e^{-r\underline{\tau}_m}\mathds{1}_{\{\tau^1_m<\tau^2_m\}}\Big(\widetilde{V}^i_{m+1}(X_{\tau^1_m}\,; \bm{s}^j)-\mathds{1}_{\{i=1\}}K^i_m(X_{\tau^1_m})\Big)\\ \label{eq:ES_BestR}
&\qquad\qquad\qquad\qquad+e^{-r\underline{\tau}_m}\mathds{1}_{\{\tau^1_m=\tau^2_m\}}\Big(\widetilde{V}^i_{m+1}(X_{\tau^i_m}\,; \bm{s}^j)-\mathds{1}_{\{i=1\}}K^i_m(X_{\tau^i_m})\Big)\Bigg],
\end{align}
for $i\in\{1,2\}, j\neq i$, $\forall x\in\mathcal{D}$ and all $m\in\mathcal{M}$. Above $\mathcal{T}$ denotes all $\mathbb{F}$-stopping times, but the optimizer $\tilde{\tau}^i_m$ is expected to be associated to a threshold $\tilde{s}^i_m$. We use the shorthand notation $\E_x\left[\cdot\right]:=\E\left[\cdot|X_0=x\right]$, and the subscript in $\widetilde{V}^i_m$ to indicate the conditioning on $M_0=m$. Intuitively, at regime $m$ player $i$ implements a timing strategy to exercise her control at $\tau^i_m$, and a realization of these two stopping times yields a ``\textit{leader}'',  who acts first and switches $M_{\underline{\tau}_m}$. Note that in the case $\{\tau^1=\tau^2\}$ Player 1 has the priority to switch.

To approach the coupled system \eqref{eq:ES_BestR} we first consider the corresponding generic local constrained optimal stopping problem (which uncouples \eqref{eq:ES_BestR} by removing $\widetilde{V}^i_{m-1}, \widetilde{V}^i_{m+1}$ from the right-hand-side) and then the game equilibrium that is characterized as the best response to $\bm{s}^{j,\ast}$. See \cite{bayraktar2010one} for a related analysis of unconstrained optimal switching problems.

\subsection{Building Block}
To find the best-response of player $i$, we consider a \emph{local} optimal stopping problem of the form
\begin{align}\label{eq:MS_problem}
\widetilde{v}^i(x;\tau^j)=\sup_{\tau^i\in\mathcal{T}}\mathbb{E}_x\left[\mathds{1}_{\{\tau^i<\tau^j\}}e^{-r\tau^i}h^i(X_{\tau^i})+
\mathds{1}_{\{\tau^i>\tau^j\}}e^{-r\tau^j}l^i(X_{\tau^j})+\mathds{1}_{\{\tau^i=\tau^j\}}e^{-r\tau^i}g^i(X_{\tau^i})\right],
\end{align}
where $\tau^j$ is a given stopping time, $h^i(\cdot)$ is the \textit{leader} payoff from switching before $\tau^j$, and $l^i(\cdot)$ is the \textit{follower} payoff from switching after $\tau^j$. $g^i(\cdot)$ denotes the payoff of player $i$ when both players want to switch simultaneously. In our setting, $g^1 = h^1$, while $g^2=l^2$ due to the priority of Player~1.

In order to obtain threshold-type equilibrium, one would expect the optimizer to \eqref{eq:MS_problem}, $\tilde{\tau}^i$ to be of threshold-type, given $\tau^j$ is of threshold-type. However, as discussed at length in \cite{aid2017capacity} this is not always true. If player $j$ behaves aggressively, player $i$ would try to preempt right before, leading to lack of optimal $\tau^i$.

\begin{asp}\label{asp:MS_asp}
	\begin{enumerate}[label=(\roman*)]
		\item The exogenous stopping time $\tau^j$ is of \textit{threshold-type},
		\begin{align*}
		\tau^j:= \inf\{t\geq0\,:\,X_t\in\Gamma^j\}, j\in\{1,2\},\quad\text{
		with } \Gamma^1:=[s^1, \overline{d}) \text{ and } \Gamma^2:=(\underline{d}, s^2]. \end{align*}
		\item There are two function classes $\mathcal{H}_{\text{inc}}$ and $\mathcal{H}_{\text{dec}}$ defined in \cref{AP_Methd} such that $h^1\in\mathcal{H}_{\text{inc}}$ and $h^2\in\mathcal{H}_{\text{dec}}$.
		\item\label{asp:MS_noPreemp} player $i$ is not incentivized to preempt at $s^j$, i.e.~$h^i(s^j) < l^i(s^j)$.
	\end{enumerate}
\end{asp}

Under the above assumptions, it is known that the solution of \eqref{eq:MS_problem} is of threshold-type. Specifically, this can be established using the smallest concave majorant method, see e.g.~\cite{aid2017capacity, dayanik2003optimal, FerrariMoriarty15}. Let us remark that  \cref{asp:MS_asp} \ref{asp:MS_noPreemp} is essential for this result and would be hard to check in the sequel. Nevertheless, if the rest of  \cref{asp:MS_asp} is fulfilled, there exists \textit{uniquely} a \textit{preemptive} best-response, see \cite{aid2017capacity}.

\begin{pro}\label{pro:MS_OSP}
	Suppose that all conditions of \cref{asp:MS_asp} are satisfied. Let $F,G$ be the solutions to $(\mathcal{L}-r)u=0$, where $\mathcal{L}$ is the infinitesimal generator of $X$. Set
\begin{align}\label{eq:MS_Wronsk}
W(x_1, x_2) &:=F'(x_1)G(x_2)-F(x_1)G'(x_2)\\
\mathcal{W}(x_1, x_2)&:=F(x_1)G(x_2)-F(x_2)G(x_1).
\end{align}
Then the value function of \eqref{eq:MS_problem} is admitted as
	\begin{align*}
	\widetilde{v}^i(x;\tau^j) = \begin{cases}
	h^i(x), & \qquad\text{for } x\in\Gamma^i,\\
	l^i(x), & \qquad\text{for } x\in\Gamma^j,\\
	\widetilde{\omega}^iF(x)+\widetilde{\nu}^iG(x), & \qquad \text{for }x\in\mathcal{D}\setminus\left(\Gamma^i\cup\Gamma^j\right),
	\end{cases}
	\end{align*}
	where the optimal stopping region $\Gamma^i = \Gamma(\tilde{s}^i)$ is of \textit{threshold-type} and defined \textit{uniquely} through the threshold $\tilde{s}^{i}:=\tilde{s}^{i}(s^j)$ (with $\tilde{s}^1>s^2\text{ and }\tilde{s}^2<s^1$) that satisfies
	\begin{align}\label{eq:MS_Threshold}
	h^i(\tilde{s}^i)W(\tilde{s}^{i}, s^j)-l^i(s^j)W(\tilde{s}^{i}, \tilde{s}^{i})-(h^i)'(\tilde{s}^{i})\mathcal{W}(\tilde{s}^{i}, s^j) = 0.
	\end{align}
The coefficients $\widetilde{\omega}^i:=\widetilde{\omega}^i(\tilde{s}^i, s^j)$ and $\widetilde{\nu}^i:=\widetilde{\nu}^i(\tilde{s}^i, s^j)$ are defined as
	\begin{align}\label{eq:MS_Coefficient}
	\widetilde{\omega}^i = \frac{h^i(\tilde{s}^i)G(s^j)-l^i(s^j)G(\tilde{s}^i)}{\mathcal{W}(\tilde{s}^i, s^j)},\qquad
	\widetilde{\nu}^i = \frac{l^i(s^j)F(\tilde{s}^i)-h^i(\tilde{s}^i)F(s^j)}{\mathcal{W}(\tilde{s}^i, s^j)}.
	\end{align}
\end{pro}

\begin{remark}
	The above proposition subsumes the case where only one player is able to act. In this situation we may simply take
	$s^1=\overline{d}$ or $s^2=\underline{d}$, and player $i$ then effectively solves a standard optimal stopping problem as a special case of \eqref{eq:MS_problem}. See related discussion in \cite[Sec 4.1]{aid2017capacity}. These cases arise in the boundary stages $\underline{m}, \overline{m}$ associated to \eqref{eq:ES_BestR}.
\end{remark}

\subsection{Best-Response Verification Theorem}\label{sec:BVT}
By construction of MNEs in  \cref{defn:ne}, game payoffs and threshold-type strategies associated to an equilibrium necessarily solve the local optimizing  problems stated in \eqref{eq:ES_BestR}. Moreover, they necessarily are fixed-points to the following pair of optimizing problems at each regime $m$:
\begin{align}
\begin{cases}
V^1_m(x) =& \displaystyle\sup_{\tau^1_m\in\mathcal{T}}\E_x\Bigg[\int_0^{\underline{\tau}_m}e^{-rt}\pi^1_m(X_t)dt+e^{-r\underline{\tau}_m}\mathds{1}_{\{\tau^1_m>\tau^{2,\ast}_m\}}
\Big(V^1_{m-1}(X_{\tau^{2,\ast}_m})\Big)\\
&\qquad\qquad\qquad\qquad\qquad\qquad+e^{-r\underline{\tau}_m}\mathds{1}_{\{\tau^1_m\leq\tau^{2,\ast}_m\}}\Big(V^1_{m+1}(X_{\tau^1_m})-K^1_m(X_{\tau^1_m})\Big)\Bigg],\\
V^2_m(x) =& \displaystyle\sup_{\tau^2_m\in\mathcal{T}}\E_x\Bigg[\int_0^{\underline{\tau}_m}e^{-rt}\pi^2_m(X_t)dt+e^{-r\underline{\tau}_m}\mathds{1}_{\{\tau^{1,\ast}_m>\tau^2_m\}}
\Big(V^2_{m-1}(X_{\tau^2_m})-K^2_m(X_{\tau^2_m})\Big)\\
&\qquad\qquad\qquad\qquad\qquad\qquad+e^{-r\underline{\tau}_m}\mathds{1}_{\{\tau^{1,\ast}_m\leq\tau^2_m\}}\Big(V^2_{m+1}(X_{\tau^{1,\ast}_m})\Big)\Bigg],
\end{cases}
\end{align}
where $\tau^{1,\ast}_m,\tau^{2,\ast}_m$ are the stopping times associated to the thresholds $s^{1,\ast}_m,s^{2,\ast}_m$. Comparing to the generic problem in \eqref{eq:MS_problem} and subtracting $D^i_m(x)=\mathbb{E}_x\left[\int^\infty_0 e^{-rt}\pi^i_m(X_s)dt\right]$, we then wish to set
\begin{align} \label{eq:leader-follower}
\begin{cases}
h^1_m(x):=V^1_{m+1}(x)-D^1_m(x)-K^1_m(x),\\
l^1_m(x):=V^1_{m-1}(x)-D^1_m(x),\\
\end{cases}\qquad
\begin{cases}
h^2_m(x):=V^2_{m-1}(x)-D^2_m(x)-K^2_m(x),\\
l^2_m(x):=V^2_{m+1}(x)-D^2_m(x).
\end{cases}
\end{align}
Plugging above
into \eqref{eq:MS_Threshold} and \eqref{eq:MS_Coefficient} for all $m$ and combining, we obtain a coupled nonlinear system in $s^i_m, \omega^i_m, \nu^i_m$,
 whose solutions are expected to be a MNE of the switching game. We now propose a verification theorem which confirms that this is indeed the case. Our proof in \cref{AP_equi} follows the methods in \cite{aid2016nonzero} who considered  nonzero-sum games with impulse controls.
\begin{thm}[Verification Theorem]\label{thm:DA_equi}
Let $\Gamma^{1,\ast}_m:=[s^{1,\ast}_m, \overline{d}), \Gamma^{2,\ast}_m:=(\underline{d},s^{2,\ast}_m]$, $s^{1,\ast}_m>s^{2,\ast}_m$ and $\omega^1_m\geq0,  \omega^2_m\leq0, \nu^1_m\leq0, \nu^2_m\geq0$. Define
	\begin{subequations}\label{eq:DA_vFunc}
		\begin{align}
		V^1_m(x)&=\begin{cases}
		V^1_{m+1}(x)-K^1_m(x), & \text{ for } x\in\Gamma^{1,\ast}_m,\\
		V^1_{m-1}(x),& \text{ for } x\in\Gamma^{2,\ast}_m,\\
		D^1_m(x) + \omega^1_m F(x) + \nu^1_m G(x), & \text{ for } x\in\mathcal{D}\setminus\left(\Gamma^{1,\ast}_m\cup\Gamma^{2,\ast}_m\right),
		\end{cases}\label{eq:DA_vFunc1}\\
		V^2_m(x)&=\begin{cases}
		V^2_{m+1}(x), & \text{ for } x\in\Gamma^{1,\ast}_m,\\
		V^2_{m-1}(x)-K^2_m(x),& \text{ for } x\in\Gamma^{2,\ast}_m,\\
		D^2_m(x) + \omega^2_m F(x) + \nu^2_m G(x), & \text{ for } x\in\mathcal{D}\setminus\left(\Gamma^{1,\ast}_m\cup\Gamma^{2,\ast}_m\right).
		\end{cases}
		\end{align}
	\end{subequations}
Assume that (cf.~\cref{asp:MS_asp})
	\begin{itemize}
		\item [--] $D^1_{m+1}-D^1_m-K^1_m\in\mathcal{H}_{\text{inc}}$ \textit{for }$m<\overline{m}$, \textit{and} $D^2_{m-1}-D^2_m-K^2_m\in\mathcal{H}_{\text{dec}}$ \textit{for} $m>\underline{m}$;
		\item [--] $V^1_{m-1}(s^{2,\ast})\geq V^1_{m+1}(s^{2,\ast}_m)-K^1_m(s^{2,\ast}_m)$, \textit{for} $m>\underline{m}$, \textit{and} $V^2_{m+1}(s^{1,\ast})\geq V^2_{m-1}(s^{1,\ast}_m)-K^2_m(s^{1,\ast}_m)$,\textit{ for }$m<\overline{m}$;
		\item [--] \textit{thresholds} $s^{i,\ast}_m$ \textit{and coefficients} $\omega^i_m,\nu^i_m$, $i\in\{1,2\}$,  $m\in\mathcal{M}$ \textit{satisfy a system of non-linear equation stated in \eqref{eq:AP_System} - \eqref{eq:AP_sysbd}}.
	\end{itemize}
Then, $	\left(\bm{s}^{1, \ast},\,\bm{s}^{2,\ast}\right):=\,\left(\Gamma^{1,\ast}_m,\,\Gamma^{2,\ast}_m\right)_{ m\in\mathcal{M}}$ {\textit{ is a Markov Nash Equilibrium}}, and $V^i$'s in \eqref{eq:DA_vFunc} are the corresponding equilibrium payoffs.
\end{thm}

We slightly abuse the notation in \eqref{eq:DA_vFunc} as $V^1_{\overline{m}+1}$ and $V^2_{\underline{m}+1}$ do not exist. However, since in fact $s^{1,\ast}_{\overline{m}}=\overline{d}$ and $s^{2,\ast}_{\underline{m}}=\underline{d}$, so that $\Gamma^{1,\ast}_{\overline{m}}=\Gamma^{2,\ast}_{\underline{m}}=\emptyset$, the respective equations in \eqref{eq:AP_System} - \eqref{eq:AP_sysbd} are indeed well-defined.

The proof of \cref{thm:DA_equi}
can be repeated to  obtain an analogous verification theorem for the system of equations corresponding to the best-response value function $\widetilde{V}^i_m(x\,; \bm{s}^j)$ as defined in \eqref{eq:ES_BestR} for any threshold-type rival strategy $\bm{s}^j$:

\begin{cor}\label{cor:BestRVT}
Let $\bm{s}^2$ be the fixed switching thresholds of P2 and $\widetilde{V}^1_{\cdot}(\cdot;\bm{s}^2)$ be constructed as in \eqref{eq:DA_vFunc1}. Suppose that
		\begin{itemize}
		\item [---] $D^1_{m+1}-D^1_m-K^1_m\in\mathcal{H}_{\text{inc}}$ \textit{for }$m<\overline{m}$;
		\item [---] $\widetilde{V}^1_{m-1}(s^2_m;\bm{s}^2)\geq \widetilde{V}^1_{m+1}(s^2_m;\bm{s}^2)-K^1_m(s^{2}_m)$, \textit{for} $m>\underline{m}$;
		\item [---] $(\bm{\tilde{s}}^1, \bm{s}^2, \bm{\tilde{\omega}}^1, \bm{\tilde{\nu}}^1)$ is  a solution to \eqref{eq:AP_System1} \& \eqref{eq:AP_sysbd1}.
	\end{itemize}
Then $\tilde{\bm{s}}^1\equiv\tilde{\bm{s}}^1(\bm{s}^2)$ are the best-response thresholds, and $\widetilde{V}^1_m(x)$ are the corresponding best-response value function of P1. 
\end{cor}

\cref{thm:DA_equi} provides a direct approach to find a MNE of the switching game via solving the system of equations for the threshold vectors $\bm{s}^i$ and equilibrium payoffs defined through $\bm{\omega}^i$ and $\bm{\nu}^i$. Unfortunately, because this is a large system of equations (namely there are $6|\mathcal{M}-1|$ equations in total), the latter is non-trivial even numerically. In particular, most standard root-finding algorithms require a reasonable initial guess. In our experience, providing such a guess is not easy, so that the high-dimensional optimization algorithm frequently does not converge. Thus, in Section \ref{sec:SAM} we propose two approaches to obtain threshold vectors and game payoffs \emph{close} to those in equilibrium.

\begin{remark}\label{rem:general-switching}
	In a more general setting, players are allowed to act on $M_t$ in multiple ways. Thus, to each regime $m$ there is an associated \emph{action set} $C^i_m \subseteq \mathcal{M}$ that determines the potential new regimes that player $i$ can switch $M$ into. This generalizes the presentation above where $C^1_m = \{ m+1\}$ and $C^2_m = \{ m -1 \}$ are singletons. Analogously, one may consider other transition diagrams (e.g.~Player~1 acts by directly ``resetting'' to baseline regime $\overline{m}$, $C^1_m = \{ \overline{m} \}$ for all $m$). When $C^i_m$ has multiple elements, the corresponding player must choose \emph{how to switch}, not just \emph{when}. In the latter case we need to specify the respective switching costs, i.e.~to consider $K^i(m,m')$ which defines the cost of switching from $m$ to $m'$.
Such an extension can be handled by replacing the leader payoff in \eqref{eq:leader-follower} with $h^1_m(x) = \max_{ m' \in C^1_m } [V^1_{m'}(x) -D^1_m(x)- K^1(x,m,m') ]$ and the follower payoff with $\ell^1_m(x) = V^1_{m'}(x) - D^1_m(x),$ where $m' = \arg\max \{ m'' \in C^2_m : V^2_{m'}(s^2_m)-D^2_m - K^2(s^2_m,m,m')\}$.  The above $\max$-terms resemble the intervention operators in impulse control.
\end{remark}

\subsection{Equilibrium Macro Dynamics}\label{sec:SMG_summary}

The macro market evolution $M^\ast$ emerging in equilibrium is a time inhomogeneous \textit{non-Markovian} process with discrete state space $\mathcal{M}$. Thanks to the stationary nature of the threshold-type strategies, the behavior of $M^\ast$ is highly tractable and is the subject of this subsection.

Recall that in \eqref{eq:STS_constructM1} we define the sequence of regimes $M^\ast$ traverses, i.e. $\tM^\ast_n\equiv M^\ast_{\sigma_n}$. According to \eqref{eq:STS_constructM1}, $\tM^\ast_n$ has memory: the next transition of $\tM^\ast_n$ is affected by the last transition. For example, if $\tM^\ast_n = +1$ and the previous regime  was $\tM^\ast_{n-1} = +2$, this implies that the latest switch was due to Player 2, and hence we begin the sojourn in regime $+1$ at location $s^{2,\ast}_{+2}$,  i.e. $\tilde{X}^{(n)}_0 = X^x_{\sigma_n}=s^{2,\ast}_{+2}$, while if the previous state was $\tM^\ast_{n-1} = 0$ then it was Player 1 who switched last and we begin the sojourn at $s^{1,\ast}_0$, i.e. $\tilde{X}^{(n)}_0 = X^x_{\sigma_n}=s^{1,\ast}_{0}$.

To capture this 1-step memory we define the extended state space
\begin{align}\label{eq:state-E}
E := \{ \underline{m}^-, (\underline{m}+1)^-, (\underline{m}+1)^+, \cdots, m^-, m^+, \cdots, (\overline{m}-1)^-, (\overline{m}-1)^+, \overline{m}^+ \} \cup \{ \underline{m}^a,\,\overline{m}^a \},
\end{align}
where the superscript ``$+$'' corresponds to the previous transition being made by Player 1 (``up move in M'') and ``$-$'' corresponds to Player 2 making a ``down move in $M$''. We discuss the last two states $\underline{m}^a,\,\overline{m}^a$ below.

Instead of $M^{\ast}_t$ we now define its extended jump chain $\cM_n$ that takes values in $E$ and represents $(\tM^\ast_{n-1}, \tM^\ast_n)$.  Note that $\cM_0$ is undefined, as we need to know the previous transition to know the state of $\cM$.
Let us use \cref{fig:MS_Sample} to explain how $\cM$ behaves. The macro market starts at $X_0=0$ and $M^\ast_0=0$, while $\cM^\ast$ starts when $(X_t)$ hits $s^{2,\ast}_0$ with $\cM^\ast_1=(-1)^-$. The first sojourn begins at $s^{2,\ast}_0$ and ends when $(X_t)$ hits $s^{2,\ast}_{-1}$, leading us to $\cM^\ast_2=(-2)^-$, and so forth.

We proceed to compute the qualitative behavior of $M^\ast$ via $\cM_n$. In the case that $X$ is recurrent, the nature of threshold strategies implies that $M^\ast$ will also have recurrent dynamics. To quantify the dynamic macro equilibrium we then compute the long-run distribution of $M^\ast$ on $\mathcal{M}$. The latter is summarized via the transition probabilities of $\cM^\ast_n$ and the sojourn times $\xi_m$ of $\cM^\ast_n$.

In the case when $X$ is transient, $M^\ast$ should be transient too. Specifically, we should encounter the situation that $\tilde{\tau}^{1,n} \wedge \tilde{\tau}^{2,n} = +\infty$ (see \eqref{eq:STS_constructsig}), so that no more switches take place and $M^\ast$ remains constant forever or ``absorbed''. Under the assumption that $X$ is continuous and regular, this phenomenon can only occur at the boundary states of $\mathcal{M}$, whereby one player is a priori restricted from switching. This yields a one-sided switching region and hence the possibility of a scenario that $M^\ast_t\equiv \overline{m}$ (or $\underline{m}$), for all $t$ conditional on $M^\ast_0 = \overline{m}$, i.e.~that  $X$ never hits $s^2_{\overline{m}}$ starting at $s^1_{\overline{m}-1}$ (or $s^1_{\underline{m}}$ starting at $s^2_{\underline{m}+1}$). Note that given $X_t=x, M^\ast_t=\overline{m}$ (recall that the pair ($X_t, M^\ast_t$) is Markovian) one can not determine whether $M^\ast$ is absorbed or not. This is handled via \textit{taboo probabilities} \cite[Ch.~Taboo Probabilities]{chung1967markov} which are taken into account by adding the two ``absorbing'' states $\{\underline{m}^a,\,\overline{m}^a\}$ to $E$. Probabilistically, when switching up from $(\overline{m}-1)^\pm$, potential absorption can be captured by nature tossing a coin to decide whether the new state of $\cM$ is $\overline{m}^+$ or $\overline{m}^a$.

Returning to the case of recurrent $X$, let $\vec{\Pi}$ denote the invariant distribution of $\cM^\ast$, solved from $\vec{\Pi}\bm{P}=\vec{\Pi}$, where $\bm{P}$ is the transition probability matrix of $\cM^\ast$. Furthermore, let $\vec{\xi}$ be the vector of expected sojourn times at each state of $\cM$, defined as
\begin{align}\label{eq:SMS_xi}
\xi_{m^-}:=\mathbb{E}\big[\tilde{\tau}^{1,n}\wedge\tilde{\tau}^{2,n}\,|\,\cM^\ast_n=m^-\big],\qquad \xi_{m^+}:=\mathbb{E}\big[\tilde{\tau}^{1,n}\wedge\tilde{\tau}^{2,n}\,|\,\cM^\ast_n=m^+\big],
\end{align}
where the threshold hitting times $\tilde{\tau}^{i,n}$ are defined in \eqref{eq:STS_constructTau}. It follows that the long-run proportion of time that $M^\ast$ spends at regime $m$ (recall that $M^\ast_t=m$ is captured by $\cM^\ast_{\tn(t)} = m^\pm$) is given by:
\begin{align}\label{eq:SMG_longrun}
\rho_m =\frac{\Pi_{m^+}\xi_{m^+}+\Pi_{m^-}\xi_{m^-}}{\sum_{j \in \mathcal{M}} \{\Pi_{j^+} \xi_{j^+} + \Pi_{j^-} \xi_{j^-}  \} },\qquad\text{for all }m\in\mathcal{M}.
\end{align}

Now let us consider $X$ to be \textit{non-recurrent} so that one or both of the boundary regimes are absorbing, w.l.o.g  $\overline{m}^+$ for example.  In the long-run we then trivially have $\lim_{t\to \infty} M^\ast_t = \overline{m}$ and the quantities of interest in this situation are the expected number of controls exercised by player $i$ before $M^\ast$ gets absorbed, i.e.
\begin{align}\label{eq:SMS_numsw}
\mathbb{N}^i_m(x):=\lim_{T\rightarrow\infty}\mathbb{E}_x\bigg[\sum_k\mathds{1}_{\{\sigma^i_k\leq T\}}\,\big|\,M^\ast_0 = m\bigg],\quad i\in\{1,2\},
\end{align}
and the expected time until absorption,
\begin{align}\label{eq:SMS_abTime}
\mathbb{T}_m(x):= \mathbb{E}_x\Big[\min\big\{t\geq 0:\cM^\ast_{\tn(t)}\in\{\underline{m}^a, \overline{m}^a\}\big\}\,\big|\,M^\ast_0 = m\Big].
\end{align}
Analytic evaluation of these quantities is given in \cref{ap:DofM} which also provides expressions for the transition matrix $\bm{P}$ of $\cM^\ast$ and sojourn times $\vec{\xi}$. Computations specific to the OU \cref{eg:PF_MA} and the GBM \cref{eg:PF_LA} processes are also discussed.

\subsection{Stackelberg Switching}
We emphasize that the order of switches is never pre-determined and so the identity of the $n$-th switcher, $P_n$, is resolved endogenously based on game evolution and the realization of $(X_t)$. A variant of the switching game would be to pre-specify the identity of the player making the next switch, but not its timing, akin to a Stackelberg equilibrium where the leader and follower roles are fixed but timing strategy remains. The latter situation also arises organically if we restrict $M_t \in \{-1, +1\}$ which implies that players will alternate in their actions: $... \le \sigma^1_k \le \sigma^2_k \le \sigma^1_{k+1} \le \sigma^2_{k+1} \le ... $ Indeed, at any given stage only one firm can control $(M_t)$ so no consideration of simultaneous competition is needed (See \cite{boyer2012dynamic}).

It is instructive to consider a stationary threshold-type equilibrium in this setting, which reduces to characterizing the two thresholds $s^{1,\ast}_{-1}$ and $s^{2,\ast}_{+1}$. Furthermore, if their profit rates and switching costs depend on the local market environment $(X_t)$ symmetrically around 0 and $(X_t)$ is a process symmetric around 0 (like the OU process), we may search for a symmetric equilibrium with $s^{1,\ast}_{-1}=-s^{2,\ast}_{+1}=: \ms$ and $V^1_{.}(x)=V^2_{.}(-x)$ for any $x\in\mathcal{D}$. In turn this reduces finding the MNE to solving a \emph{single} nonlinear equation in $\ms$, providing some insight into the respective structure.

Examining \cref{thm:DA_equi} for Player~1, the system of equations is simplified to
\begin{align*}
V^1_{-1}(x)&=\begin{cases}
V^1_{+1}(x)-K^1_{-1}(x), & x\geq \ms,\\
D^1_{-1}(x)+\omega^1_{-1} F(x), & x<\ms,
\end{cases}\\
V^1_{+1}(x)&=\begin{cases}
D^1_{+1}(x)+\nu^1_{+1} G(x),&x>-\ms,\\
V^1_{-1}(x), &x\leq -\ms,
\end{cases}
\end{align*}
where $\ms,\omega^1_{-1},\nu^1_{+1}$ satisfy the following system (compare to \eqref{eq:AP_sysbd})
\begin{subequations}\label{eq:sys11}
	\begin{align}
	&\big(V^1_{+1}-D^1_{-1}-K^1_{-1}\big)(\ms)\cdot F'(\ms) -\big(V^1_{+1}-D^1_{-1}-K^1_{-1}\big)'(\ms)\cdot F(\ms) = 0,\label{eq:sys11a}\\
	&\big(V^1_{+1}-D^1_{-1}-K^1_{-1}\big)(\ms) - \omega^1_{-1} F(\ms) = 0,\label{eq:sys11b}\\
	&\big(V^1_{-1}-D^1_{+1}\big)(-\ms) - \nu^1_{+1} G(-\ms)=0.\label{eq:sys11c}
	\end{align}
\end{subequations}
Note that the last two equations specify $\omega^1_{-1},\nu^1_{+1}$ in terms of $V^1_{\pm 1}(\pm\ms)$.  One can now show that this system admits at least one solution.

\begin{cor}\label{cor:11case}
	Suppose that profit rates and switching costs are continuous and depend on the local market environment $(X_t)$ symmetrically about 0, and $(X_t)$ is a process symmetric around 0. Then there exists a threshold-type MNE for the switching game with $\mathcal{M}=\{-1, +1\}$.
\end{cor}

\begin{proof}
	See \cref{AP:prof11case}.
\end{proof}

\section{Sequential Approach to MNEs}\label{sec:SAM}
To approximate the system of nonlinear equations \eqref{eq:AP_System} - \eqref{eq:AP_sysbd} proposed in \cref{thm:DA_equi} we provide two sequential approaches. The first approach is through best-response iterations among threshold-type strategies, while the other inducts on equilibrium in finite-switch strategies.  The resulting threshold vectors $\bm{s}^i$  can be used as initial guesses in a root-finding algorithm.

\subsection{Constructing MNE by Best-response Iteration}\label{sec:BRI}
Given the rival's strategy, determining the best-response of one player is similar to a single-agent optimal  switching problem, which has been studied in \cite{bayraktar2010one,carmona2008optimal}. Let us assume that Player 2 implements a threshold-type strategy $\bm{s}^2$ as in  \cref{def:thstrat}. The best-response of Player 1 is then expected to be characterized through \eqref{eq:ES_BestR}, which is a system of \textit{coupled} optimal stopping problems.

We then decouple this system, in particular to apply
\cref{pro:MS_OSP} that provides the best-response threshold and game payoff of Player 1 once the leader/follower payoffs are fully specified. To do so, we consider auxiliary problems where the number of actions/switches available to Player 1 is bounded. Namely, Player 1 is constrained to ever use at most $N^1(\geq 1)$ controls. Her corresponding set of strategies is defined as
\begin{align}\label{eq:BST_finitestrat}
\mathcal{A}^{1, (N^1)}:= \left\{(\ba^1, \bm{s}^2)\in\mathcal{A}: \tau^1(n)=+\infty, n>\tn(1, N^1) \right\},
\end{align}
where $\tau^1(n)$ is the stopping rule Player 1 adopts at the $n$-th ``round'' of the game and $\tn(1,N^1)$ defined in \eqref{eq:PF_constructIk} denotes the round at which Player 1 exercises her $N^1$-th switch. Note that now the stopping sets are allowed to explicitly depend on the remaining number of controls left (equivalent to number of switches already used plus an initial constraint).  The best-response of Player 1 with $N^1$ controls is then admitted as
\begin{align}\label{eq:BTS_BestRn}
\widetilde{V}^{1,(N^1)}_m(x\,;\bm{s}^2):= \sup_{\ba^{1,(N^1)}\in\mathcal{A}^{1,(N^1)}}J^1_{m}(x; \ba^{1,(N^1)},\bm{s}^2),\qquad\forall x\in\mathcal{D},
\end{align}
for all $m\in\mathcal{M}$. When Player 1 has zero controls $N^1=0$, her payoff at any stage $m$ is fully determined by $\bm{s}^2$, for instance at regime $\underline{m}+1$
\begin{align}\label{eq:BRI_zeroctr}
\widetilde{V}^{1, (0)}_{\underline{m}+1}(x\,;\bm{s}^2)=\mathbb{E}_x\left[\int^{\tau^2_{\underline{m}+1}}_0e^{-rt}\pi^1_{\underline{m}+1}\left(X_t\right)dt\right]
+\mathbb{E}_x\left[e^{-r\tau^2_{\underline{m}+1}}\right]\cdot D^1_{\underline{m}}(s^2_{\underline{m}+1}),
\end{align}
where the last term is the NPV of fixed-market-state cashflows defined in \eqref{eq:MS_dcf}.

\begin{pro}\label{pro:BTS_conv}
	Given a threshold-type strategy $\bm{s}^j$ of player $j$,  the best-response game payoffs of player $i$ with finite controls  converge as $N^i \to \infty$, i.e.~$\forall x\in\mathcal{D}$,
	$$\widetilde{V}^{i, (N^i)}_m(x\,;\bm{s}^j)\nearrow\widetilde{V}^i_m(x\,;\bm{s}^j),\quad\text{for all } m\in\mathcal{M}\qquad \text{ as }N^i\nearrow\infty.$$
\end{pro}

Proof of  \cref{pro:BTS_conv} is inspired by \cite{bayraktar2010one} and stated in
\cref{AP_BestRgame payoff}. Moreover, strong Markov property of $X$ and Dynamic Programming Principle (DPP) imply that
\begin{equation}\label{eq:ES_BestRDp}
\begin{split}
\widetilde{V}^{1,(N^1)}_m(x\,;\bm{s}^2) & = \sup_{\tau^{1}(1)\in \mathcal{T}}\E_x\Bigg[\int_0^{\underline{\tau}_m}e^{-rt}\pi^1_m(X_t)dt+e^{-r\underline{\tau}_m}\mathds{1}_{\{\tau^{1}(1)>\tau^2_m\}}\cdot\widetilde{V}^{1,(N^1)}_{m-1}(X_{\tau^2_m}\,;\bm{s}^2) \\
&\qquad\qquad+e^{-r\underline{\tau}_m}\mathds{1}_{\{\tau^{1}(1)\leq\tau^2_m\}}\Big(\widetilde{V}^{1,(N^1-1)}_{m+1}(X_{\tau^{1}(1)}\,;\bm{s}^2) -K^1_m(X_{\tau^{1}(1)})\Big)\Bigg],
\end{split}
\end{equation}
for all $m\in\mathcal{M}$, $\forall x\in\mathcal{D}$, with $\tau^2_m$ the first hitting time of $\Gamma^2_m=(\underline{d}, s^2_m]$, and dependence of $\tau^1(1)$ on $N^1$ omitted for brevity. We refer to \cite{bayraktar2010one, carmona2008optimal} who proved that DPP holds in this problem and \cite{aid2017capacity} who analyzed finite-control stopping games.

Notice that game payoffs \eqref{eq:BRI_zeroctr} can be treated as starting points to implement a backward Dynamic Programming scheme to solve the finite-control optimal stopping problem introduced in \eqref{eq:ES_BestRDp}. Suppose that $\widetilde{V}^{1, (N^1)}_{m-1}(\cdot\,;\bm{s}^j)$ and $\widetilde{V}^{1, (N^1-1)}_{m+1}(\cdot\,;\bm{s}^j)$ are determined, and \cref{asp:MS_asp} holds. We denote
\begin{align*}
\widetilde{v}^{1,N^1}(x\,;\tau^2_m) &:= \widetilde{V}^{1, (N^1)}_{m}(x\,;\bm{s}^j) - D^1_m(x),\\
h^{1,N^1}(x) &:=  \widetilde{V}^{1, (N^1-1)}_{m+1}(x\,;\bm{s}^j) - D^1_m(x)-K^1_m(x),\\
l^{1,N^1}(x)& :=  \widetilde{V}^{1, (N^1)}_{m-1}(x\,;\bm{s}^j) - D^1_m(x)，
\end{align*}
and  apply \cref{pro:MS_OSP} with leader/follower payoffs $h^{1,N^1},l^{1,N^1}$ to obtain best-response game payoff $\widetilde{V}^{1, (N^1)}_{m}(x\,;\bm{s}^2)$, which is parameterized by $\widetilde{\omega}^{1, (N^1)}_m$,$\widetilde{\nu}^{1,(N^1)}_m,\tilde{s}^{1, (N^1)}_m$. Thanks to  \cref{pro:BTS_conv} we know $\widetilde{V}^{1, (N^1)}_{m}(x\,;\bm{s}^2)$ converges, thus expect $\tilde{s}^{1, (N^1)}_m \to \tilde{s}^1_m$ would converge as well as $N^1 \to \infty$. Thus, for $N^1$ large, we may use $\bm{\tilde{s}}^{1,(N^1)}$ to define a time-stationary strategy that is a proxy for the best response.

\begin{figure}[b!]
	\centering
	\subfigure[Approximations of Game Values]{
		\label{fig:SA_ConvGV}
		\includegraphics[width=0.45\textwidth]{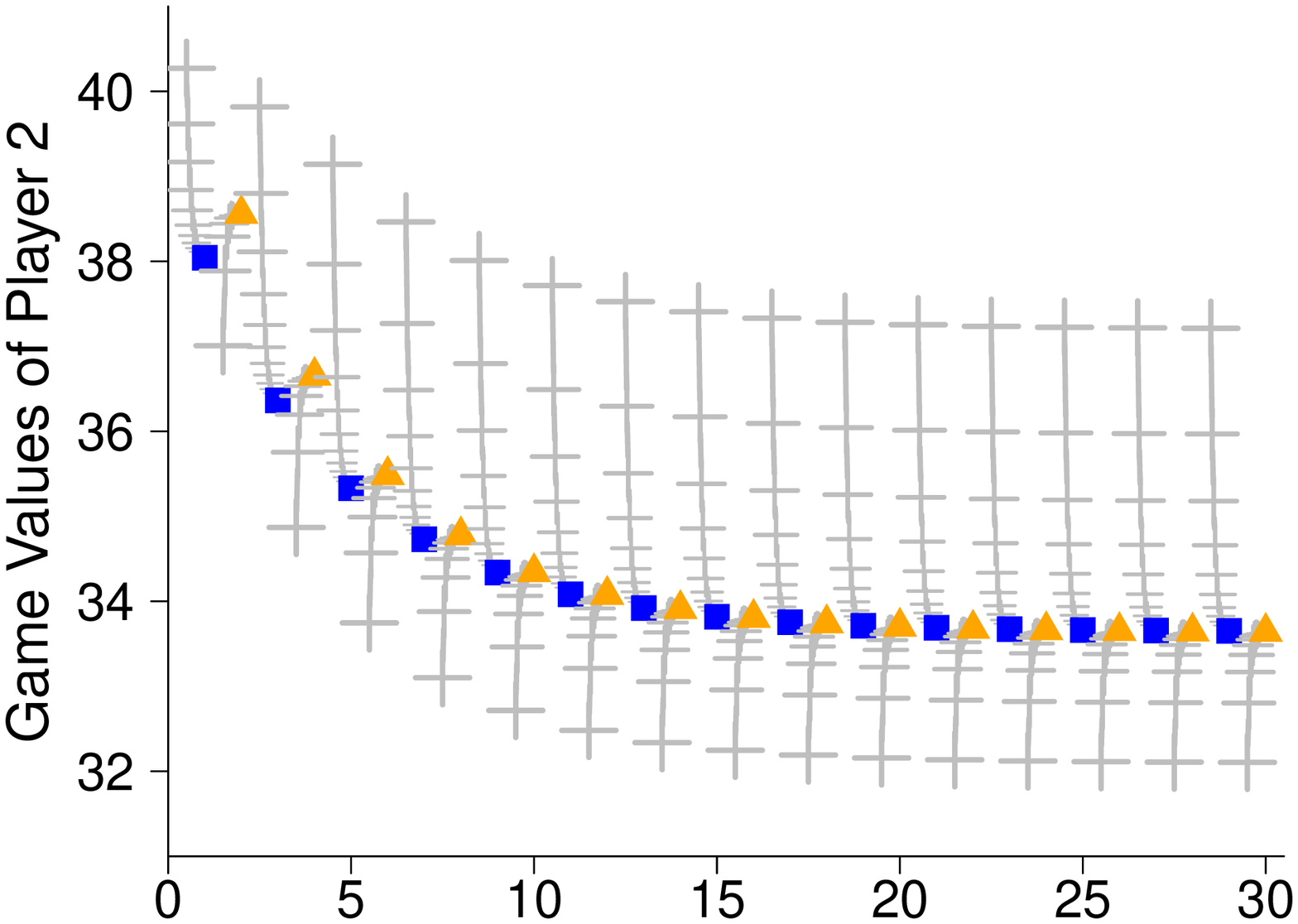}}
	\subfigure[Approximations of Thresholds]{
		\label{fig:SA_ConvTh}
		\includegraphics[width=0.45\textwidth]{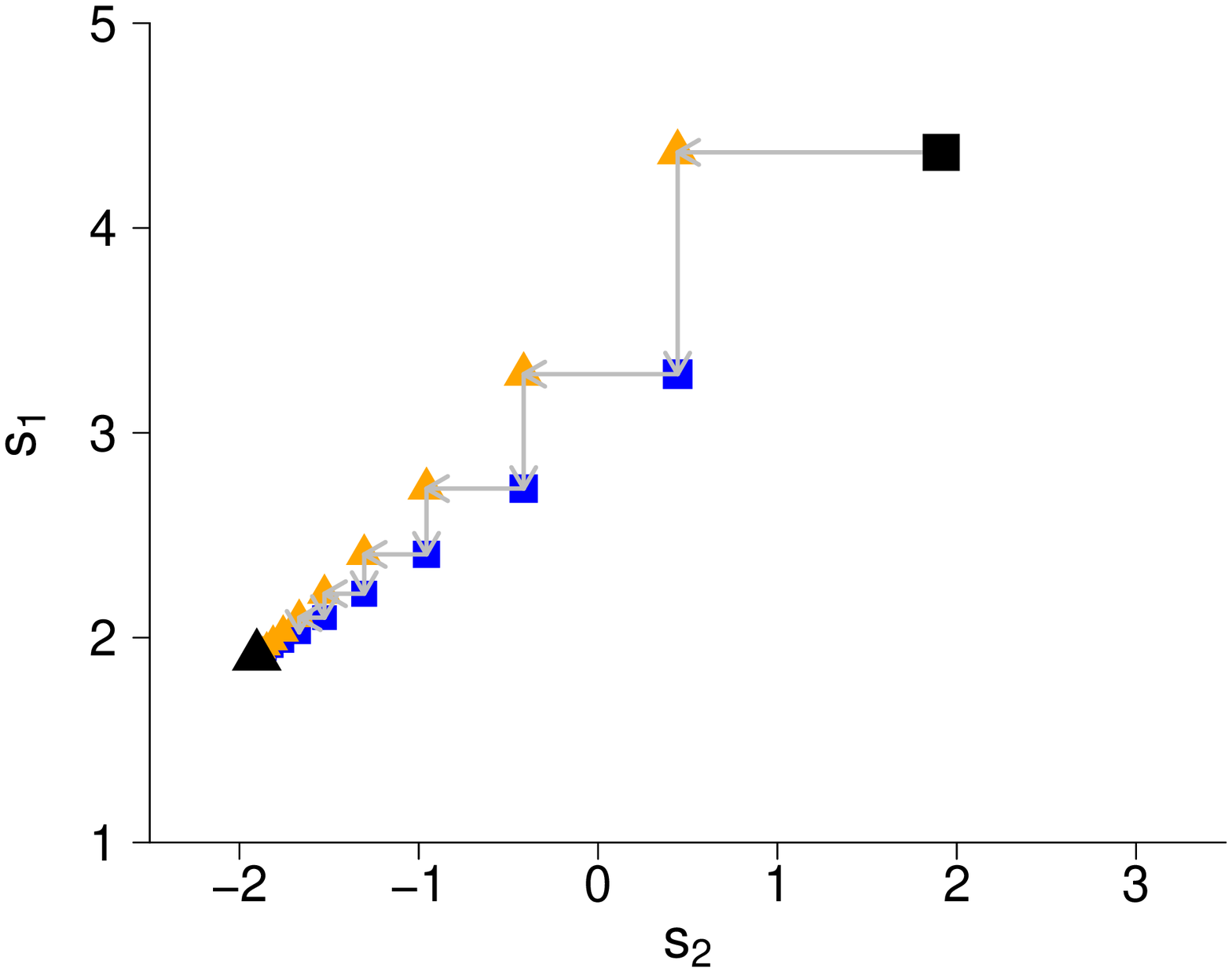}}
	\caption{\label{fig:SA_Cov} Finding fixed point of best-response maps via the tatonnement process over $a$ with $\mathcal{M}=\{-1, 0,+1\}$. Squares represent rounds $a=1,3,\ldots$ where player's 2 strategy is fixed. Triangles represent even rounds $a=2,4,\ldots$ where player's 1 strategy is fixed. (Left): Game values of Player 2 with $M_0=0$ and $X_0=0$ indexed according to $\tilde{V}^{2,N_2}(x; s^{1,a})$ with $N_2 = 1,2,\ldots, 30$. (Right): Thresholds $s^{i,a}_0$ as a function of $a$ at $m=0$ and $N=30$. The enlarged square represents the first round $a=1$ and the enlarged triangle represents the last $a=30$ round which appears to be close to a fixed point.}
\end{figure}
Building upon the preceding convergence result, we propose the following algorithm to determine a threshold-type
Markov Nash equilibrium. Essentially, we apply the t\^atonnement approach, alternating in finding the best-response strategies of the two players, expecting to converge to an associated best-response fixed point. These alternating best-responses are indexed by ``rounds'' $a=1,2,\ldots, A$. At odd rounds, Player 1 solves for her best response $(i=1, j=2)$; at even rounds, Player 2 solves for her best response $(i=2, j=1)$:
\begin{itemize}
	\item [\ding{172}:] Set the strategy of player $j$ to be of threshold-type as $\bm{s}^{j,a}$:
	\begin{itemize}
		\item For $a=1$, set $\bm{s}^{2,1}$ as the monopoly thresholds of P2, i.e.~when P1 is not allowed to switch ($N^1 =0$ case). The thresholds $\bm{s}^{2,1}$ can then be obtained by solving a single-agent optimal switching problem.
		\item For $a>1$ set $\bm{s}^{j,a}=\widetilde{\bm{s}}^{j, a-1}$.
	\end{itemize}
	
	\item [\ding{173}:] Solve for $\widetilde{\bm{s}}^{i,a}$ and value function $\widetilde{V}^{i, (N)}_m(\cdot\,;\bm{s}^{j,a})$ for all $m\in\mathcal{M}$:
	\begin{itemize}
		\item Solve optimal stopping problems when player $i$ is allowed at most $n$ switches and player $j$ applies $\bm{s}^{j,a}$ using \cref{pro:MS_OSP}, iteratively for $n=1,\ldots, N$.
		\item Record $\widetilde{V}^{i,(N)}_m(\cdot)$ and the approximate best-response strategy $\widetilde{\bm{s}}^{i,a}(\bm{s}^{j,a-1}) \simeq \tilde{\bm{s}}^{i,(N)}$
	\end{itemize}
	\item [\ding{174}:] Change the roles of $i$ and $j$ (alternate which player is solving for the best response)
	\item [\ding{175}:] Repeat steps \ding{172} - \ding{174} as $a=1,\ldots,$ until the maximum change in $|\widetilde{\bm{s}}^{i,(N),a}-\widetilde{\bm{s}}^{i,(N),a-2}|$, $i\in\{1,2\}$ are both less than a predetermined tolerance level $Tol$ (or simply for $A$ rounds).
\end{itemize}

\cref{fig:SA_Cov} illustrates the above best-response induction in one of our case-studies. In each round we iterate to find the best response assuming player $i$ has up to $N^i$ switches. During the odd rounds $a=1,3,\ldots$ Player 2 implements the stationary strategy $\bm{s}^{2,a}$ and her game values (gray `+') decrease  as the number of Player 1's controls $N^1=1,\ldots,30$ increases. In contrast, during the even iterations, Player 2 game values  $\widetilde{V}^{2,(N^2)}_m$ converge upwards as $N^2=1,\ldots, 30$. The corresponding thresholds $s^{i,a}_m$ are shown on the right panel. We observe that both game values and thresholds converge after 30 inner iterations over $N^i$, and over $A=30$ outer tatonnement rounds (a total of $30 \times 30 \times 2$ optimal stopping problems solved via  \cref{pro:MS_OSP}). In particular, we may take $s^{i,(N),A}_m$ as an approximation of a best response fixed-point and hence of the equilibrium $s^{i,\ast}_m$.

\subsection{Constructing MNEs by Equilibrium Induction}\label{sec:scheme-localEqm}
Another approach to construct an (approximate) threshold-type MNE of the switching game is to take limits in a finite-control game of timing. This links to the earlier work by the authors in \cite{aid2017capacity}. Suppose that both players are constrained to \textit{finite} control strategies with respective bounds $n^1, n^2$ on total allowed number of switches. Specifically we consider strategies of the form
\begin{align}
\ba^{i, (n^1,n^2)}:=\left(\Gamma^{i,(k^1,k^2)}_m\right)^{k^1\leq n^1, k^2\leq n^2}_{m\in\mathcal{M}},\qquad\text{with }\Gamma^{i,(0, k^j)}_m\equiv \emptyset,
\end{align}
where $k^i \le n^i$ denotes the number of controls remaining for player $i$, and index stages of this game as
$$(M_t, N^1_t, N^2_t):=\{\text{macro market regime, \# controls remaining for P1, \# controls remaining for P2}\},$$
with $M_t\in\mathcal{M}$, and $N^i_t$ is a non-increasing piecewise-constant process on $\mathbb{N}$ with $N^i_0=k^i$ for $i\in\{1,2\}$. Duopoly games of this type are studied by \cite{aid2017capacity}, who determine local equilibria at each game stage by backward dynamic programming and patch them to construct a global one.

At sub-stage $(m, k^1,k^2)$, the local equilibrium is characterized as a fixed point of these players' best-response based on \cref{pro:MS_OSP}. Taking Player~1 as an example again, her leader and follower payoffs are related to her equilibrium game payoffs at adjacent stages which are known when implementing backward dynamic programming:
\begin{align}\label{eq:hl-inductive} \begin{cases}
h^{1,(k^1,k^2)}_m(x) &:=  V^{1, (k^1-1,k^2)}_{m+1}(x)- D^1_m(x)-K^1_m(x),\\
l^{1,(k^1,k^2)}_m(x)& :=  V^{1,(k^1,k^2-1)}_{m-1}(x) - D^1_m(x), \end{cases}
\end{align}
and their equilibrium strategies $(\tau^{1,(k^1,k^2),\ast}_m,\tau^{2,(k^1,k^2),\ast}_m)$ and game payoffs solve a pair of optimal stopping problems:
\begin{align}\label{eq:EI_localequi}
\begin{cases}
V^{1,(k^1,k^2)}_m(x)-D^1_m(x)=&\displaystyle\sup_{\tau^{1,(k^1,k^2)}_m\in\mathcal{T}}\mathbb{E}_x\bigg[\mathds{1}_{\{\tau^{1,(k^1,k^2)}_m<\tau^{2,(k^1,k^2),\ast}_m\}}\cdot h^{1,(k^1,k^2)}_m(X_{\tau^{1,(k^1,k^2)}_m})\\
&\qquad\qquad\qquad\qquad+\mathds{1}_{\{\tau^{1,(k^1,k^2)}_m>\tau^{2,(k^1,k^2),\ast}_m\}}\cdot l^{1,(k^1,k^2)}_m(X_{\tau^{2,(k^1,k^2),\ast}_m})\bigg],\\
V^{2,(k^1,k^2)}_m(x)-D^2_m(x)=&\displaystyle\sup_{\tau^{2,(k^1,k^2)}_m\in\mathcal{T}}\mathbb{E}_x\bigg[\mathds{1}_{\{\tau^{1,(k^1,k^2),\ast}_m<\tau^{2,(k^1,k^2)}_m\}}\cdot l^{2,(k^1,k^2)}_m(X_{\tau^{1,(k^1,k^2),\ast}_m})\\
&\qquad\qquad\qquad\qquad+\mathds{1}_{\{\tau^{1,(k^1,k^2),\ast}_m>\tau^{2,(k^1,k^2)}_m\}}\cdot h^{2,(k^1,k^2)}_m(X_{\tau^{2,(k^1,k^2)}_m})\bigg].
\end{cases}
\end{align}
Note that simultaneous switches can be ruled out since on the event $\{\tau^{i,(k^1,k^2)}_m=\tau^{j,(k^1,k^2),\ast}_m\}$, stopping by Player 1 is strictly dominated by the strategy of first waiting, and then optimally switching as follower. In \cite{aid2017capacity} we show that the local equilibrium exists under some regularity conditions on $D^i$'s and $K^i$'s, however uniqueness cannot be guaranteed. Moreover, such a local equilibrium is not always of threshold-type, as  \textit{preemptive} equilibria may emerge.

\begin{figure}[t!]
	\centering
	\includegraphics[width=0.8\textwidth]{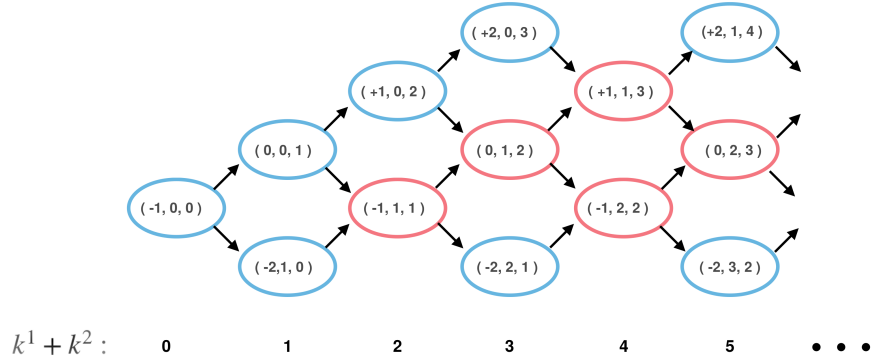}
	\caption{\label{fig:SA_AG1}A schematic diagram illustrating induction on local timing equilibria of Section~\ref{sec:scheme-localEqm}, starting at $(-1,0,0)$ and with $\underline{m}=-2$ and $\overline{m}=+2$, which leads to $\Delta_{-2}=-1,\Delta_{-1}=0, \dots,\Delta_{+2}=3$ in \eqref{eq:DMNE_convEqui}. The diagram illustrates the reachable stages $(m,k^1,k^2)$ relative to $(M_0, 0,0)$ and using the ``forward'' dynamic programming scheme. Blue circles denote single-agent optimization sub-stages that correspond to optimal stopping problems, while red circles denote interior stages where local timing equilibrium is determined according to \eqref{eq:EI_localequi}. Boundary stages are those where $k^1 = 0$ or $k^2 = 0$ or $m \in \{\underline{m}, \overline{m}\}$. Stages not reachable from $(-1,0,0)$ are omitted. }
\end{figure}

\begin{figure}[b!]
	\centering
	\subfigure[Finite-Control Equilibrium Payoffs $V^{i,(n,n+1)}_0(X_0)$]{
		\label{fig:MA_GPEquiInduc}
		\includegraphics[width=0.45\textwidth]{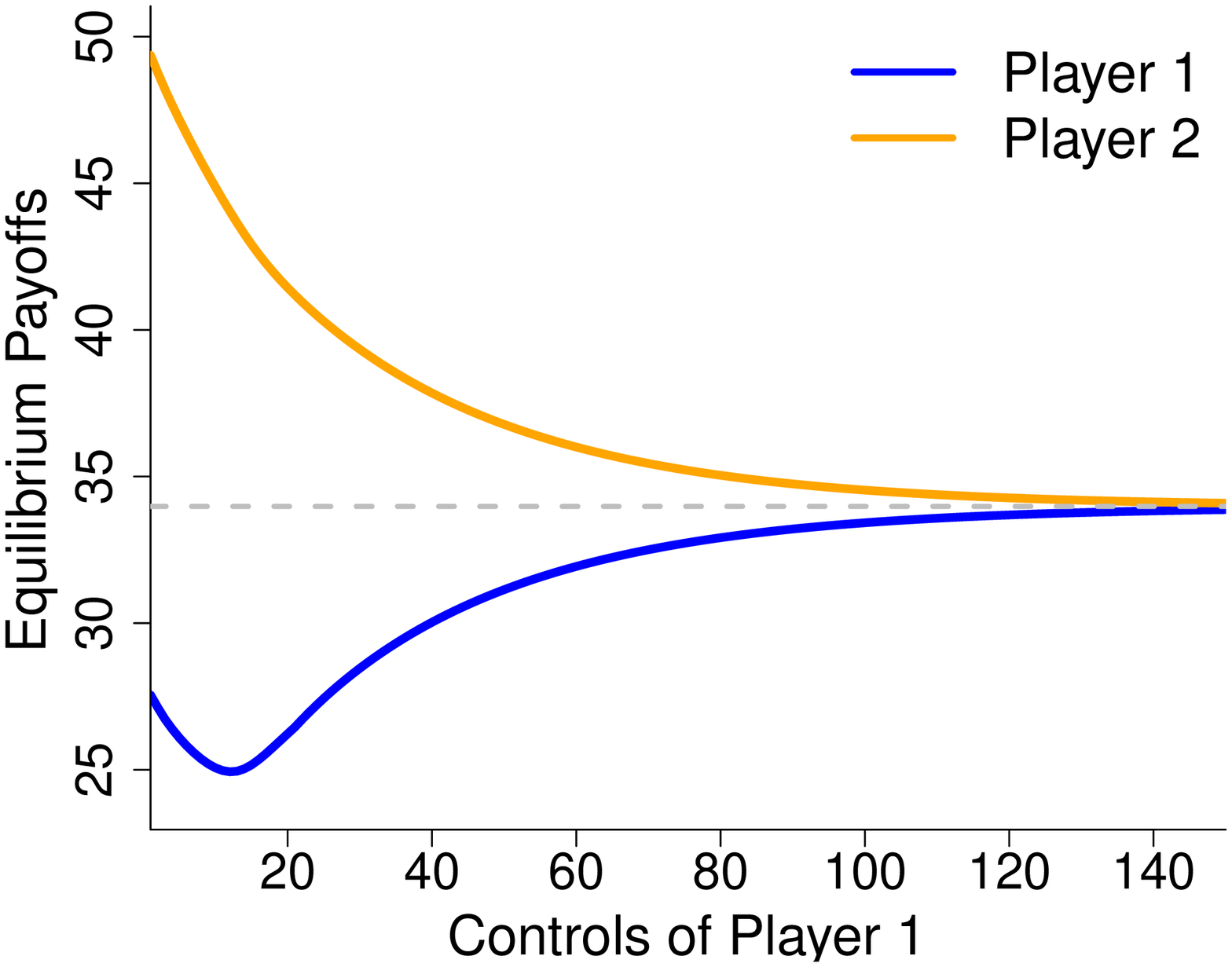}}
	\subfigure[Finite-Control Distribution of $M^\ast$]{
		\label{fig:MA_f2ad_distrib_est}
		\includegraphics[width=0.45\textwidth]{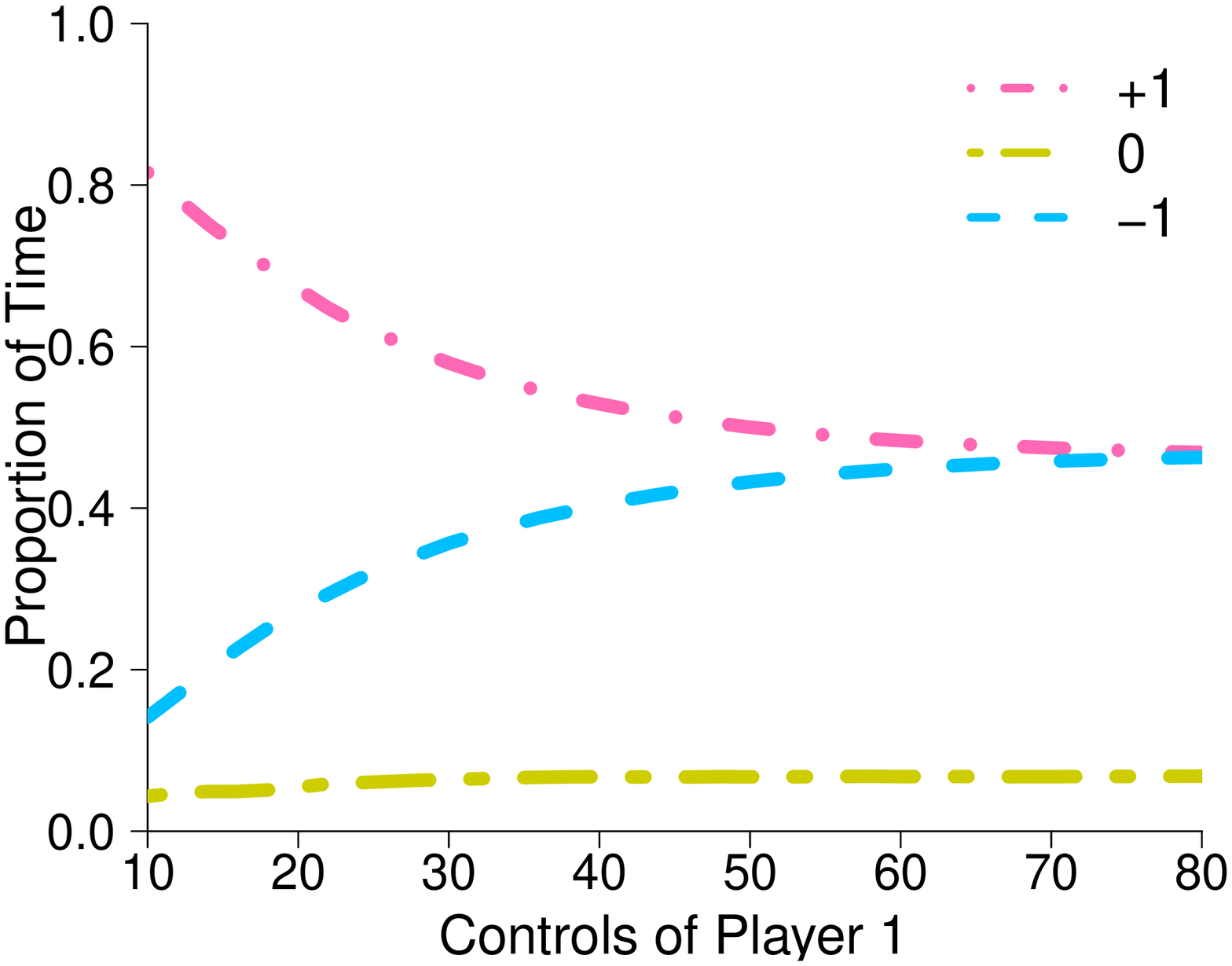}}
	\caption{\label{fig:MA_EquiInduc} Equilibrium induction using the forward scheme in \cref{fig:SA_AG1}. We postulate that Player 2 has one extra control, $N^2_0 = N^1_0 + 1 \Leftrightarrow \Delta_0 = -1$. \emph{Left}: Equilibrium payoffs $V^{i, (N^1_0,N^2_0)}_{M_0}(X_0)$ with $X_0=0$, $M_0=0$ indexed by $N^1_0$. The dashed line denotes the limiting payoff $V^i_0(X_0)$ in the original infinite-control game. \emph{Right:} time-averaged distribution of $M^{\ast,(N^1_0, N^2_0)}_t$ on $[0,\bar{T})$, $\bar{T}=50$. Specifically, we show the values of $\rho_m(\bar{T}):=\E \left[ \frac{1}{\bar{T}}\int^{\bar{T}}_0\mathds{1}_{\{M^\ast_s=m\}}ds \big|\, M^\ast_0 = 0 \right]$ for $m\in \{-1,0,1\}$ as $N^1_0$ increases.}
\end{figure}

In the example sketched in Figures~\ref{fig:SA_AG1}-\ref{fig:MA_f2ad_distrib_est}, we implement a \textit{forward} scheme to generate a sequence of equilibria starting at sub-stage $(m,k^1, k^2) = (-1, 0,0)$ where the payoffs are $V^{i, (0, 0)}_{-1}(x) = D^i_{-1}(x)$. With this known, we can solve for the local equilibria at stages $(0, 0, 1)$ and $(-2, 1, 0)$ utilizing \eqref{eq:hl-inductive}. Iterating, we find local equilibria for all triplets $(m,k^1,k^2)$ shown in the Figure (Throughout, we make the ansatz that local equilibria  are all of threshold-type at any sub-stage $(m, k^1, k^2)$). These triplets can be characterized as $k^2 = k^1 + \Delta_m$, where  the auxiliary parameter $\Delta_m$ is the difference between the number of switches available to the players at regime $m$. For  instance $\Delta_{-1}=0$ in \cref{fig:SA_AG1}, so that the players are  equally endowed whenever they are at regime $M_t=-1$, cf.~the sub-stages  $(-1,1,1), (-1,2,2),\ldots.$ The sub-stages $(m,k^1,k^2)$ that are not reachable from $(-1,0,0)$ are omitted and in this instance  we need not consider the respective local equilibria.

Using the terminal game stage $(-1, 0,0)$ and continuing up to $k^1 \le N$, the above forward scheme iteratively yields a sequence of equilibrium thresholds $s^{i,(n, n+\Delta_m)}_m$ and game coefficients $(\omega^{i, (n,n+\Delta_m)}_m, \nu^{i,(n,n+\Delta_m)})$. The resulting game payoffs are shown in \cref{fig:MA_GPEquiInduc}. As mentioned, the parameter $\Delta_m$ influences all the equilibria in \cref{fig:SA_AG1}. For example, in the presented scheme, the game will eventually end with $M_t = -1$ for $t$ large enough. Nevertheless, as $N$ increases, we expect that this effect vanishes, so that the limits are independent of $\Delta_m$:
\begin{align}\label{eq:DMNE_convEqui}
\begin{pmatrix}
s^{i,(n, n+\Delta_m)}_m\\
\omega^{i, (n,n+\Delta_m)}_m\\
\nu^{i,(n,n+\Delta_m)}
\end{pmatrix}
\,\xrightarrow{\text{ as }n\nearrow\infty\,}\,
\begin{pmatrix}
s^{i,\ast}_m\\
\omega^{i,\ast}_m\\
\nu^{i,\ast}_m
\end{pmatrix},\qquad i\in\{1, 2\}, m\in\mathcal{M}.
\end{align}
This convergence can be observed in~\cref{fig:MA_EquiInduc} where the underlying symmetries imply $V^{1, (n,n+1)}_0(0) = V^{2, (n+1,n)}_0(0)$. Thus, we may interpret the top curve in \cref{fig:MA_GPEquiInduc} as the game payoff in the finite-stage setup when the player has \emph{one more} switch than her rival, and the bottom curve
as her game payoff when she has \emph{one fewer} switch relative to the rival. As $n \to \infty$, the relative benefit vanishes and both
$V^{i, (n,n\pm 1)}_0(x)$ approach $V^{i}_0(x)$. { Similarly, in~\cref{fig:MA_f2ad_distrib_est} the finite-horizon distribution of $M^{\ast,(n,n+1)}$  converges to that of $M^\ast$ on $[0,\bar{T})$, $\bar{T}=50$.}

Two issues arise with the above scheme. First, the associated equilibrium payoffs $V^{i,(N^1_0, N^2_0)}$  are not monotone in terms of $N^1_0$ or $N^2_0$. For instance, higher $N^1_0$ benefits P1, while higher $N^2_0$ harms her since her rival now has more flexibility. Changing both $N^i$'s simultaneously leads to ambiguous results: in \cref{fig:MA_GPEquiInduc} P1's payoff decreases first, then increases in terms of $N^1_0=N^2_0-1$. Thus convergence in \eqref{eq:DMNE_convEqui} is hard to prove. Second, the local timing game might generate multiple threshold-type equilibria~\cite{aid2017capacity}. As a result, equilibrium selection becomes important when inducting on $N^i_0$'s.

\begin{remark}
	Setting  $\Delta_m=\underline{m}$ (resp.~$\Delta_m=\overline{m}$) is equivalent to granting P2 (resp.~P1) infinite number of allowed switches, while her rival is restricted to finite number of controls. This of course confers an ultimate advantage to the privileged player who will ultimately ``win out'' the competition. For example, taking $(M_0, N^1_0, N^2_0) = (0,4,2)$ in the running example means that P2 only has 2 switches, while P1 has four, so she will ultimately succeed in driving to the best possible regime $\lim_{t\to \infty} M_t = +2$ and will never require more than 4 switches anyway (recall that $\overline{m}=+2$). Thus this setting resembles the auxiliary game discussed in Section~\ref{sec:BRI}, except that both players are now dynamically optimizing their thresholds.
\end{remark}

\section{Case Study: Mean-reverting Market Advantage}\label{sec:MA}
Continuing \cref{eg:PF_MA}, we describe the local market fluctuation $(X_t)$ by an Ornstein-Uhlenbeck (OU) process mean-reverting to $\theta=0$:
\begin{align}\label{eq:OU}
dX_t=-\mu X_tdt+\sigma dW_t,
\end{align}
with $\mu=0.15, \sigma=1.5,\,\mathcal{D}=\mathbb{R}$ (i.e.~natural boundaries $\underline{d}=-\infty$ and $\overline{d}=+\infty$). This implies that the stationary distribution of $X$ is Gaussian, $\mathcal{N}(0, 7.5)$. For the discounting rate we take $r=10\%$. The profit rates  $\pi^i_m$ are constant and listed in \cref{tb:NE_pi}. Note that $\pi^i$'s are monotone but \textit{concave}  in terms of the regime $m$. A motivating economic context is the advertising competition between two firms. They can make an advertising campaign  by paying $K^i$, with the cost dependent on the exchange rate $(X_t)$. The effect of advertising (i.e.~exercising a change in $M$)
is to enhance one firm's dominance in the market, bringing her higher profit rates. Due to diminishing returns to scale, improvement in the profit rate decreases as the firm captures more and more market share, so that $\pi^i_m$ is concave in $m$.

In \cref{tb:NE_pi} our intent is that the profit ladder extends to the right and to the left forever, however the above concavity makes it uneconomical to reach extreme levels of dominance. Thus, we progressively enlarge the number of market regimes considered: $M_t \in  \{-1,0,1\}$ (Case I), $ M_t \in \{-2,-1,0,1,2\}$ (Case II), and $M_t \in \{-3, \ldots, 3\}$ (Case III). Below we also consider an asymmetric situation with $M_t \in \{-1,0,1,2\}$.

The switching costs $K^i$'s are affected by $X_t$: \begin{align}
K^i_m(x):= c_i\cdot(1+e^{(-1)^i\beta_ix}), \qquad i\in\{1,2\},
\end{align}
where $c_i=0.5,\,\beta_i=0.5$, $i\in\{1,2\}$. Thus, Player 1 can make cheap switches to dominate the market when $x \gg 0$ and Player 2 has the advantage when $x \ll 0$. For simplification, $K^i$'s are not directly affected by $(M_t)$. By construction, the profit rates and all other parameters are symmetric (about zero), so in equilibrium we expect players to act symmetrically as $X$ fluctuates from positive to negative and vice versa.

\begin{table}[htb]
	\centering
	\begin{tabular}{c|ccccccc}
		\hline
		\hline
		$m$	&  $-3$ & $-2$ & $-1$ & $0$ & $+1$ & $+2$ & $+3$   \\
		\hline
		\hline
		$\pi^1_m$	 & 0.0 & 1.5 & 2.8 & 4.0 & 5.1 & 5.9 & 6.0 \\
		\hline
		$\pi^2_m$	&  6.0 & 5.9 & 5.1 & 4.0 & 2.8 & 1.5 & 0.0 \\
	\end{tabular}
	\caption{\label{tb:NE_pi}Profit rate ladders $\pi^i_m$ for Section~\ref{sec:MA}.}
\end{table}

Best-response induction associated to this case study with $\mathcal{M} = \{-1, 0, 1\}$ was the one sketched in \cref{fig:SA_Cov} and explained in Section \ref{sec:BRI}. We also implement the equilibrium induction (see Section \ref{sec:scheme-localEqm}) for both Case I \& II and observe that players behave aggressively when they have more controls than their rivals in the finite-control scenario. As sketched in \cref{fig:MA_EquiInduc}, Player 2 will have the ``last word'' and $\lim_{t\rightarrow\infty}M^\ast_t=-1$. Consequently, she can behave more aggressively, be the leader more frequently, and reap higher payoff already in the medium-term.

\subsection{Equilibrium Thresholds}
\begin{figure}[b!]
	\centering
	\subfigure[Equilibrium payoff $V^1_0(x)$ of Player 1]{
		\label{fig:MA_GVCaseI-II} 
		\includegraphics[width=0.45\textwidth]{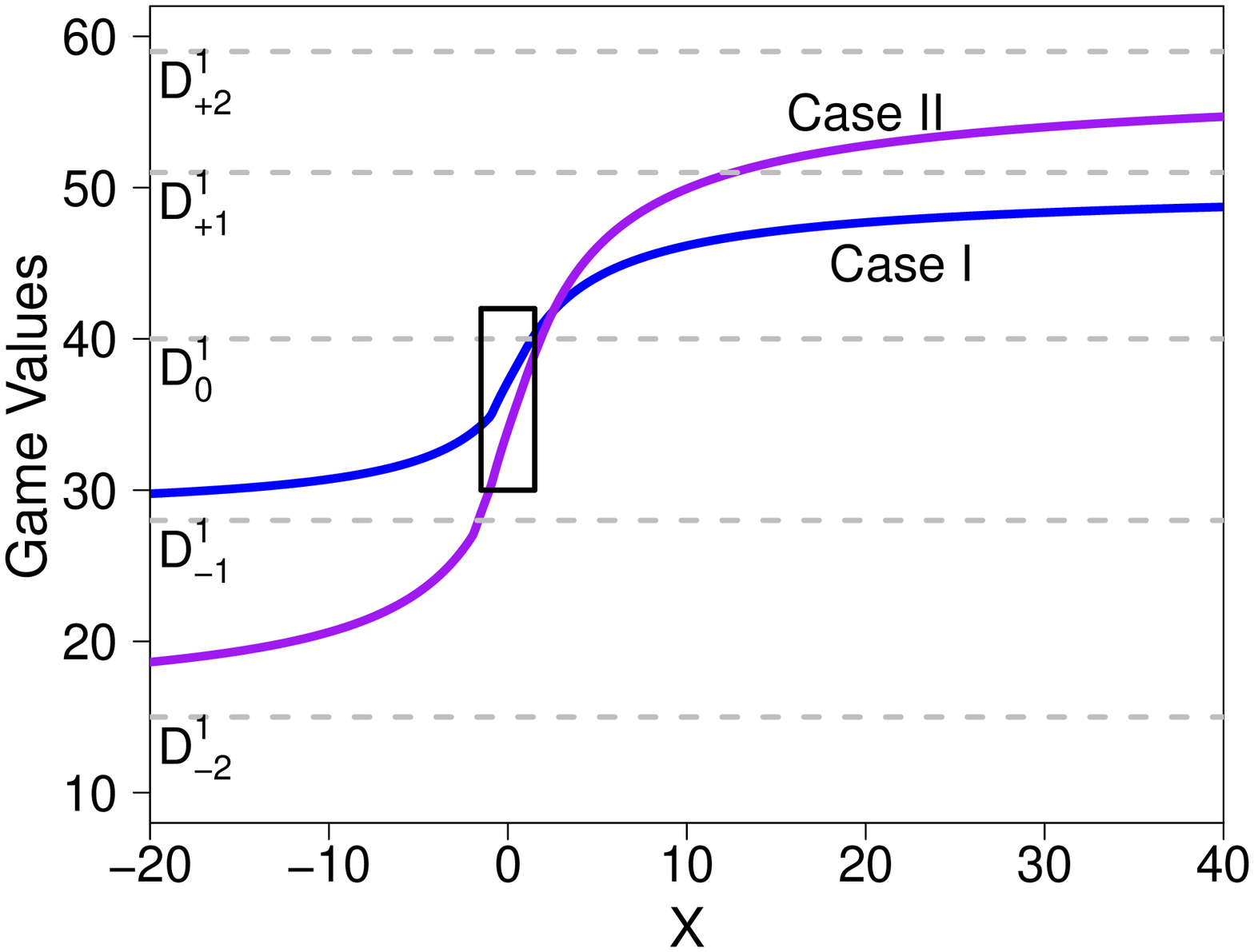}}
	\subfigure[Zoomed-In view of $V^1_0(x)$ ]{
		\label{fig:MA_GVZoomIn}
		\includegraphics[width=0.45\textwidth]{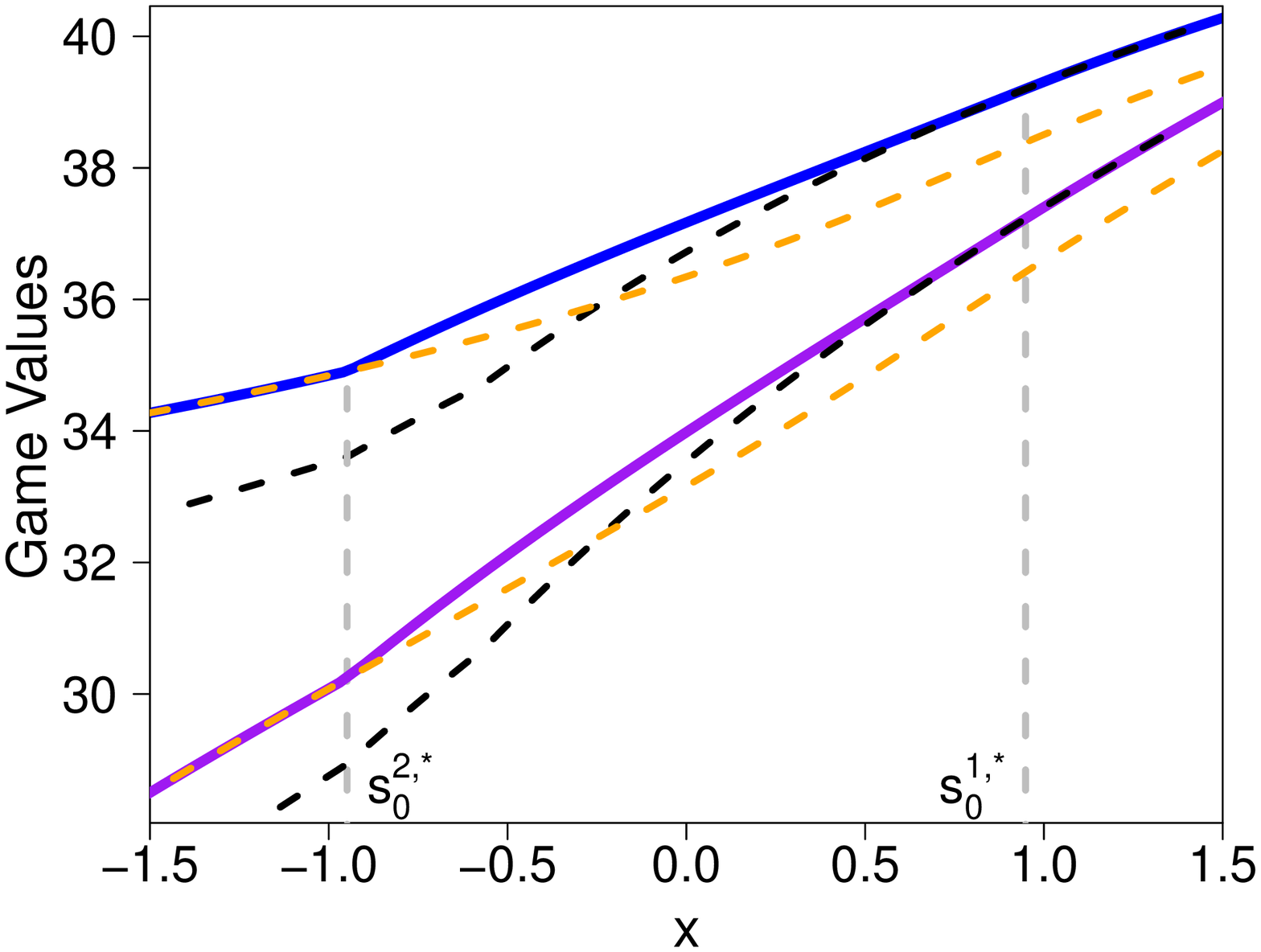}}
	\caption{\label{fig:MA_GV} Equilibrium payoff $V^1_0(x)$ of Player 1 for Case I and Case II. \emph{Left panel}: the dashed levels indicate $D^1_m$; $D^1_{\pm 1}$ are the asymptotes of $V^1_m$ for Case I, and $D^1_{\pm 2}$ are the asymptotes of $V^1_m$ in Case II. The box denotes the region corresponding to the zoomed-in right panel. \emph{Right}: Dashed curves denote the leader payoffs $V^1_{+1}(x)-K^1_0(x)$ and the follower payoffs $V^1_{-1}(x)$ for each Case. The same resulting thresholds $s^{i,\ast}_0$, $i\in\{1,2\}$ are adopted in \emph{both} Cases.}
\end{figure}

\cref{tb:NE_compare} lists the computed equilibrium thresholds for the three cases. Recall that in Case I $M$ is restricted to be in $\{-1,0,1\}$, so Player 1 (P2) is not allowed to act when $M_t = +1$ ($M_t = -1$, respectively), hence there is no $s^{1,\ast}_1$ or $s^{2,\ast}_{-1}$. Thus, there are 4 total thresholds to be computed, and 12 equations in the system \eqref{eq:AP_System}. Due to the  symmetric parameter setting, thresholds of Player 1 are symmetric to thresholds of Player 2 around 0, so in principle the equilibrium is fully characterized by the pair $s^{1,\ast}_0, s^{1,\ast}_{-1}$. Similarly, in Case II there are 8 thresholds (4 unique ones) and 24 equations, and in Case III there are 12 thresholds and 36 equations.

\begin{table}[hbt]
	\centering
	$\begin{array}{c|c|c|r|r|r|r|r|r}
	\multirow{2}{*}{Regime} & \multirow{2}{*}{$\pi^1_m$} & \multirow{2}{*}{$\pi^2_m$}&
	\multicolumn{2}{c|}{\text{Case I}}  & \multicolumn{2}{c|}{\text{Case II}} &  \multicolumn{2}{c}{\text{Case III}}     \\
	&&& \text{Player 1} & \text{Player 2} & \text{Player 1} & \text{Player 2} & \text{Player 1} & \text{Player 2} \\
	\cline{1-3}\cline{8-9}
	+3&6.0&0.0&  &  &  &  &  -  & 3.47853  \\ \cline {6-7}
	+2&5.9&1.5&  &  & - & 0.25184 & 13.16216 & 0.25184   \\ \cline {4-5}
	+1&5.1&2.8& - & -0.56688  & 1.90352 & -0.56688 & 1.90352 & -0.56688  \\
	\,\,0&4.0&4.0& 0.94861 & -0.94861 & 0.94861 & -0.94861 & 0.94861 & -0.94861  \\
	-1&2.8&5.1& 0.56688 &  - & 0.56688 &  -1.90352 & 0.56688 &  -1.90352  \\ \cline {4-5}
	-2&1.5&5.9&  &  & -0.25184 & - & -0.25184 &  -13.16216  \\ \cline {6-7}
	-3&0.0&6.0&  &  &  &  & -3.47853 & -  \\ \cline{8-9}
	\end{array}$
	\caption{Equilibrium thresholds $s^{i, \ast}_m$ for Cases I, II \& III of Section~\ref{sec:MA}. \label{tb:NE_compare}}
\end{table}

A major finding is that the players implement the \textit{same} thresholds at each interior regime in \textit{all} cases.  For example, $s^{1,\ast}_0 = 0.94681$ in all three Cases I/II/III. Therefore, when in regime 0, Player 1 ``does not see'' whether stage +2 is reachable or not, and only makes her decision based on $\pi^i_{\pm 1}$. This phenomenon is un-intuitive in the following two aspects. On the one hand, Player 1 adopts the same equilibrium threshold $s^{1,\ast}_0=0.94861$ despite the fact that she can further exercise switches to enhance her dominance at state +1 in Cases II \& III. The latter would be expected to make the switch from 0 to +1 more valuable and therefore make P1 more aggressive in regime 0. As we see, this intuition, while valid in single-agent contexts, fails in the constructed equilibrium.  On the other hand, players are also myopic about the multi-step threat from the switches of the other player. For example, P2 implements the same threshold $s^{2,\ast}_{+1}=-0.56688$, though in Cases II \& III she is facing the threat that P1 may switch the market to an even more disadvantageous regime.

\cref{fig:MA_GVCaseI-II} plots equilibrium payoffs of Player 1 $x\longmapsto V^{1}_{0}(x)$ (constructed as \eqref{eq:DA_vFunc}) when they are at equal strength $M_0 =0$  for Case I \& II. As expected, $V^1_m$'s are continuous, increasing and bounded: $V^{1, (I)}_0$ is bounded by $D^1_{-1}$ and $D^1_{+1}$, while $V^{1,(II)}_0$ is bounded by $D^1_{-2}$ and $D^1_{+2}$. 
Note that when these players are at equal strength locally (i.e.~$X_0=0$) Player 1 has lower equilibrium payoff in Case II, which can be interpreted as an influence of heavier competition between them.

The right panel of \cref{fig:MA_GVZoomIn} illustrates the nature of the variational inequalities for $V^1_m$.
Dashed lines represent leader payoff $V^1_{+1}(x)-K^1_{0}$ from  switching $(M_t)$ to (+1) and the follower payoff $V^1_{-1}(x)$. We have that $V^1_0(x)$ coincides with $V^1_{-1}(x)$  for $x<s^{2,\ast}_m$ (stopping region of P2), and smooth-pastes to $V^1_{+1}(x)-K^1_0(x)$ at the switching threshold $s^{1,\ast}_0$ of P1. As explained, this switching threshold is the same in Case I and II even though all the game payoffs (in particular $V^1_0$ and $V^1_{\pm 1}$) change.

\subsection{Macroscopic Market Structure in Equilibrium}

Due to the mean-reverting property of the OU process in \eqref{eq:OU}, the resulting $M^\ast$ is recurrent. In particular, at the extreme regimes $\cM$ is guaranteed to move back towards zero. \cref{tb:NE_Pi} presents the resulting long-run proportion of time $M^\ast$ spends at each regime, $\rho_m$ given by \eqref{eq:SMG_longrun}. The table also lists the transition matrix $\mathbf{P}$ of the extended jump chain $\cM$, and mean sojourn times $\vec{\xi}$ for Case I, whence $\cM$ takes values in $\{-1^{-}, 0^-, 0^+, +1^{+}\}$. Note that due to the limited number of regimes, the invariant distribution of $\cM$ is uniform, i.e.~the original jump chain $\tM$ spends half the time at regime 0 and 25\% of its time at regimes $\pm 1$. Of course, the corresponding sojourn times are not equal, so the long-run distribution of $M^\ast$ is more complex.

\begin{table}[ht]
\begin{minipage}{0.62\textwidth}
		\begin{tabular}{c|c|c|r|r|r}
\multicolumn{6}{}{}\\
	{Regime} & {$\pi^1_m$} & {$\pi^2_m$}&  {Case I} & {Case II} & {Case III}  \\
	\hline\hline
	$+3$&6.0&0.0&   &  & 0.00002 \\ \cline{5-5}
	$+2$&5.9&1.5&    & 0.34476 & 0.34474  \\ \cline{4-4}
	$+1$&5.1&2.8&   0.47186 &  0.12711 &   0.12711 \\
	$\,\,0$&4.0&4.0 & 0.05628 &  0.05628 &   0.05628 \\
	$-1$&2.8&5.1&  0.47186 & 0.12711  &    0.12711 \\ \cline{4-4}
	$-2$&1.5&5.9&   &  0.34476 &  0.34474\\ \cline{5-5}
	$-3$&0.0&6.0&   &  & 0.00002  \\
\end{tabular}
\end{minipage}
\begin{minipage}{0.35\textwidth}
		\begin{align*}
	\bm{P}=\begin{blockarray}{ccccc}
	& {(-1)^-} & {(0)^+} & {(0)^-} & {(+1)^+}\\
	\begin{block}{c(cccc)}
	{(-1)^-} & 0 & 1 & 0& 0\\
	{(0)^+} &0.205 &0&0&0.795\\
	{(0)^-} &0.795&0&0&0.205\\
	{(+1)^+} &0&0&1&0\\
	\end{block}
	\end{blockarray}, \\
\vec{\xi} = \!\! \begin{blockarray}{ccccc}
	& {(-1)^-} & {(0)^+} & {(0)^-} & {(+1)^+}\\
	\begin{block}{c(cccc)}
	&4.431 & 0.264 & 0.264 & 4.431 \\
	\end{block}
	\end{blockarray}.
	\end{align*}
\end{minipage}
\caption{Equilibrium stationary distribution $\rho_m$ of $M^\ast$ for Cases I, II \& III. \emph{Right:} Dynamics of $\cM^\ast$ in Case I. \label{tb:NE_Pi}}
\end{table}

From \cref{tb:NE_compare}, we observe that the thresholds $s^{i,\ast}_0$ are quite low, so that $M^\ast$ does not spend much time in regime 0 and the market is typically not at  ``equal strength''. In Case I, only one level of market dominance is possible and so we observe rapid switches from ``equal strength'' to ``P1 dominant'' or ``P2 dominant'', each of which occurs around $\rho^{(I)}_{\pm 1} =47\%$ of the time.  In Case II, because $s^{i,\ast}_0$ remain the same, we have the same $\rho^{(II)}_0$, so the market continues to be dominated (but now by different degrees) by one player around 47\% of time. Thus, the long-run distribution of regimes $\{+1, +2\}$ in Case II can be considered as ``splitting'' of those 47\% of regime $+1$ in Case I. Moreover, when one player dominates the market, she will max-out her dominance most of the time ($\rho^{(II)}_{\pm 2} = 34\%$ out of 47\%).

A second finding is that one can effectively endogenize the domain $\mathcal{M}$ of $M$. Recall that concavity of profit rates $\pi^i_m$ in terms of $m$ reduces players' incentive to make further switches if the game stage is already advantageous. On the contrary, the rival becomes more incentivized to switch $M$ back towards $0$. In the presented example, we make the marginal gain in profit rates minimal when going from +2 to +3 (and -2 to -3 for Player 2, respectively). As a result, in Case III there is very little incentive for P1 to switch from +2 to +3, reflected in the very high equilibrium threshold $s^{2,\ast}_{+2}=13.1621594$.
Because this threshold is far above the mean-reverting level $\theta=0$, it follows that these players are not likely to enhance their dominance up to the maximum level and regimes $\pm 3$ will take place extremely rarely; according to \cref{tb:NE_Pi}, $M^\ast$ spends less than 0.001\% of time in those extreme regimes. Consequently, from a financial perspective it is reasonable to simply restrict $M$ to be in $\{-2,-1, 0, +1,2\}$, since effectively $\rho^{(III)}_m \simeq \rho^{(II)}_m$ for all $m$.

\subsection{Effect of Profit Ladder}
\begin{figure}[b!]
	\centering
	\subfigure[Thresholds of Player 1]{
		\label{fig:Thres1} 
		\includegraphics[width=0.45\textwidth]{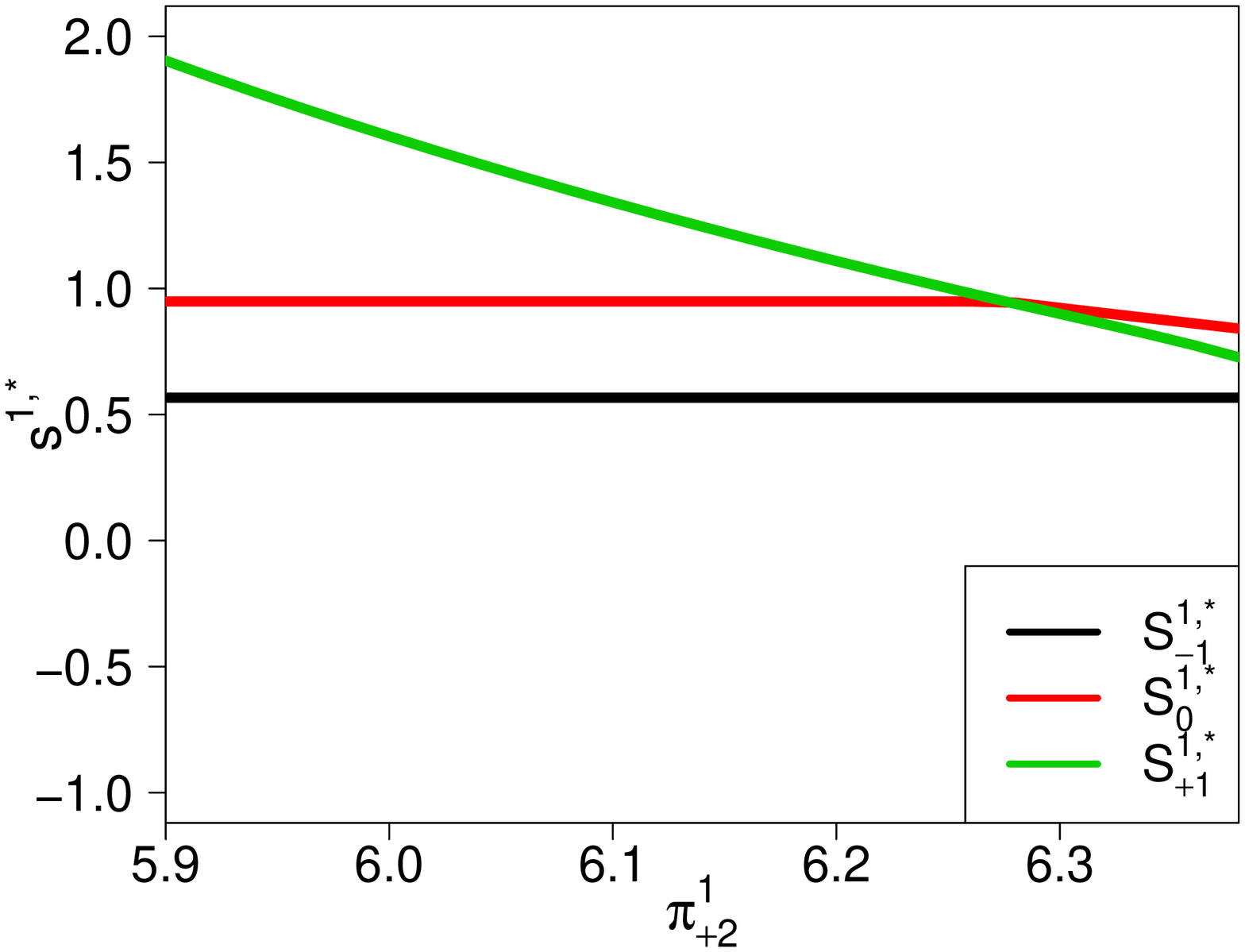}}
	\subfigure[Thresholds of Player 2]{
		\label{fig:Thres2}
		\includegraphics[width=0.45\textwidth]{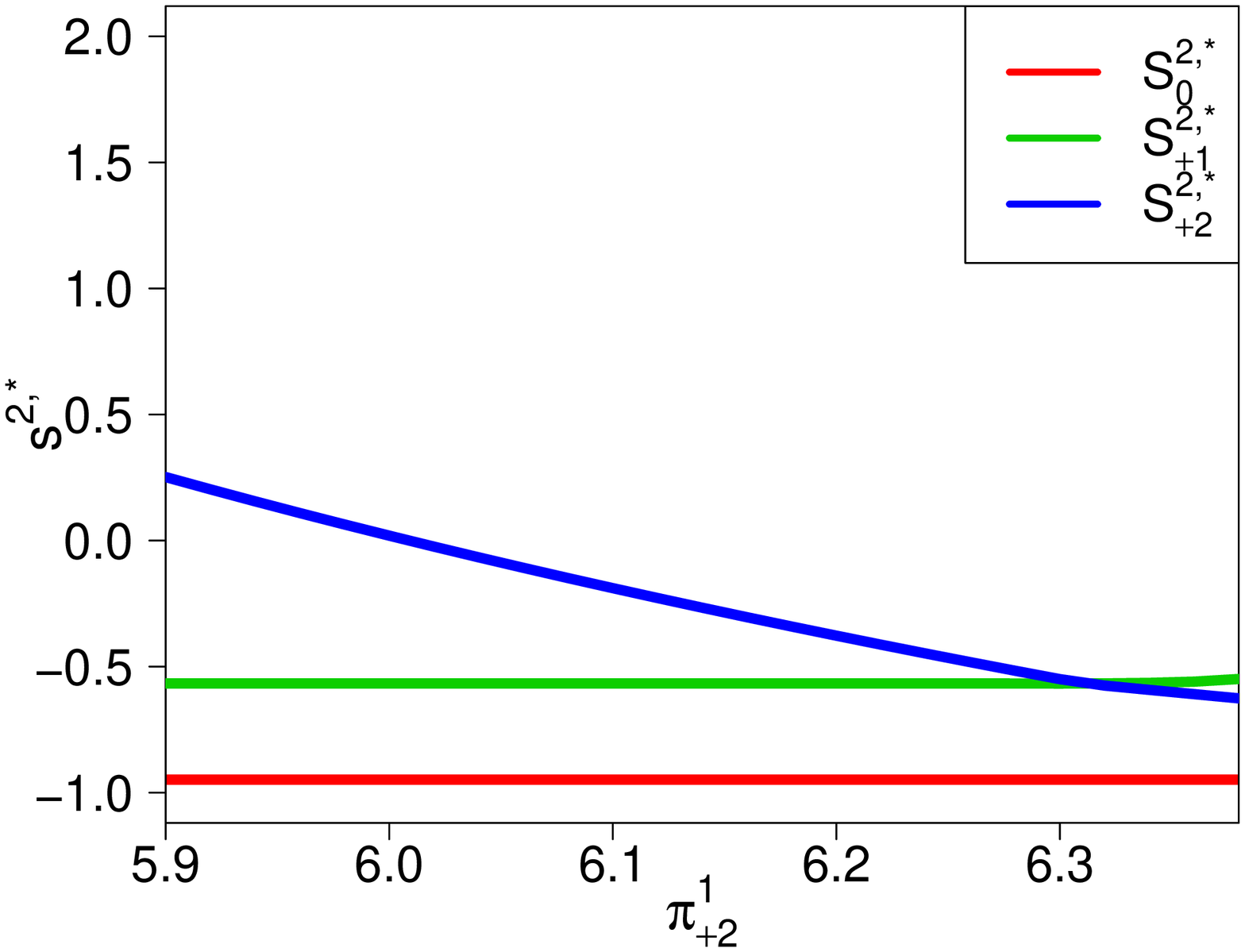}}
	\caption{Equilibrium switching thresholds $s^{i,\ast}_m$ (P1 on the left, P2 on the right) as the profit rate $\pi^1_{+2}$ of P1 varies in Case IV. We use the same type of line for each regime $m$ across the two panels.
}
	\label{fig:EffofPi} 
\end{figure}

To isolate the effect of the profit rates $\pi^i_m$, we construct threshold-type equilibria with $M$ restricted to $\mathcal{M}^{(IV)}=\{-1, 0, +1, +2\}$, and vary profit rate $\pi^1_2$ of Player 1 at stage +2 (all other profit rates remain as in \cref{tb:NE_pi}). Resulting equilibrium thresholds of these players are sketched in \cref{fig:EffofPi} by solid lines.
As expected, when $\pi^1_{+2}$ increases, there is more benefit to being in regime $+2$ and as a result, Player~1 is more willing to make a switch up from +1. Consequently, she implements  lower switching thresholds and $s^{1,\ast}_{+1}$ is decreasing in $\pi^1_{+2}$. On the contrary, she does not change her threshold $s^{1,\ast}_0$ at regime $0$, confirming the myopic nature of equilibrium thresholds discussed in the previous section. However, eventually $\pi^1_{+2} - \pi^1_{+1}$ is large enough (or alternatively $s^{1,\ast}_{+1}$ is low enough) to trigger simultaneous switches ($s^{1,\ast}_0>s^{1,\ast}_{+1}$), so that P1 will pass directly from regime 0 to regime 2 (recall that we assume that this incurs two switching costs, linearly added). In the latter situation, she switches sooner already in regime 0, see the extreme right of \cref{fig:Thres1}, where $s^{1,\ast}_0$ starts changing, as soon as $s^{1,\ast}_{+1} < s^{1,\ast}_0$.

Turning attention to Player~2, her switching threshold  $s^{2,\ast}_0$ in regime 0 is never affected by $\pi^1_{+2}$. Moreover, while her profit rate in regime +2 is unaffected, more aggressive behavior of P1 who switches into $M_t = +2$ more frequently, causes her to respond by lowering $s^{2,\ast}_{+2}$. Additionally, in the situation where P1 goes straight from 0 to +2 ($\pi^1_{+2} \ge 6.25$), P2 increases $s^{2,\ast}_{+1}$, adjusting his strategy in response to a more aggressive strategy of P1 which reduces his anticipated gain from switching $M$ from +1 to 0. These observations illustrate the complex feedback effects between thresholds in different market states and the underlying $\pi^i_m$'s.

\subsection{Effect of Switching Costs}
\begin{figure}[b!]
	\centering
		\subfigure[Switching thresholds of Player 1]{
		\label{fig:SvsK}
		\includegraphics[width=0.45\textwidth]{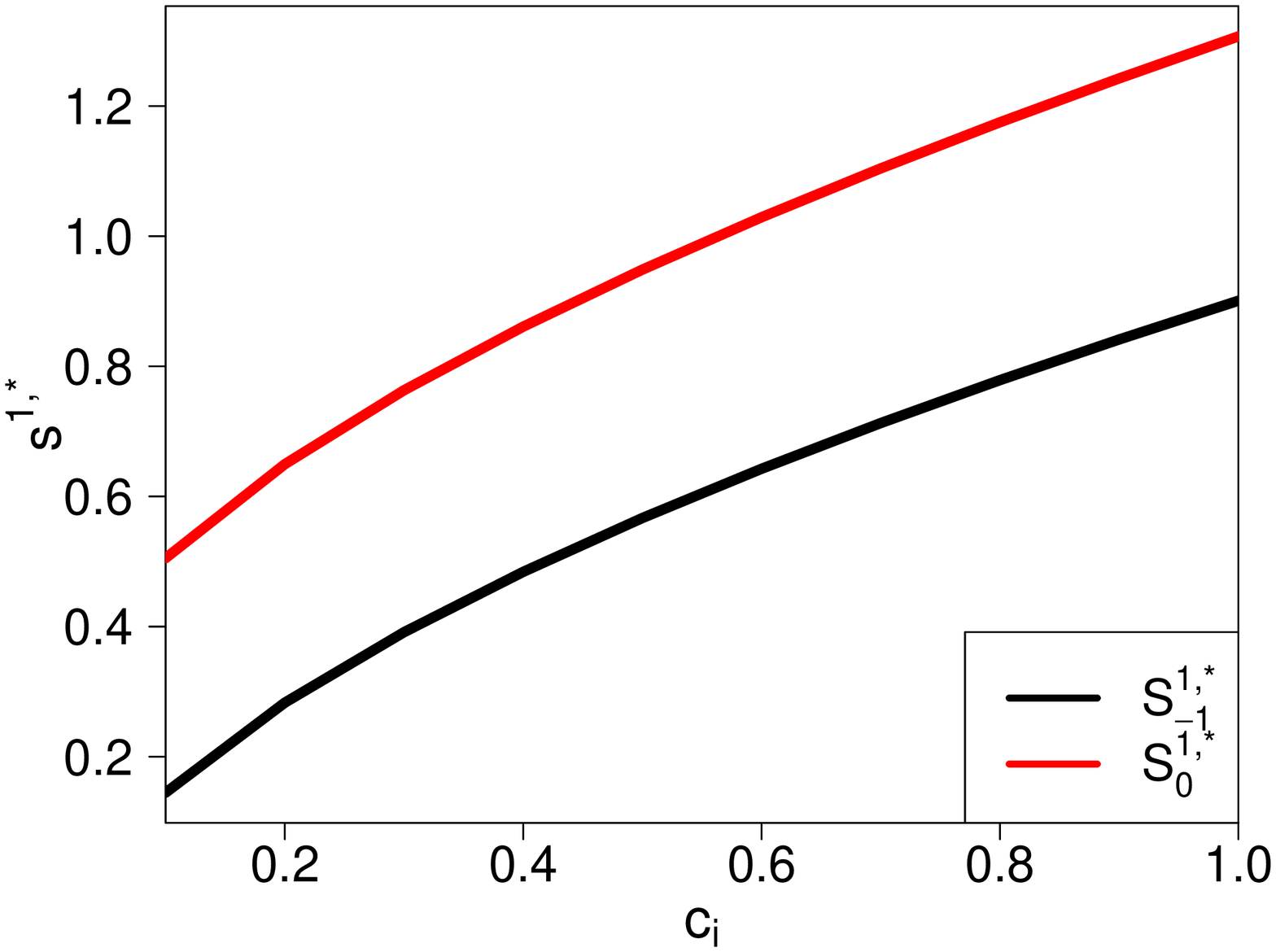}}
	\subfigure[Equilibrium payoff of Player 1 ]{
		\label{fig:GPvsK} 
		\includegraphics[width=0.45\textwidth]{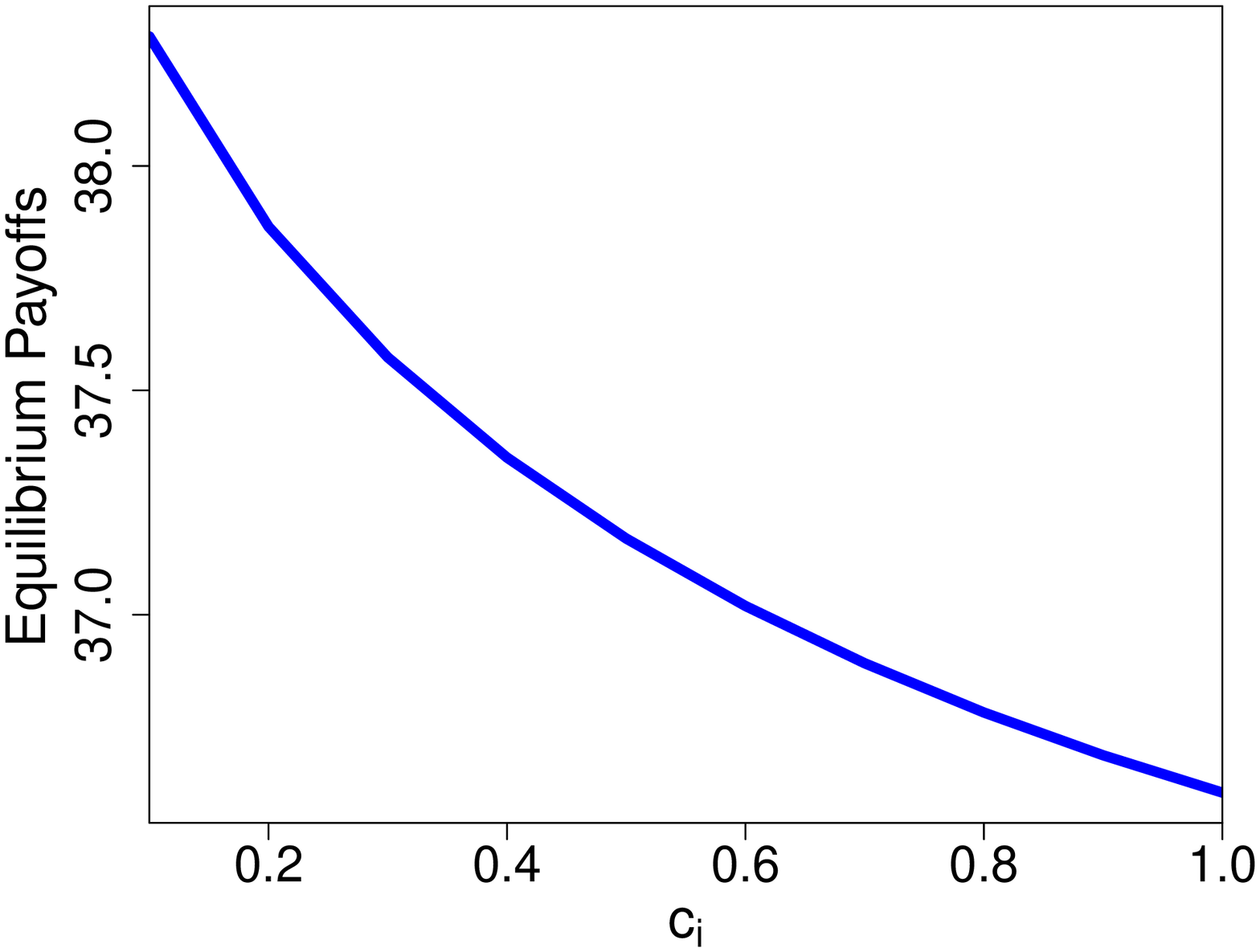}}
	\caption{\emph{Left}: Equilibrium thresholds adopted by P1 as $c_i$ varies in $[0.1, 1]$. Thresholds of P2 are anti-symmetric about 0. \emph{Right}: Equilibrium payoff at ``equal-strength'', i.e.~$V^{1}_0(0)=V^2_0(0)$ due to symmetry.}
	\label{fig:VariatK} 
\end{figure}

Another essential parameter is  the switching cost $K$. To study the effect of $K$, we vary the overall level of switching costs in  $K^i(x)=c_i\cdot(1+e^{(-1)^i\beta_i})$. Specifically, we try $c_i \in [0.1, 1]$, i.e.~from 20\% to 200\% relative to the baseline $c_i=0.5$ used in preceding setup, with all other parameters unchanged. The resulting equilibrium thresholds of Player 1 and her equilibrium payoff at the ``equal strength'', $V^1_0(0)$ are sketched in \cref{fig:VariatK} for Case I where $\mathcal{M}=\{-1, 0, +1\}$. Observe that as $c_i$ decreases (from the right to the left in \cref{fig:VariatK}) P1 adopts lower thresholds while P2 adopts higher thresholds by symmetry, which means they switch the macro market environment more frequently. For instance, the expected sojourn time of $\tM^\ast$ at regime $(+1)^+$ drops from 6.522 at $c_i=1.0$ to 1.887 at $c_i=0.1$, while the expected sojourn time at regime $(0)^+$ drops from 0.422 at $c_i=1.0$ to 0.105 at $c_i=0.1$. Recall that $K^i_0(x)$ is a function of $x$, so that lowering the threshold is equivalent to paying more. While lower (single) switching cost induces more frequent switching, the overall cost of switching still declines, so that equilibrium payoffs increase as $c_i$ declines.

\subsection{Multiple Threshold-type Equilibria}
As mentioned, existence of \emph{multiple} {threshold-type} equilibria is highly likely. According to \cref{thm:DA_equi}, any suitable solution to the system of non-linear equations is an MNE of the switching game. This situation arises in Case II above and is ``detected'' by selecting different local equilibria during the equilibrium induction of Section~\ref{sec:scheme-localEqm}. Specifically, in some sub-stages there are two different threshold-type equilibria in the local stopping game, which can be interpreted as ``Sooner'' (players behave more aggressively and switch quickly once $(X_t)$ deviates from zero) and ``Later'' (players are more relaxed and $s^i_m$ are larger in absolute value). This phenomenon was already documented for stopping games in \cite{aid2017capacity}. Then during equilibrium induction we consistently choose (\textit{i}) later equilibria (this is what was done and reported above in Tables~\ref{tb:NE_compare}-\ref{tb:NE_Pi}); (\textit{ii}) sooner equilibria. This generates two different sequences of $\bm{s}^{i,\ast}_m$, which ultimately yield two different solutions to the nonlinear system, reported in \cref{tb:MultiEqui}. We note that in the ``Sooner'' equilibrium which was described previously, players effectively skip regimes $\pm 1$ as $s^{1,\ast}_0>s^{1,\ast}_{+1}$ and $s^{2,\ast}_0<s^{2,\ast}_{-1}$. For example, starting at $M^\ast_0 = 0$ and $X_0 = 0$, P1 will not switch up until $X_t = 1.389 = s^{1,\ast,S}_0$, but then directly go to $M^\ast = +2$ because $1.389 > 1.029 = s^{1,\ast,S}_{+1}$.
As a result, the alternative stationary distribution $\vec{\rho}^{S}$ of $M^{S,\ast}$ is only supported on  $\{-2,0,2\}$ (interestingly, $\rho^S_0 > \rho^L_0$ so in the aggressive equilibrium the market is more frequently at ``equal strength''). Such multiple instantaneous switches are indicative of their aggressiveness---once $(X_t)$ moves in their preferred direction, players attempt to extract maximum dominance by switching into regime $\pm 2$. In line with previous analysis, the Sooner equilibrium carries \emph{lower} equilibrium payoffs, as players are penalized for aggressive interventions that leads to ``wasted'' effort, e.g.~$V^{i, S}_0(0)=33.65<V^{i,L}_0(0)=33.98$.

%

\begin{remark}
  In Case I there seems to be a unique equilibrium, which we conjecture is due to having only a single interior regime where players compete simultaneously. Thus, with $\mathcal{M}=\{-1,0,1\}$ we always observe a unique local threshold-type equilibrium during either of the finite-control inductions. It remains an open problem to establish more precise conditions regarding equilibrium uniqueness in the infinite-control switching game. Similarly, we do not have the machinery to check whether \emph{further} threshold-type equilibria exist in Case II.
\end{remark}

\begin{table}[hbt]
	\centering
	$\begin{array}{c|c|ccccc|r}
	\multicolumn{2}{r|}{}& -2 & -1 & 0 & +1 & +2 & V^i_0(0) \\\hline\hline
	\multirow{3}{*}{Later}& s^{1,\ast,L}_m&  -0.25184 & 0.56688 & 0.95861 & 1.90352 & -&\multirow{3}{*}{33.98}\\
	& s^{2,\ast,L}_m & - & -1.90352 & -0.94861 & -0.56688 & 0.25184 & \\ \cline{2-7}
	& \rho^L_m & 0.34476 & 0.12711 & 0.05628 & 0.12711 & 0.34476 & \\\hline\hline
	\multirow{3}{*}{Sooner}& s^{1,\ast,S}_m&  0.17983 & -0.19139 & 1.38933 & 1.02891 &-&\multirow{3}{*}{33.65} \\
	& s^{2,\ast,S}_m & - & -1.02891 & -1.38933 & 0.19139 & -0.17983 &\\ \cline{2-7}
	& \rho^S_m & 0.41143 & 0 & 0.17714 & 0 & 0.41143&\\
	\end{array}$
	\caption{Equilibrium thresholds $s^{i,\ast}_m$ and long-run distribution $\vec{\rho}$ of $(M^\ast_t)$ associated to two distinct equilibria in Case II}\label{tb:MultiEqui}
\end{table}

\section{Case Study: Long-run Advantage}\label{sec:LA}
Returning to \cref{eg:PF_LA}, we now consider local market fluctuations $X$ to follow a Geometric Brownian motion \eqref{eq:gbm} with
  drift $\mu=0.08$, volatility $\sigma=0.25,$ and discounting rate $r=10\%$. Because $\mu-\frac{1}{2}\sigma^2>0$, $\lim_{t \to \infty} X_t = +\infty$ a.s., and so in the long-run Player 1 will dominate the market since she will eventually have the advantage in terms of $X$.

The profit rates  $\pi^i_m$ are constant and given by
\begin{align*}
\pi^1_{-1} & =0; \quad \pi^1_0 = 3; \quad \pi^1_{+1} = 5; \\
\pi^2_{-1} & =5; \quad \pi^2_0 = 3; \quad \pi^2_{+1} = 0.
\end{align*}
The switching costs $K^i_m$'s are again independent of $m$ and driven by $X$:
\begin{align}
K^1(x) := (10-x)_+, \qquad K^2(x):= (-2+x)_+.
\end{align}
This case study can be interpreted as competition between an energy producer using a renewable resource (Player 1) and a producer using exhaustible resources (Player 2). The competition is in terms of generating capacity, with $M_t$ denoting the \emph{relative} production capacity. Here $X_t$ represents the marginal cost of exhaustibility which connects to the relative cost of increasing capacity.
We expect that $X_t \to +\infty$ (``peak oil''); as non-renewable resources are depleted, P2 becomes noncompetitive. In the long run, P1 will therefore dominate, however there is no upper bound on how many times the competing investments in new capacity will take place. Thus, the market will first go through a transient phase where both producers compete, and then will eventually enter the high-$X$ regime where the renewable P1 dominates and (endogenously) never relinquishes her advantage.

\begin{figure}[b!]
	\centering
	\subfigure[Sample Trajectory]{
		\label{fig:LA_SampleGBM} 
		\includegraphics[width=0.45\textwidth]{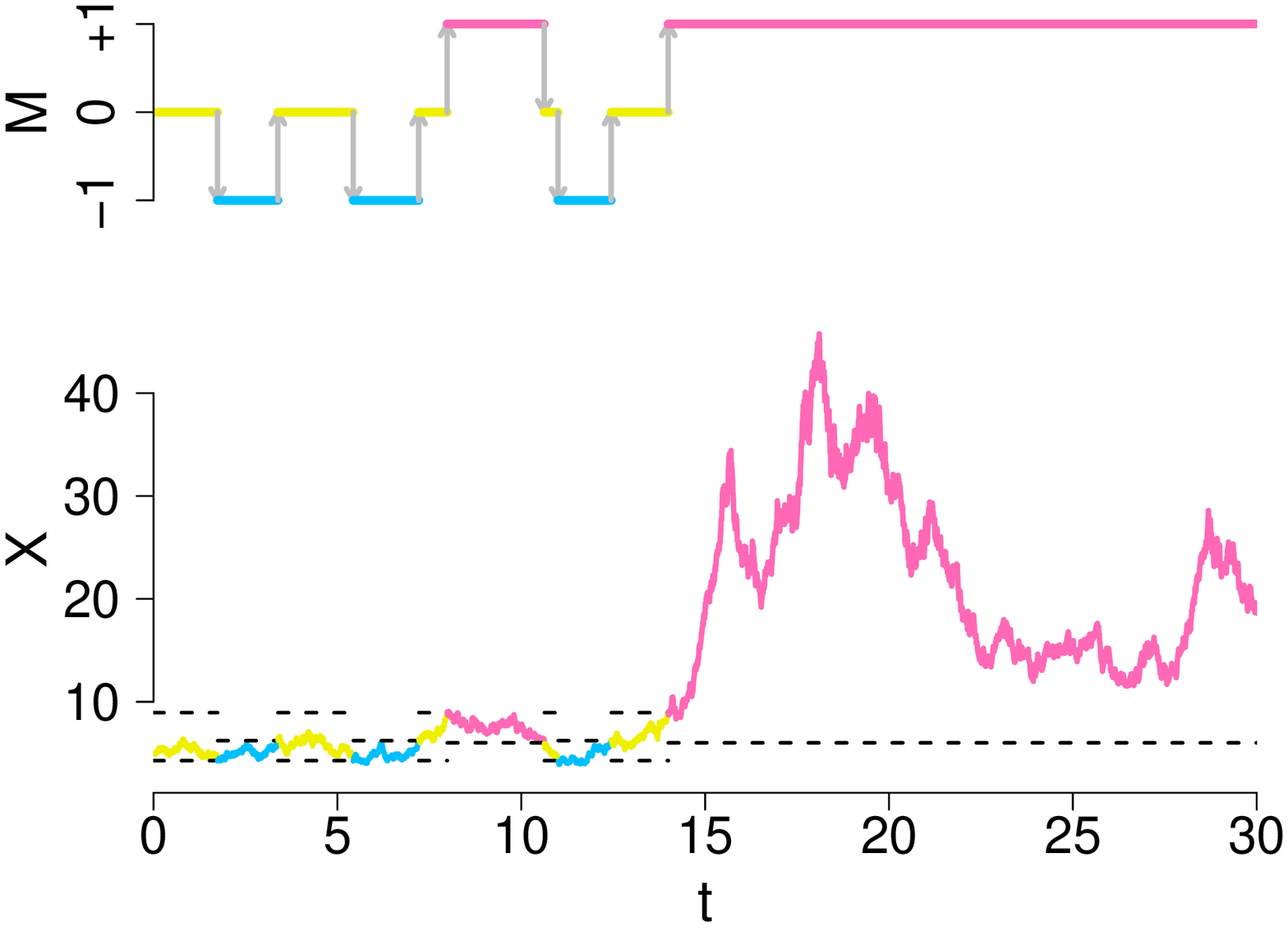}}
	\subfigure[Distribution of $M^\ast_t$]{
		\label{fig:M_dist_est}
		\includegraphics[width=0.45\textwidth]{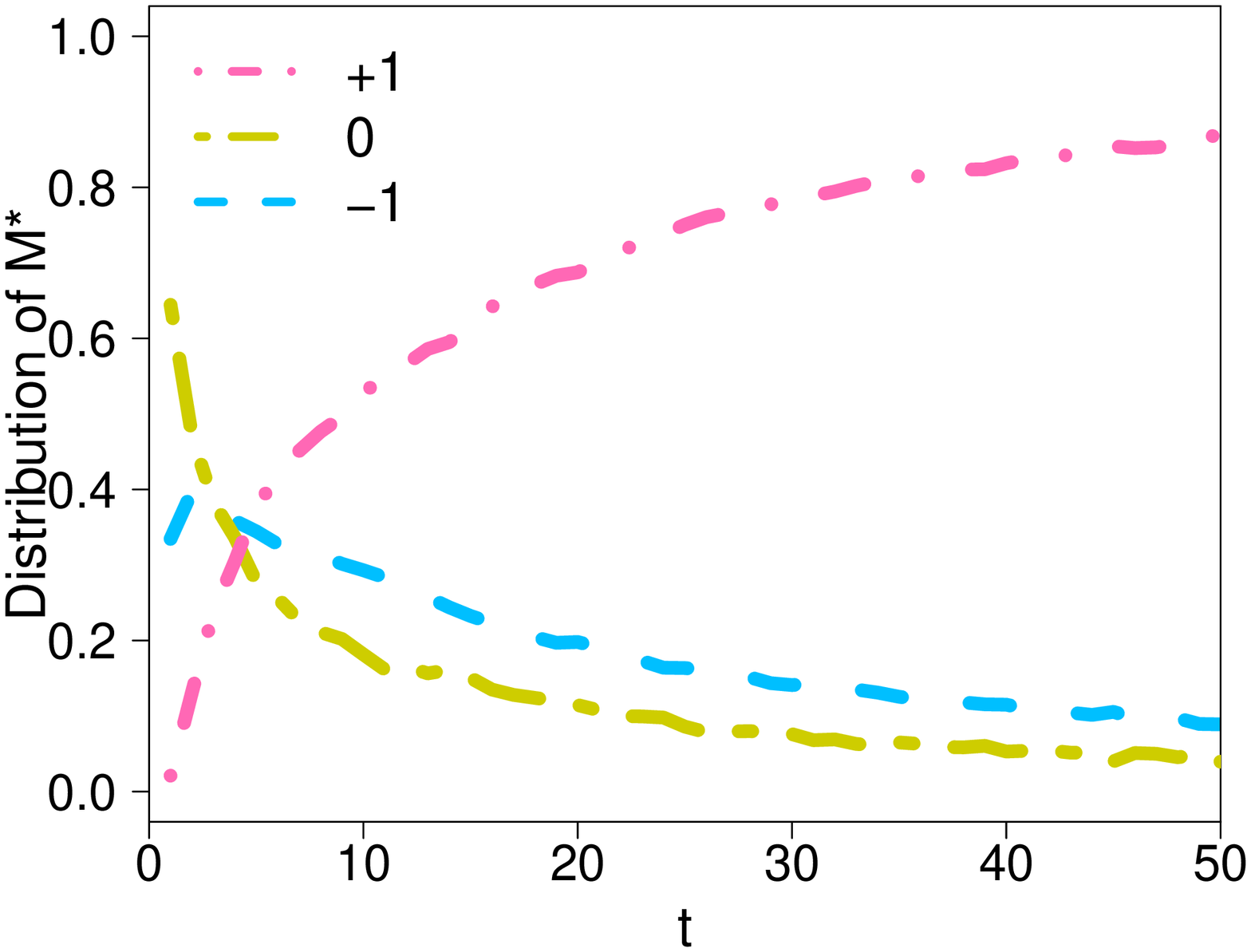}}
	\caption{\emph{Left}: A trajectory of $(X_t, M^\ast_t)$ for the GBM example of Section~\ref{sec:LA}, starting from $X_0=5$, $M^\ast_0=0$. \emph{Right}: Distribution of $M^\ast_t \in \{ -1, 0, 1\}$ as a function of $t$.
	\label{fig:LA_GV&Dist}} 
\end{figure}

\begin{table}[htb]
	\centering
	\begin{tabular}{cccrr}
		& \multicolumn{3}{c}{Expansion Thresholds $s^i_m$ }  & {Ave Plants Built} \\ 
		\cline {2-4}
		& $-1$ & $0$ & $+1$ &    \\
		\hline\hline
		Player 1	& 5.9796 & 8.9594 & - & 3.8151   \\ 
		\hline
		Player 2		& - & 4.1296 & 5.9574 & 2.8151 
	\end{tabular}
	\caption{\label{tb:LA_equi}Equilibrium in the GBM case study of Section~\ref{sec:LA}. Average Plants Built refers to the expected number of switches by player $i$ starting at $X_0=5, M_0 = 0$. }
\end{table}

\cref{tb:LA_equi} shows the resulting equilibrium thresholds, and \cref{fig:LA_SampleGBM}~plots a trajectory of $(X_t)$ and $(M^\ast_t)$ starting at $X_0=5$, $M^\ast_0=0$. The right panel \cref{fig:M_dist_est} shows the distribution of $M^\ast$ via $t \mapsto \mathbb{P}(M^\ast_t=m)$. We observe that Player 2 is likely to make the first expansion ($\mathbb{P}(M^\ast_t = -1)$ increases for low $t$), while in the medium-term Player 1 becomes more and more likely to be dominant. In line with $X_t \to +\infty$ (due to $\mu-\sigma^2/2>0$)  we have $\mathbb{P}(M^\ast_t = +1) \to 1$ as $t$ grows. The probability of absorption for $\cM^\ast$ when moving up from the states $(0)^{\pm}$ is (see \cref{AP_transP})
\begin{align*}
P_{(+1)^a} = \lim_{u\uparrow\infty}\frac{(s^{1,\ast}_0)^{1-\frac{2\mu}{\sigma^2}}-(s^{2,\ast}_{1})^{1-\frac{2\mu}{\sigma^2}}}{u^{1-\frac{2\mu}{\sigma^2}}-(s^{2,\ast}_{+1})^{1-\frac{2\mu}{\sigma^2}}} = 1 - \Big(\frac{s^{1,\ast}_0}{s^{2,\ast}_{1}}\Big)^{1-\frac{2\mu}{\sigma^2}}=0.4709,
\end{align*}
leading to the transition probability matrix $\bm{P}$ of $\cM^\ast$ as
\begin{align}\label{eq:LA_transP}
\bm{P}=\begin{blockarray}{cccccc}
&{(-1)^-}& {(0)^+} & {(0)^-}&{(+1)^+}&{(+1)^a}\\
\begin{block}{c(ccccc)}
{(-1)^-} &0 & 1 & 0 & 0 & 0  \\
{(0)^+} &0.374 & 0 & 0 & 0.331 & {0.295}  \\
{(0)^-} &0.379 & 0 & 0 & 0.329 & {0.292}  \\
{(+1)^+}&0 & 0 & 1 & 0 & 0  \\
{(+1)^a} &0 & 0 & 0 & 0 & 1  \\
\end{block}
\end{blockarray}.
\end{align}
Note that in the scenario plotted in \cref{fig:LA_SampleGBM}, $M^\ast_t = +1$ after $t=15$ which can be interpreted as ``absorption''. The theoretical average time until absorption (defined in \eqref{eq:SMS_abTime}) is $\mathbb{T}_0(5)=30.775$. In the Figure we also note that P1 makes 5 switches up and P2 makes 4 switches down. Recall that on the infinite time horizon P1 will always make one more switch since $M^\ast_t = +1$ eventually. The last column of \cref{tb:LA_equi} shows the \emph{average} number of expansions implemented by each producer, $\mathbb{N}^i_0(5)$ defined in \eqref{eq:SMS_numsw}.

\subsection{Effect of the Drift $\mu$ and Volatility $\sigma$}
\begin{figure}[b!]
	\centering
	\subfigure[$\sigma=0.25$, $\mu\in(0.05, 0.15)$]{
		\label{fig:LA_NumSw_mu} 
		\includegraphics[width=0.45\textwidth]{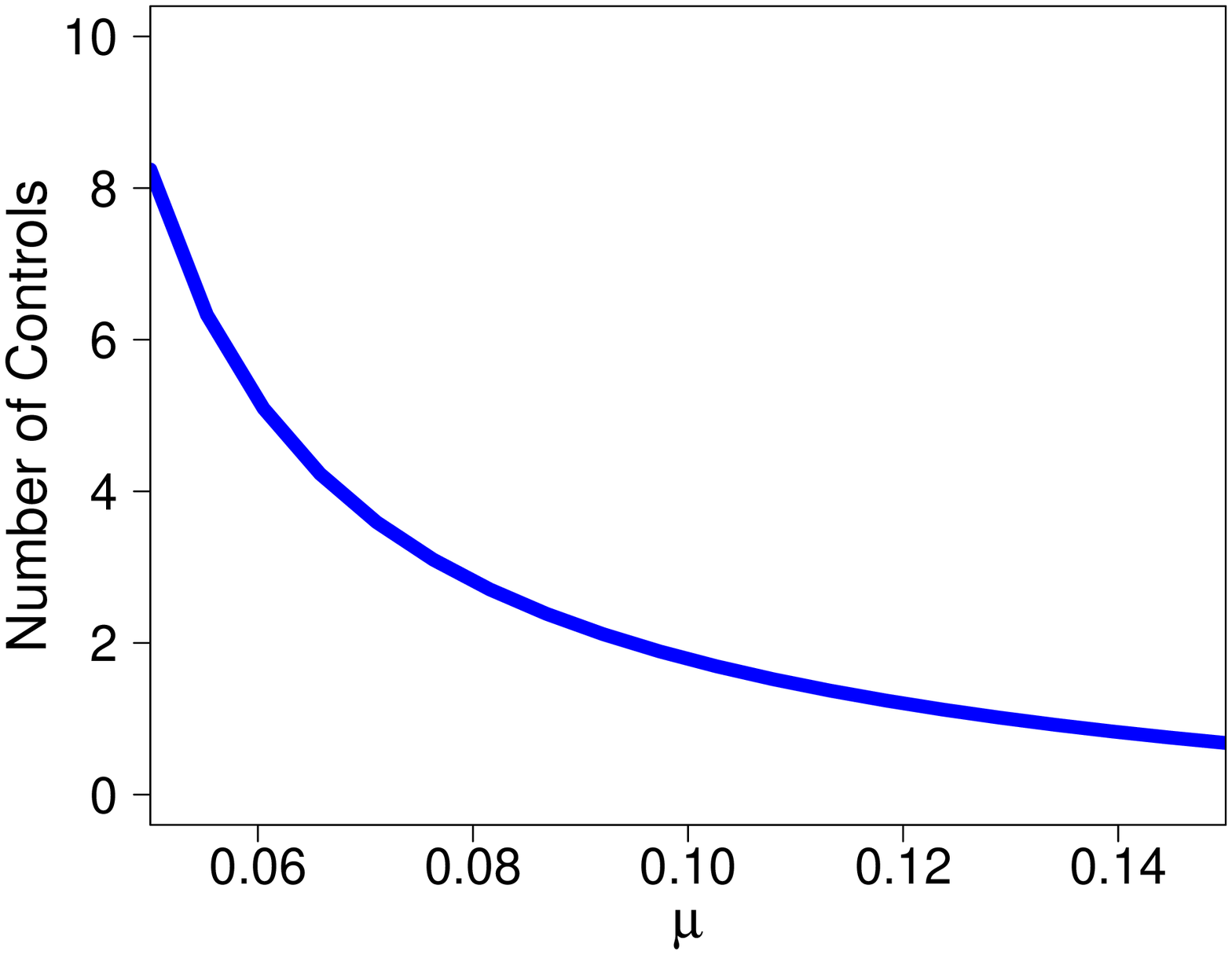}}
	\subfigure[$\sigma=0.25$, $\mu\in(0.05, 0.15)$]{
		\label{fig:LA_distrib_est}
		\includegraphics[width=0.45\textwidth]{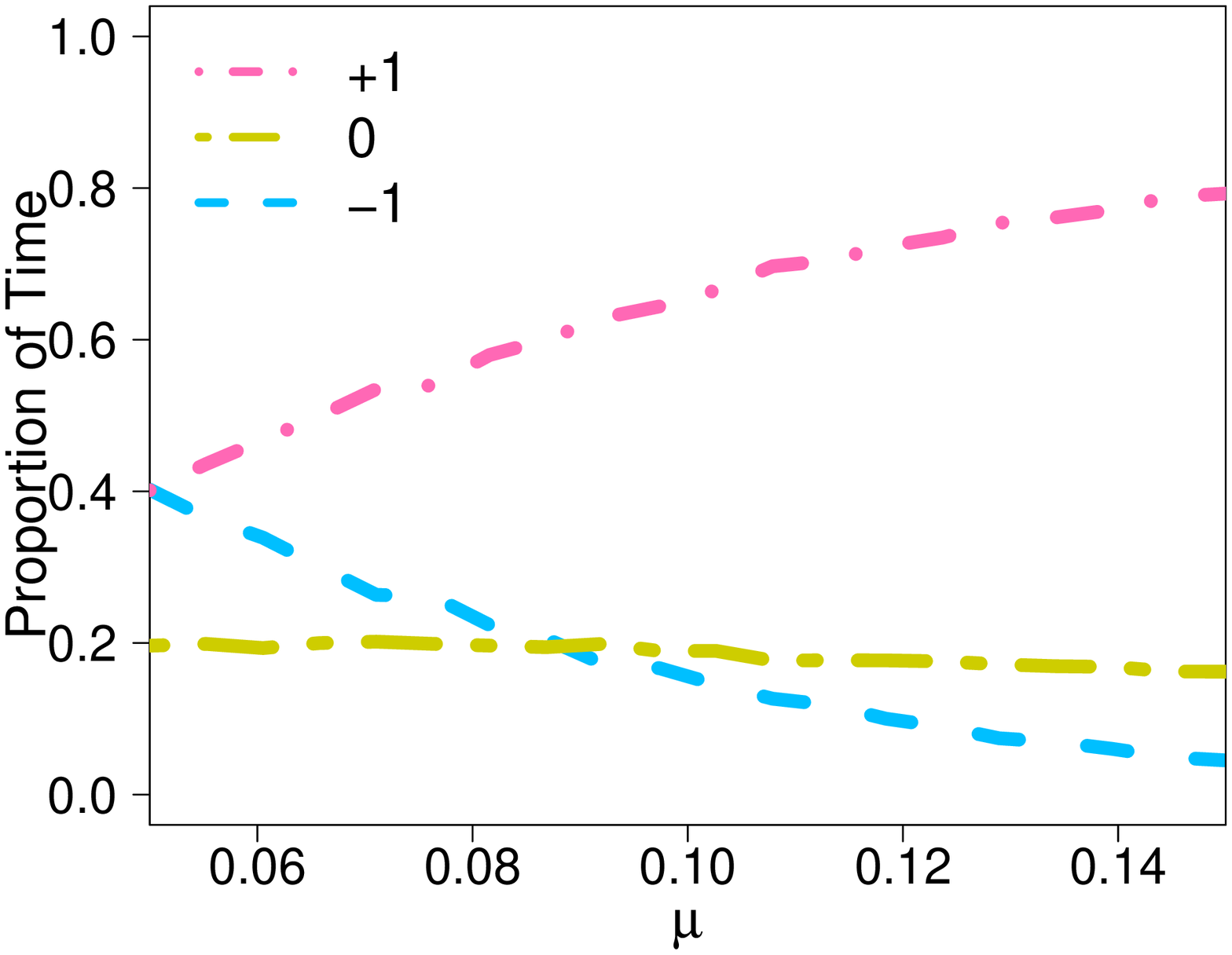}}
	\subfigure[$\mu=0.08$, $\sigma\in(0.20, 0.35)$]{
		\label{fig:LA_NumSw_sig} 
		\includegraphics[width=0.45\textwidth]{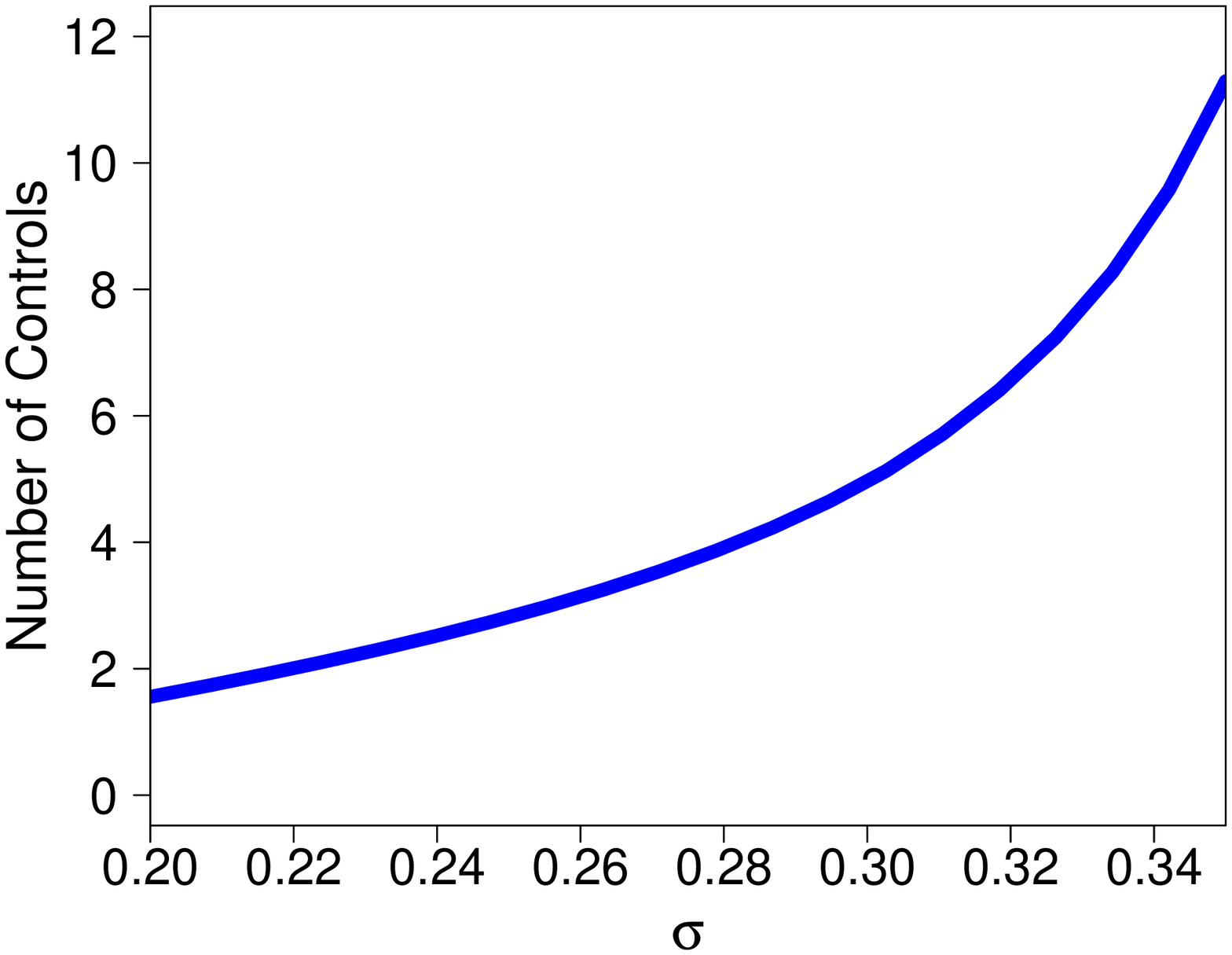}}
	\subfigure[$\mu=0.08$, $\sigma\in(0.20, 0.35)$]{
		\label{fig:LA_distrib_est_sig}
		\includegraphics[width=0.45\textwidth]{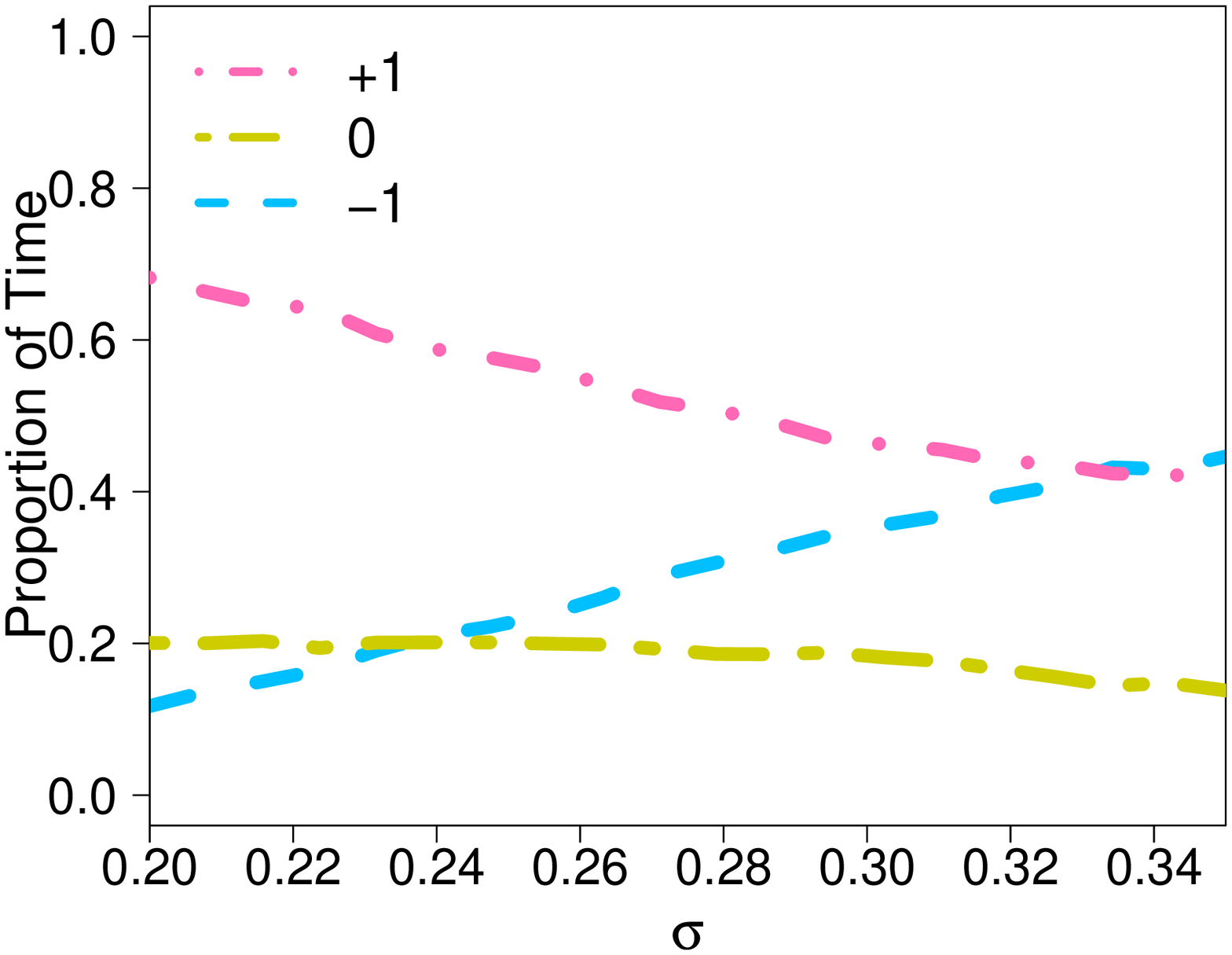}}
	\caption{ \emph{Left}: Average total number of switches exercised by P2. \emph{Right}: Estimated proportion of time that $(M^\ast_t)$ spends in regime $m$ in the next $\bar{T}=30$ years, $\rho_m(\bar{T})$ from \eqref{eq:LA_rhoT}, for each $m\in\mathcal{M}$. In all cases we take $X_0=5$, $M^\ast_0=0$.}
	\label{fig:LA_VariatMu} 
\end{figure}
We examine the effect of the drift  $\mu$ and volatility $\sigma$ in \eqref{eq:gbm} on the equilibrium strategies and the macro market regime $M^\ast$. To do so, we evaluate the expected number of expansions carried out by Player~2 conditional on $M_0 =0$, $\mathbb{N}^2_0(5)$ (which always satisfies $\mathbb{N}^2_0(5)=\mathbb{N}^1_0(5)-1$). Additionally, we also compute the proportion of time that $M^\ast$ spends at each regime $m \in \{ -1, 0, 1\}$ in the next $\bar{T}=30$ years, $\rho_m(\bar{T})$:
\begin{align}\label{eq:LA_rhoT}
\rho_m(\bar{T}):=\E \left[ \frac{1}{\bar{T}}\int^{\bar{T}}_0\mathds{1}_{\{M^\ast_t=m\}}dt \Big|\, M^\ast_0 = 0 \right].
\end{align}

\cref{fig:LA_VariatMu} shows the results as we vary $\mu$ from $0.05$ to $0.15$ with fixed $\sigma=0.25$, or in complement vary $\sigma \in [0.20,0.25]$ with fixed $\mu=0.08$. As expected, higher $\mu$ increases the tendency of $(X_t)$ to go to $+\infty$ and hence enforces the dominance of Player~1; thus Player~2 expands less. A similar effect holds as $\sigma$ falls ---with less fluctuations there are fewer opportunities for P2. As a result, the overall number of switches, which can be viewed as the ``observable competition'', decreases as $\mu$ increases or $\sigma$ decreases. A related effect is observed in \cref{fig:LA_distrib_est}: the dominance of P1, $\rho_{+1}(\bar{T})$, increases as $\mu$ rises or $\sigma$ falls. For $\mu$ low ($\sigma$ high), the transition to long-run advantage takes place more slowly, so the players are more even-handed in the medium term on $[0,\bar{T})$. Note that higher volatility hurts Player~1, intensifying the medium-term competition and causing both players to expend a lot of capital on repeated expansion.
Finally, we remark that the proportion of time $M^\ast$ spends in regime 0, $\rho_0(\bar{T})$ is quite stable with respect to different combinations of $\mu$ and $\sigma$.

\section{Conclusion}\label{sec:Conc}

The above experiments illustrate our constructive approach to dynamic equilibria in repeated timing games. In particular, the discrete nature of the macro regime $M$ allows a high degree of analytic tractability, including precise quantification of its behavior in  equilibrium. The above was a strong motivation for our reliance on threshold-type equilibria which allow structural insights into the strategic interaction between the players, the short-term fluctuations in $X$, and the emerging $M^\ast$.

A natural extension would be to consider other control settings, for instance impulse control on $M$ (allowing for a continuum of regimes) or singular control. Handling such extensions must start with precise analysis of the ``building block'' local problem which in our case was the constrained optimal stopping in \eqref{eq:MS_problem}. While the latter has been extensively investigated for optimal stopping, the existing theory for alternative control settings, such as constrained singular control, remains incomplete.

While our work is rooted in concrete economic applications, we are cognizant that many realistic aspects of the model have been left out. More sophisticated forms of payoff rates and switching costs (for example capturing instantaneous supply-demand competition) would be straightforward to incorporate. Another potential extension is to allow the dynamics of $X$, i.e.~the coefficients $\mu(X_t)$ and $\sigma(X_t)$ in \eqref{eq:PF_sde} to depend on the regime $m$, which implies that the generator $\mathcal{L}$ needs to be re-indexed as $\mathcal{L}_m$. Modulo technical and notational modifications, our construction, in particular the system of equations \eqref{eq:AP_System} - \eqref{eq:AP_sysbd} and the finite-control approximations in Section~\ref{sec:SAM}, should follow straightforwardly.

A much more challenging extension would be to allow for noise in $M$ which amounts to considering a two- (or in general multi-) dimensional stochastic state process. The main challenge is that there is very limited theory about constrained optimal stopping in multiple dimensions---a control now maps into a whole curve $s^i(m)$ and it is far from clear what the best-response to such $s^j(m)$ ought to look like.

Abstractly, it remains an open question to find other feasible equilibria in the switching game. We strongly suspect that there is a general lack of equilibrium uniqueness, so the more fruitful question is to construct other tractable equilibria.

\appendix

\section{Solution of the Constrained Optimal Stopping Problem}\label{AP_Methd}
Let $F(\cdot)$ and a $G(\cdot)$ be the fundamental increasing/decreasing solutions to the ODE:
	\begin{align}\label{eq:MS_ode}
	(\mathcal{L}-r)u(x)=0, \qquad x\in\mathcal{D},
	\end{align}
where $\mathcal{L}=b(x)\frac{d}{dx}+\frac{\sigma^2(x)}{2}\frac{d^2}{dx^2}$ is the infinitesimal generator of $X$. These linearly independent solutions are positive, continuous, strictly monotone and convex and admit the representations (see \cite[vol.II, p.292]{rogers2000diffusions})
	\begin{equation}
	\E_x\left\{e^{-r\tau_{a}} \mathds{1}_{\{\tau_a<\infty\}}\right\}=\begin{cases}
	\frac{F(x)}{F(a)}, &\qquad \text{if } x\leq a,\\
	\frac{G(x)}{G(a)}, &\qquad \text{if } x\geq a.\\
	\end{cases}
	\end{equation}
We recall that for the OU process the fundamental solutions are:
\begin{align*}
F_{OU}(x):=\int^\infty_0u^{\frac{r}{\mu}-1}e^{\sqrt{\frac{2\mu}{\sigma^2}}(x-\theta)u-\frac{u^2}{2}}du,\qquad\text{and}\qquad G_{OU}(x):=\int^\infty_0u^{\frac{r}{\mu}-1}e^{-\sqrt{\frac{2\mu}{\sigma^2}}(x-\theta)u-\frac{u^2}{2}}du,
\end{align*}
and for {geometric Brownian motion} they are:
\begin{align*}
F_{GBM}(x):=x^{\eta^+},\qquad\text{and}\qquad G_{GBM}(x):=x^{\eta^-},
\end{align*}
where $\eta^+$ and $\eta^-$ are the positive and negative roots of the quadratic equation $\frac{\sigma^2}{2}\eta(\eta-1)+\mu\eta-r=0$.

Then it is well known that the solution of the constrained optimal stopping problem \eqref{eq:MS_problem} is the \textbf{smallest concave majorant} of the transformed leader payoff
\begin{align}\label{eq:MS_transH}
\hat{H}^1(y):=\frac{\hat{h}^1}{G}\circ \left(\frac{F}{G}\right)^{-1}(y),\qquad \text{and }\qquad \hat{H}^2(z):=\frac{\hat{h}^2}{F}\circ \left(\frac{G}{F} \right)^{-1}(z),
\end{align}
where $\hat{h}^i(s^j)=\max\{l^i, g^i\}(s^j)$ and $\hat{h}^i(x)=h^i(x)$ otherwise. Specifically, $\{\frac{\tilde{V}^1}{G}\}(F(x)/G(x))$ (resp. $\{\frac{\tilde{V}^2}{F}\}(G(x)/F(x))$) is the concave envelope of $\hat{H}^i(\cdot)$. To ensure that this concave envelope is of threshold-type we invoke the following definition, cf.~de Angelis et al.~\cite{FerrariMoriarty15}:

\begin{defn}\label{def:AP_Hinc}
	Let $\mathcal{H}$ be the class of real valued functions $h\in\mathcal{C}^2(\mathcal{D})$ such that $\displaystyle\limsup_{x\rightarrow\underline{d}}\left|\frac{h(x)}{G(x)}\right|=\displaystyle\limsup_{x\rightarrow\bar{d}}\left|\frac{h(x)}{F(x)}\right|=0$, and $\mathbb{E}_x\left[\int^\infty_0e^{-rt}\left|(\mathcal{L}-r)h\left(X_t\right)\right| dt\right]<\infty$ for all $x\in\mathcal{D}$. Denote by $\mathcal{H}_{\text{inc}}$ (resp.~$\mathcal{H}_{\text{dec}}$) the set of all $h\in\mathcal{H}$ such that $x \mapsto (\mathcal{L}-r)h(x)$ is strictly \textit{positive} (resp. \textit{negative}) on $(\underline{d},x_h)$ and strictly \textit{negative} (resp. \textit{positive}) on $(x_h,\bar{d})$ for some $x_h\in\mathcal{D}$.
\end{defn}

When the leader payoff $h^1$ (resp. $h^2$) is in the class $\mathcal{H}_{\text{inc}}$~(resp. $\mathcal{H}_{\text{dec}}$) and \cref{asp:MS_asp} \ref{asp:MS_noPreemp} holds it follows that the transformed $\hat{H}^i$ is convex and then concave. Consequently, its smallest concave majorant is a straight line which is tangent to $\hat{H}^i$ at a unique point, and then coincides with $\hat{H}^i$.

This construction reduces to determining the tangency point of $\hat{H}^i$, corresponding to the transformed threshold $s^i(s^j)$.
Moreover, the coefficient $\widetilde{\omega}^i$ in \cref{pro:MS_OSP} corresponds to the slope of the above straight line segment and $\widetilde{\nu}^i$ corresponds to the $y$-intercept. From the fact that being a follower is assumed to sub-optimal at $L_i$, it follows that $\widetilde{\omega}^1,\widetilde{\nu}^2\geq0$ and $\widetilde{\nu}^1, \widetilde{\omega}^2\leq0$. We refer the reader to \cite{aid2017capacity} for further details.

The case where \cref{asp:MS_asp} \ref{asp:MS_noPreemp} is violated leads to a unique preemptive equilibrium in~\eqref{eq:MS_problem} see~\cite[Section 3]{aid2017capacity} for detailed discussion.

\section{Proof of \cref{thm:DA_equi}}\label{AP_equi}

We first state the full system of nonlinear equations characterizing the thresholds $\bm{s}^{i,\ast}$ and coefficients $\omega^i,\nu^i$ in \eqref{eq:DA_vFunc}. First for each $m\notin\{\underline{m},\overline{m}\}$ there are 6 equations:
\begin{subequations}\label{eq:AP_System}
	\begin{align}
	\begin{cases}\label{eq:AP_System1}
	\left\{V^1_{m+1}-D^1_m-K^1_m\right\}(s^{1,\ast}_m)\cdot W(s^{1,\ast}_m,s^{2,\ast}_m)-\left\{V^1_{m-1}-D^1_m\right\}(s^{2,\ast}_m)\cdot W(s^{1,\ast}_m,s^{1,\ast}_m) & \\
	\qquad\qquad\qquad\qquad\qquad\qquad\qquad\qquad-\left\{V^1_{m+1}-D^1_m-K^1_m\right\}'(s^{1,\ast}_m)\cdot \mathcal{W}(s^{1,\ast}_m,s^{2,\ast}_m) & =0,\\
	\left\{V^1_{m+1}-D^1_m-K^1_m\right\}(s^{1,\ast}_m)\cdot G(s^{2,\ast}_m)-\left\{V^1_{m-1}-D^1_m\right\}(s^{2,\ast}_m)\cdot G(s^{1,\ast}_m)-\omega^1_m\cdot \mathcal{W}(s^{1,\ast}_m,s^{2,\ast}_m) & =0,\\
	\left\{V^1_{m-1}-D^1_m\right\}(s^{2,\ast}_m)\cdot F(s^{1,\ast}_m)-\left\{V^1_{m+1}-D^1_m-K^1_m\right\}(s^{1,\ast}_m)\cdot F(s^{2,\ast}_m)-\nu^1_m\cdot \mathcal{W}(s^{1,\ast}_m,s^{2,\ast}_m) & =0,
	\end{cases}\\
	\begin{cases}
	\left\{V^2_{m-1}-D^2_m-K^2_m\right\}(s^{2,\ast}_m)\cdot W(s^{2,\ast}_m,s^{1,\ast}_m)-\left\{V^2_{m+1}-D^2_m\right\}(s^{1,\ast}_m)\cdot W(s^{2,\ast}_m,s^{2,\ast}_m) & \\
	\qquad\qquad\qquad\qquad\qquad\qquad\qquad\qquad-\left\{V^2_{m-1}-D^2_m-K^2_m\right\}'(s^{2,\ast}_m)\cdot \mathcal{W}(s^{1,\ast}_m,s^{2,\ast}_m) & =0,\\
	\left\{V^2_{m-1}-D^2_m-K^2_m\right\}(s^{2,\ast}_m)\cdot G(s^{1,\ast}_m)-\left\{V^2_{m+1}-D^2_m\right\}(s^{1,\ast}_m)\cdot G(s^{2,\ast}_m)-\omega^2_m\cdot \mathcal{W}(s^{2,\ast}_m,s^{1,\ast}_m) & =0,\\
	\left\{V^2_{m+1}-D^2_m\right\}(s^{1,\ast}_m)\cdot F(s^{2,\ast}_m)-\left\{V^2_{m-1}-D^2_m-K^2_m\right\}(s^{2,\ast}_m)\cdot F(s^{1,\ast}_m)-\nu^2_m\cdot \mathcal{W}(s^{2,\ast}_m,s^{1,\ast}_m)& =0,
	\end{cases}
	\end{align}
\end{subequations}
where $F$ and $G$ are the solutions to the ODE \eqref{eq:MS_ode}, and $W(\cdot,\cdot), \mathcal{W}(\cdot,\cdot)$ are from \eqref{eq:MS_Wronsk}.
In boundary regimes $m=\underline{m}$ and $m=\overline{m}$ we have the following systems of 3 equations:
\begin{subequations}\label{eq:AP_sysbd}
	\begin{align}
	&\begin{cases}\label{eq:AP_sysbd1}
	s^{1,\ast}_{\overline{m}} = \overline{d},  \omega^1_{\overline{m}}=0,\, \nu^1_{\underline{m}}=0,\\
	\left\{V^1_{\overline{m}-1}-D^1_{\overline{m}}\right\}(s^{2,\ast}_{\overline{m}})-\nu^1_{\overline{m}}\cdot G(s^{2,\ast}_{\overline{m}})=0,\\
	\left\{V^1_{\underline{m}+1}-D^1_{\underline{m}}-K^1_{\underline{m}}\right\}(s^{1,\ast}_{\underline{m}})\cdot F'(s^{1,\ast}_{\underline{m}})-\left\{V^1_{\underline{m}+1}-D^1_{\underline{m}}-K^1_{\underline{m}}\right\}'(s^{1,\ast}_{\underline{m}})\cdot F(s^{1,\ast}_{\underline{m}})=0,\\
	\left\{V^1_{\underline{m}+1}-D^1_{\underline{m}}-K^1_{\underline{m}}\right\}(s^{1,\ast}_{\underline{m}})-\omega^1_{\underline{m}}\cdot F(s^{1,\ast}_{\underline{m}})=0,
	\end{cases} \\
	&\begin{cases}
	s^{2,\ast}_{\underline{m}} = \underline{d},\,\omega^2_{\overline{m}}=0,\,\nu^2_{\underline{m}}=0,\\
	\left\{V^2_{\overline{m}-1}-D^2_{\overline{m}}-K^2_{\overline{m}}\right\}(s^{2,\ast}_{\overline{m}})\cdot G'(s^{2,\ast}_{\overline{m}})-\left\{V^2_{\overline{m}+1}-D^2_{\overline{m}}-K^2_{\overline{m}}\right\}'(s^{2,\ast}_{\overline{m}})\cdot G(s^{2,\ast}_{\overline{m}})=0,\\
	\left\{V^2_{\overline{m}-1}-D^2_{\overline{m}}-K^2_{\overline{m}}\right\}(s^{2,\ast}_{\overline{m}})-\nu^2_{\overline{m}}\cdot G(s^{2,\ast}_{\overline{m}})=0,\\
		\left\{V^2_{\underline{m}+1}-D^2_{\underline{m}}\right\}(s^{1,\ast}_{\underline{m}})-\omega^2_{\underline{m}}\cdot F(s^{1,\ast}_{\underline{m}})=0.
	\end{cases}\label{eq:boundary-vi-cases}
	\end{align}
\end{subequations}

\begin{proof}[Proof of \cref{thm:DA_equi}]
	To begin with, we argue that by construction of $V^i$'s in \eqref{eq:DA_vFunc} we have:
	\begin{enumerate}
		\item  $V^i_m\in C^2\!\left(\mathcal{D}\setminus(\bm{s}^i\cup\bm{s}^j)\right)\cap C^1\left(\mathcal{D}\setminus\bm{s}^j\right)\cap C\left(\mathcal{D}\right)$, for $\forall m\in\mathcal{M}$, $i\in\{1,2\}$, $i\neq j$;

\item $V^i_m$ is at most linear growth, i.e.
\begin{align}\label{eq:VTgrowth}
|V^i_m(x)|\leq C(1+|x|), \text{ for } \forall x\in\mathcal{D};
\end{align}
		\item  $V^i$'s satisfy the following system of variational inequalities (VIs) for $m<\overline{m}$:
		\begin{subequations}\label{eq:VI1}
			\begin{align}
			&V^1_{m+1}-K^1_m-V^1_m\leq0, &&\text{in }\mathcal{D},\label{eq:VI1a}\\	
			&V^1_{m-1}-V^1_m = 0, &&\text{in }\Gamma^{2,\ast}_m,\label{eq:VI1b}\\		
			&V^2_{m+1}-V^2_m = 0, &&\text{in }\Gamma^{1,\ast}_m,\label{eq:VI1c}\\
			&\max\left\{\left(\mathcal{L}-r\right)V^1_m+\pi^1_m, V^1_{m+1}-K^1_m-V^1_m\right\}=0, &&\text{in }\mathcal{D}\setminus\Gamma^{2,\ast}_m,\label{eq:VI1d}
			\end{align}
		\end{subequations}
		and for $m>\underline{m}$,
		\begin{subequations}\label{eq:VI2}
			\begin{align}
			&V^2_{m-1}-K^2_m-V^2_m\leq0, &&\text{in }\mathcal{D},\\
			&V^1_{m-1}-V^1_m = 0, &&\text{in }\Gamma^{2,\ast}_m,\\	
			&V^2_{m+1}-V^2_m = 0, &&\text{in }\Gamma^{1,\ast}_m,\\
			&\max\left\{\left(\mathcal{L}-r\right)V^2_m+\pi^2_m, V^2_{m-1}-K^2_m-V^2_m\right\}=0, &&\text{in }\mathcal{D}\setminus\Gamma^{1,\ast}_m.
			\end{align}
		\end{subequations}
	\end{enumerate}
	The smoothness of $V^i$ follows directly from the regularity of $F(\cdot)$ and $G(\cdot)$ and the piecewise construction. The second statement follows from the linear growth assumption imposed on $D^i$'s in \eqref{eq:MS_dcf} and signs of coefficients ${\omega}^i, {\nu}^i$. Note that this is a natural property of correct equilibrium payoffs since the best-response game payoff of player $i$ satisfies
	\begin{align}
		\min_{m\in\mathcal{M}}D^i_m(x)\leq \widetilde{V}^i_m(x;\ba^j) \leq \max_{m\in\mathcal{M}}D^i_m(x),
	\end{align}
	and a MNE is characterized as a fixed-point of best-responses. The key is the last assertion; we show \eqref{eq:VI1}, with \eqref{eq:VI2}  then following analogously. Comparing the system \eqref{eq:AP_System} for fixed $m$ with \eqref{eq:MS_Threshold} and \eqref{eq:MS_Coefficient}, one can see that $V^1_m-D^1_m$ is indeed the solution to the optimal stopping problem \eqref{eq:MS_problem} with
	\begin{align*}
	h^1_m(x)&:=V^1_{m+1}(x)-D^1_m(x)-K^1_m(x),\\
	l^1_m(x)&:=V^1_{m-1}(x)-D^1_m(x),
	\end{align*}
	which in turn brings the restriction on signs of $\bm{\omega}^i, \bm{\nu}^i$. Taking P1 as an example, $\omega^1_m$ corresponds to the slope of the straight line segment of its transformed smallest concave majorant which ought be positive for $m<\overline{m}$ and equal to zero  for $m=\overline{m}$. Similarly, $\nu^1_m$ corresponds to the $y$-intercept of that line segment which ought to be negative for $m > \underline{m}$ and equal to zero for $m=\underline{m}$. Furthermore, the signs of the derivatives of $F(\cdot), G(\cdot)$ imply that  $V^1_m-D^1_m$ is increasing. The assumption  $V^1_{m-1}(s^{2,\ast})\geq V^1_{m+1}(s^{2,\ast}_m)-K^1_m(s^{2,\ast}_m)$, for $m>\underline{m}$ plus the smallest concave majorant characterization then yields
	\begin{align*}
	& V^1_m-D^1_m = l^1_m \geq h^1_m = V^1_{m+1}-D^1_m-K^1_m, && \text{in }\Gamma^{2,\ast}_m,\\
	& V^1_m-D^1_m = h^1_m = V^1_{m+1}-D^1_m-K^1_m, && \text{in }\Gamma^{1,\ast}_m,\\
	& V^1_m-D^1_m \geq h^1_m = V^1_{m+1}-D^1_m-K^1_m, && \text{in }\mathcal{D}\setminus\left(\Gamma^{1,\ast}_m\cup\Gamma^{2,\ast}_m\right),
	\end{align*}
	which shows \eqref{eq:VI1a}. \eqref{eq:VI1b} and \eqref{eq:VI1c} are obtained directly from the construction of $V^i$'s, reflecting the payoff in $x$-states where the rival switches immediately. Lastly, to check \eqref{eq:VI1d}, recall that the discounted cash flows in \eqref{eq:MS_dcf} satisfy $\left(\mathcal{L}-r\right)D^i_m = -\pi^i_m$,  and by their definition $\left(\mathcal{L}-r\right)F=\left(\mathcal{L}-r\right)G=0$.
For $x \in \mathcal{D}\setminus (\Gamma^{2,\ast}_m \cup \Gamma^{1,\ast}_m)$ (``the no-action region'') we have by \eqref{eq:DA_vFunc} that $V^1_m(x) = D^1_{m}(x)+\omega^1_{m}F(x)+\nu^1_{m}G(x)$, so applying the operator $(\mathcal{L}-r)$ we get
 	\begin{align*}
	\left(\mathcal{L}-r\right)V^1_m+\pi^1_m = \left(\mathcal{L}-r\right)D^1_m+\pi^1_m = -\pi^1_m+\pi^1_m=0, \qquad x \in \mathcal{D}\setminus (\Gamma^{2,\ast}_m \cup \Gamma^{1,\ast}_m).
	\end{align*}
For $x \in \Gamma^{1,\ast}_m\setminus \Gamma^{1,\ast}_{m+1}$, we have $V^1_m(x) = V^1_{m+1}(x) - K^1_{m}(x) = D^1_{m+1}(x)+\omega^1_{m+1}F(x)+\nu^1_{m+1}G(x) - K^1_{m}$ so that
	\begin{align}
	\left(\mathcal{L}-r\right)V^1_m+\pi^1_m &= \left(\mathcal{L}-r\right)\left(V^1_{m+1}-K^1_m\right)+\pi^1_m\nonumber\\
	& = \left(\mathcal{L}-r\right)D^1_{m+1} +\left(\mathcal{L}-r\right)\left(-K^1_m\right)+\pi^1_m\label{eq:d1}\\
	& = \left(\mathcal{L}-r\right)\left(D^1_{m+1}-D^1_m-K^1_m\right)   < 0,\label{eq:d2}
	\end{align}
	where the last inequality  \eqref{eq:d2} is due to $D^1_{m+1}-D^1_m-K^1_m\in\mathcal{H}_{\text{inc}}$. Similar arguments apply to $x \in \Gamma^{1,\ast}_{m+1} \setminus \Gamma^{1,\ast}_{m+2}$ where two simultaneous switches by P1 will take place; by induction we conclude that $\left(\mathcal{L}-r\right)V^1_m+\pi^1_m < 0$ for $x \in \Gamma^{1,\ast}_m$, establishing \eqref{eq:VI1d}.
%

	We now prove $\left(\bm{s}^{1,\ast},\bm{s}^{2,\ast}\right)$ is a Nash equilibrium. To do so, we first consider the point of view of P1, letting  $\ba^1=\{\tau^1(n):n \geq 1 \}$ be her arbitrary strategy satisfying $(\ba^1,\bm{s}^{2,\ast})\in\mathcal{A}$, and $(\ts_n)_{n\geq 0}$ be the sequence of resulting switching times defined in \eqref{eq:PF_construct}, with $X_0=x$, $\tM_0=m$. As a first step, we use induction to establish that
\begin{align}
V^1_m(x) \geq  \mathbb{E}\Big[\int^{\ts_n}_0e^{-rt}\pi^1(X^x_t, \tM_{\tn(t)})dt-\sum^n_{k=1}\mathds{1}_{\{\tp_k=1\}}e^{-r\ts_k}\cdot K^1\big(X_{\ts_k}, \tM_{k-1}\big)+ e^{-r\ts_n}V^1_{\tM_n}(X^x_{\ts_n})\Big] \quad \forall n \ge 1.\label{eq:sigman-induction}
\end{align}

For $n=1$, since $\ts_1=\tau^1(1)\wedge\tau^{2,\ast}_m$, applying It\^{o}'s formula to the process $e^{-rt}V^1_{m}(X^x_t)$ over the interval $[0, \ts_1]$ and taking expectations yields
	\begin{subequations}
	\begin{align}
	V^1_m(x) & = \mathbb{E}\Big[-\int^{\ts_1}_0e^{-rt}\left(\mathcal{L}-r\right)V^1_{m}\left(X^x_t\right)dt+e^{-r\ts_1}V^1_m(X^x_{\ts_1})\Big]\nonumber\\
	&\geq  \mathbb{E}\Big[\int^{\ts_1}_0e^{-rt}\pi^1(X^x_t, \tM_0)dt+e^{-r\ts_1}V^1_m(X^x_{\ts_1})\Big]\label{eq:VT1}\\
	&\geq \mathbb{E}\Big[\int^{\ts_1}_0e^{-rt}\pi^1(X^x_t, \tM_0)dt + e^{-r\ts_1} \{ -K^1(X^x_{\ts_1}, \tM_0)\cdot \mathds{1}_{\{\ts_1=\tau^1(1)\}}+ V^1_{\tM_1}(X^x_{\ts_1})\} \Big]\label{eq:VT2},
	\end{align}
\end{subequations}
	where the inequality \eqref{eq:VT1} follows from \eqref{eq:VI1d} and the fact that $\ts_1\leq\tau^{2,\ast}_m$, and the inequality \eqref{eq:VT2} is due to \eqref{eq:VI1a} and \eqref{eq:VI1b}:
	\begin{align}
	\mathbb{E}\Big[V^1_m(X^x_{\ts_1})\Big]
	& \geq \mathbb{E}\Big[\mathds{1}_{\{\ts_1=\tau^1(1)\}}\Big\{V^1_{m+1}(X^x_{\ts_1})-K^1(X^x_{\ts_1}, \tM_0)\Big\}+\mathds{1}_{\{\ts_1=\tau^{2,\ast}_{m}\}}V^1_{m-1}(X^x_{\ts_1})\Big]\label{eq:VT3}.
	\end{align}
	
Next we show \eqref{eq:sigman-induction} for $n=2$.
	By construction, we have $\ts_2 = \tau^1(2)\wedge(\ts_1+\tau^{2,\ast}_{\tM_1})$. Consider the second-round sub-game started at initial state $X^x_{\ts_1}$; applying It\^{o}'s formula to the process $e^{-rt}V^1_{m+1}\big(X^{X^x_{\ts_1}}_t\big)$ over the interval $[0, \ts_2-\ts_1]$ and taking expectation conditional on $\tf^{(2)}_{\ts_1}$, cf.~\eqref{eq:PF_constructF}, we obtain
	\begin{align}
	V^1_{m+1}(X^x_{\ts_1})
	& \geq\mathbb{E}\Big[\int^{\ts_2-\ts_1}_0e^{-rt}\pi^1\big(X^{X^x_{\ts_1}}_t, m+1\big)dt+e^{-r(\ts_2-\ts_1)}V^1_{\tM_2}\big(X^x_{\ts_2}\big)\nonumber\\
	&\qquad\qquad\qquad\qquad\qquad\qquad -\mathds{1}_{\{\ts_2=\tau^1(2)\}}e^{-r(\ts_2-\ts_1)}\cdot K^1(X^x_{\ts_2}, m+1)\Big|\,\tf^{(2)}_{\ts_1}\Big],\nonumber
	\end{align}
analogously to \eqref{eq:VT2} and replacing $X^{X^x_{\ts_1}}_{\ts_2-\ts_1}$ by $X^x_{\ts_2}$ based on the strong Markov property of $(X_t)$. Furthermore, using $\int^{\ts_2}_{\ts_1}e^{-rt}\pi^1(X^x_t, m+1)dt = e^{-r\ts_1}\int^{\ts_2-\ts_1}_0e^{-rt}\pi^1\big(X^{X^x_{\ts_1}}_s, m+1\big)ds$
  we have
\begin{align}
\mathbb{E}\Big[\mathds{1}_{\{\ts_1=\tau^1(1)\}}e^{-r\ts_1}V^1_{m+1}(X^x_{\ts_1})\Big]
\geq\,& \mathbb{E}\Bigg[\mathds{1}_{\{\ts_1=\tau^1(1)\}}\cdot\Big[\int^{\ts_2}_{\ts_1}e^{-rt}\pi^1\big(X^x_t, m+1\big)dt+e^{-r\ts_2}V^1_{\tM_2}\big(X^x_{\ts_2}\big)\Big]\nonumber\\
&\qquad\qquad\qquad-\mathds{1}_{\{\ts_2=\tau^1(2)\}}\mathds{1}_{\{\ts_1=\tau^1(1)\}}e^{-r\ts_2}\cdot K^1(X^x_{\ts_2}, m+1)\Bigg].\label{eq:VT4}
\end{align}
Similarly, we have
	\begin{align}
\mathbb{E}\Big[\mathds{1}_{\{\ts_1=\tau^{2,\ast}_m\}}e^{-r\ts_1}V^1_{m-1}(X^x_{\ts_1})\Big] & \geq \mathbb{E}\Bigg[\mathds{1}_{\{\ts_1=\tau^{2,\ast}_m\}}\cdot\Big[\int^{\ts_2}_{\ts_1}e^{-rt}\pi^1\big(X^x_t, m-1\big)dt+e^{-r\ts_2}V^1_{\tM_2}\big(X^x_{\ts_2}\big)\Big]\qquad\nonumber\\
&\qquad\qquad\qquad\qquad-\mathds{1}_{\{\ts_2=\tau^1(2)\}}\mathds{1}_{\{\ts_1=\tau^{2,\ast}_{m}\}}e^{-r\ts_2}\cdot K^1(X^x_{\ts_2}, m-1)\Bigg].\label{eq:VT5}
\end{align}
Substituting \eqref{eq:VT4} - \eqref{eq:VT5} into  \eqref{eq:VT2}, we obtain
\begin{align}
V^1_m(x) & \geq \mathbb{E}\Big[\int^{\ts_1}_0e^{-rt}\pi^1(X^x_t, \tM_0)dt +\int^{\ts_2}_{\ts_1}e^{-rt}\pi^1(X^x_t, m+1)dt \nonumber\\
&\qquad + e^{-r\ts_2}V^1_{\tM_2}(X^x_{\ts_2}) - \mathds{1}_{\{\ts_1=\tau^1(1)\}}e^{-r\ts_1}K^1(X^x_{\ts_1}, \tM_0) \nonumber\\
&\qquad-\mathds{1}_{\{\ts_2=\tau^1(2)\}}\mathds{1}_{\{\ts_1=\tau^1(1)\}}e^{-r\ts_2}\cdot K^1(X^x_{\ts_2}, m+1)
-\mathds{1}_{\{\ts_2=\tau^1(2)\}}\mathds{1}_{\{\ts_1=\tau^{2,\ast}_{m}\}}e^{-r\ts_2}\cdot K^1(X^x_{\ts_2}, m-1)\nonumber\\
& = \mathbb{E}\Big[\int^{\ts_2}_0e^{-rt}\pi^1(X^x_t, \tM_{\tn(t)})dt-\sum^2_{k=1}\mathds{1}_{\{\tp_k=1\}}e^{-r\ts_k}\cdot K^1\big(X_{\ts_k}, \tM_{k-1}\big)+ e^{-r\ts_2}V^1_{\tM_2}(X^x_{\ts_2})\Big].\nonumber
\end{align}

Iterating this argument for $n=3,\ldots,$ establishes \eqref{eq:sigman-induction}. Let us remark that the above works without any modifications in the boundary regimes $ m \in \{\underline{m}, \overline{m}\}$, where $\tau^1(n)$ or $\tau^{2,\ast}_m$ are set to be infinite. Since $V^1_m$ is at most of linear growth from \eqref{eq:VTgrowth} and admissibility of $(\ba^1, \bm{s}^{2,\ast})$ requires $\lim_{n\rightarrow\infty}\ts_n = +\infty$, dominated convergence theorem implies
\begin{align}\label{eq:NEineq}
V^1_m(x) \geq \mathbb{E}\Big[\int^{\infty}_0e^{-rt}\pi^1(X^x_t, \tM_{\tn(t)})dt-\sum_{k=1}\mathds{1}_{\{\tp_k=1\}}e^{-r\ts_k}\cdot K^1\big(X_{\ts_k}, \tM_{k-1}\big)\Big]=J^1_m(x;\ba^1, \bm{s}^{2,\ast}).
\end{align}
Similarly for P2 we obtain that
\begin{align*}
V^2_m(x)\geq J^2_m(x;\bm{s}^{1,\ast},\ba^2), \text{ for }\forall (\bm{s}^{1,\ast}, \ba^2)\in\mathcal{A}.
\end{align*}
{Last but not least, one can verify that replacing $\ba^1$ by $\bm{s}^{1,\ast}$ in above argument leads to $\ts_1=\tau^{1,\ast}_{m}\wedge\tau^{2,\ast}_{m}$ so that $\left(\mathcal{L}-r\right)V^1_{m}\left(X^x_t\right)=-\pi^1(X^x_t, \tM_0)$ on $[0, \ts_1)$ and $V^1_{m}(X^x_t)=V^1_{m+1}(X^x_t)-K^1_{m}(X^x_t)$ at $\ts_1 =\tau^{1,\ast}_{m}$ and $V^1_{m}(X^x_t)=V^1_{m-1}(X^x_t)$ at $\ts_1 =\tau^{2,\ast}_{m}$. These turn inequalities in \eqref{eq:VT1} and \eqref{eq:VT2} into equalities, and inductively yield
\begin{align*}
V^1_m(x) = J^1_m(x;\bm{s}^{1,\ast}, \bm{s}^{2,\ast}),
\end{align*}
which, combining with \eqref{eq:NEineq},  completes the proof.}
\end{proof}

\section{Proof of \cref{pro:BTS_conv}}\label{AP_BestRgame payoff}
Let $X_0=x\in\mathcal{D}$ , $M_0=m\in\mathcal{M}$, and fix $\bm{s}^j$.  The best-response of player $i$ with $N^i\geq 1$ controls is
\begin{align}
\widetilde{V}^{i,(N^i)}_m(x\,;\bm{s}^j) = \sup_{\ba^{i,(N^i)}\in\mathcal{A}^{i,(N^i)}}J^i_m(x\,;\ba^{i,(N^i)},\bm{s}^j),
\end{align}
where
\begin{align}
\mathcal{A}^{i, (N^i)}:= \left\{(\ba^i, \bm{s}^j)\in\mathcal{A}: \tau^i(n)=+\infty, n > \tn(i, N^i) \right\},
\end{align}
with $\tn(i,N^i)$ defined in \eqref{eq:PF_constructIk} denotes the round at which player $i$ exercises her $N^i$-th switch. Since $\mathcal{A}^{i, (N^i)}\subseteq\mathcal{A}^{i, (N^i+1)}$ we have that $N^i \mapsto \widetilde{V}^{i, (N^i)}_m(x\,;\bm{s}^j)$ is \textit{non-decreasing}. Moreover, since $\widetilde{V}^{i,(N^i)}_m(x\,;\bm{s}^j)$ is bounded from above by $\max_m D^i_m(x)$, $\lim_{n\rightarrow\infty}\widetilde{V}^{i,(N^i)}_m(x\,;\bm{s}^j)$ is well-defined. It remains to show that this limit is $\widetilde{V}^i_m(x\,;\bm{s}^j)$.

Because $\mathcal{A}^{i, (N^i)}\subseteq\{\ba^i:(\ba^i,\bm{s}^j)\in\mathcal{A}\}$, we trivially obtain
\begin{align}\label{eq:AP4}
\lim_{N^i\rightarrow\infty}\widetilde{V}^{i, (N^i)}_m(x\,;\bm{s}^j)\leq\widetilde{V}^i_m(x\,;\bm{s}^j).
\end{align}
To obtain the opposite inequality, for any $\varepsilon>0$, let $\bm{\alpha}^{i}_\varepsilon:=\{\tau^i_{\varepsilon}(n): n\geq1\}$ (which depends on $x$) be a $\varepsilon$-optimal strategy satisfying $(\ba^i_\varepsilon,\bm{s}^j)\in\mathcal{A}$ and
\begin{align}\label{eq:AP3}
J^i_m(x\,;\bm{\alpha}^i_\varepsilon,\bm{s}^j) \geq \widetilde{V}^i_m(x\,;\bm{s}^j)-\varepsilon.
\end{align}
Now for a fixed $N^i\geq1$ we define the respective truncated $N^i$-finite strategy $\ba^{i,(N^i)}_\varepsilon:=\{\tau^{i,(N^i)}_{\varepsilon}(n):n\geq1\}$ as
\begin{align}
\tau^{i,(N^i)}_{\varepsilon}(n)=\begin{cases}
\tau^i_{\varepsilon}(n), & n\leq \tn(i, N^i)\\
+\infty,&o.w.
\end{cases}
\end{align}
Thus, the truncated strategy stops switching completely after the first $N^i$ switches. Denote by $M^{(N^i)}_t$ the resulting macro regime and by $(\sigma^{i, (N^i)}_k)_{k\leq N^i}$ the sequence of switching times of player $i$, cf.~\eqref{eq:PF_constructIk}, based on $(\bm{\alpha}^{i,(N^i)}_\varepsilon, \bm{s}^j)$, which we compare against the corresponding $M^{(\infty)}_t$ and $(\sigma^{i,(\infty)}_k)_{k\geq1}$ based on the non-truncated $(\bm{\alpha}^{i}_\varepsilon, \bm{s}^j)$. By the construction of the truncation,
\begin{align*}
\sigma^{i, (N^i)}_k=\sigma^i_k, \;\text{ for }k\leq N^i,\qquad M^{(N^i)}_t = M_t, \; \text{ for }t\leq\sigma^{i}_{N^i}, 
\end{align*}
and the two cashflows completely match up to $\sigma^{i,(\infty)}_{N^i}$. In the truncated version, thereafter only the other player $i$ applies her controls.
Since $\sigma^{i,(\infty)}_{N^i}\rightarrow\infty$ as $N^i\rightarrow\infty$ from admissibility of $\ba^i_\varepsilon$, it follows  that there exists $N^\varepsilon>1$ s.t.~for $\forall N >N^\varepsilon$
\begin{subequations}\label{eq:AP1}
\begin{align}
\mathbb{E}_{x,m}\left[\int^\infty_{\sigma^i_{N}}e^{-rt}|\pi^i\bigl(X_t,M^{(\infty)}_t\bigr)| \, dt \right] < \varepsilon; \\
\mathbb{E}_{x,m}\left[\int^\infty_{\sigma^i_{N}}e^{-rt}|\pi^i(X_t,M^{(N)}_t) |\, dt \right] < \varepsilon; \\
\mathbb{E}_{x,m}\left[\sum_{k=N+1}^\infty e^{-r\sigma^{i,(\infty)}_k} K^i\left(X_{\sigma^{i,(\infty)}_k}, \tM^{(\infty)}_{\tn(i, k)-1}\right) \cdot \right] < \varepsilon.
\end{align}
\end{subequations}
For the second bound we use the fact that $M$ has a finite state space so that $|\pi^i(X_t,M^{(N)}_t)| \le \max_m |\pi^i(X_t,m)|$ which still satisfies the growth condition. Using \eqref{eq:AP1} and \eqref{eq:PF_game payoff} we have for $N > N^\varepsilon$
\begin{multline}\label{eq:AP10} \big| J^i_m(x\,;\bm{\alpha}^{i, (N)}_\varepsilon,\bm{s}^j) - J^i_m(x\,;\bm{\alpha}^{i}_\varepsilon,\bm{s}^j) \big| \le
\E_{x,m}\Bigl[ \int^\infty_{\sigma^{i,\infty}_{N}}e^{-rt} \big(|\pi^i\left(X_t,M^{(\infty)}_t\right)| + |\pi^i(X_t,M^{(N)}_t)|\big) dt  \\
+ \sum_{k=N+1}^\infty e^{-r\sigma^{i,(\infty)}_k} K^i\left(X_{\sigma^{i,(\infty)}_k}, \tM^{(\infty)}_{\tn(i,k)-1}\right) \Bigr] \le 3\varepsilon.
\end{multline}
By Fatou's lemma and \eqref{eq:AP10} we obtain
\begin{align}
\liminf_{N\rightarrow\infty}J^i_m(x\,;\bm{\alpha}^{i, (N)}_\varepsilon,\bm{s}^j) &=\liminf_{N\rightarrow\infty}\E_{x,m}\left[\int_0^{\infty}e^{-rt}\! \pi^i(X_t,M^{(N)}_t)dt-\sum_{k=1}^{N} K^i\left(X_{\sigma^{i,(N)}_k}, \tM^{(N)}_{\tn(i, k)-1}\right) \cdot e^{-r\sigma^{i,(N)}_k}\right]\nonumber\\
&\geq J^i_m(x\,;\bm{\alpha}^{i}_\varepsilon,\bm{s}^j)-3\varepsilon.
\end{align}
In turn, from \eqref{eq:AP3}, we get
\begin{align}
\liminf_{N\rightarrow\infty}\widetilde{V}^{i,(N)}_m(x\,;\bm{s}^j)\geq\liminf_{N\rightarrow\infty}J^i_m(x\,;\bm{\alpha}^{i, (N)}_\varepsilon,\bm{s}^j)\geq \widetilde{V}^i_m(x\,;\bm{s}^j)-4\varepsilon,
\end{align}
which along with \eqref{eq:AP4} and letting $\varepsilon\downarrow0$ completes the proof.

\section{Dynamics of $\tM^\ast$ in Threshold-type Equilibrium}\label{ap:DofM}
In this Appendix, we present computational details related to the macro market equilibrium described in Section~\ref{sec:SMG_summary}. While the computations are largely classical, we state them for the completeness and the reader's convenience. For ease of presentation we consider the case where $s^{i,\ast}_m$'s are in strictly ascending/descending order in terms of $m$, so that all transitions of $M^\ast$ are by $\pm 1$. 

\subsection{Transition probabilities of $\cM^\ast$ in interior states}\label{AP_transP}
Conditional on $\cM^\ast_{n-1}\in\{m^-,m^+\}$ and $m\notin\partial\mathcal{M}$, we have that ($\tilde{X}^{(n)}$ being defined in \eqref{eq:STS_constructX})
\begin{align}
\cM^\ast_{n} =\begin{cases}
(m+1)^+,&\text{ if $\tilde{X}^{(n)}_t$ hits $s^{1,\ast}_m$ before $s^{2,\ast}_m$},\\
(m-1)^-,&\text{ if $\tilde{X}^{(n)}_t$ hits $s^{2,\ast}_m$ before $s^{1,\ast}_m$},
\end{cases}
\end{align}
with the starting position $\tilde{X}^{(n)}_0 = s^{1,\ast}_{m-1}$ if $\cM^\ast_n = m^+$ and $\tilde{X}^{(n)}_0 = s^{2,\ast}_{m+1}$ if $\cM^\ast_n = m^-$. Let us use $X^x$ to denote a generic copy of $X$ started at $X_0 = x$ and consider the two-sided passage times
\begin{align*}
\tau(x;a,b):=\inf\{t\geq0\,:\,X^x_t\leq a \text{ or } X^x_t\geq b\}, \qquad (a, b)\supset x.
\end{align*}
Thus we have:
\begin{equation}\label{eq:SMG_transP}
\begin{split}
\bm{P}_{m^+, (m+1)^+}=\mathbb{P}\big[X^{s^{1,\ast}_{m-1}}_{\tau(s^{1,\ast}_{m-1};s^{2,\ast}_m,s^{1,\ast}_m)}=s^{1,\ast}_m\big], \qquad \bm{P}_{m^+, (m-1)^-}=\mathbb{P}\big[X^{s^{1,\ast}_{m-1}}_{\tau(s^{1,\ast}_{m-1};s^{2,\ast}_m,s^{1,\ast}_m)}=s^{2,\ast}_m\big],\\
\bm{P}_{m^-, (m+1)^+}=\mathbb{P}\big[X^{s^{2,\ast}_{m+1}}_{\tau(s^{2,\ast}_{m+1};s^{2,\ast}_m,s^{1,\ast}_m)}=s^{1,\ast}_m\big], \qquad \bm{P}_{m^-, (m-1)^-}=\mathbb{P}\big[X^{s^{2,\ast}_{m+1}}_{\tau(s^{2,\ast}_{m+1};s^{2,\ast}_m,s^{1,\ast}_m)}=s^{2,\ast}_m\big].\\
\end{split}
\end{equation}
In turn we recall that \eqref{eq:SMG_transP} can evaluated via the scale function $S(\cdot)$ of $(X_t)$ (see Ch~VII.3 of \cite{revuz2013continuous}):
\begin{align}	\mathbb{P}\big[X^x_{\tau(x;a,b)}=b\big]=\frac{S(x)-S(a)}{S(b)-S(a)},\qquad\mathbb{P}\big[X^x_{\tau(x;a,b)}=a\big]=\frac{S(b)-S(x)}{S(b)-S(a)}.
\end{align}
Recall that $S$ is the continuous,  increasing general solution to the ODE $\mathcal{L}S=0$ that is available in closed-form for linear diffusions.

\textit{Ornstein-Uhlenbeck Process}. The scale function $S(\cdot)$ solves
\begin{align*}
\mu(\theta-x)S'(x)+\frac{1}{2}\sigma^2S''(x)=0, \qquad \Rightarrow \quad
S_{OU}(x)=\int_{-\infty}^{x}e^{\frac{\mu}{\sigma^2}(z-\theta)^2}dz,\quad x\in\mathbb{R}.
\end{align*}

\textit{Geometric Brownian motion}. The scale function $S(\cdot)$ solves
\begin{align*}
\mu xS'(x)+\frac{1}{2}\sigma^2x^2S''(x)=0, \qquad \Rightarrow \quad
S_{GBM}(x)=x^{1-2\mu/\sigma^2},\quad x\in\mathbb{R}_{+}.
\end{align*}

\subsection{Transition probabilities of $\cM^\ast$ in boundary regimes}

Recall that at regimes $\underline{m}^-, \overline{m}^+$ only one player can switch.
For a recurrent $X$, she is guaranteed to do so eventually and  we simply have
\begin{align}\label{eq:AP_bdProb}
\bm{P}_{\underline{m}^-, (\underline{m}+1)^+}=\bm{P}_{\overline{m}^+, (\overline{m}-1)^-}=1.
\end{align}

When $X$ is transient, one player will be permanently dominant in the long-run and at least one of the following absorbing probabilities
\begin{equation}\label{eq:SMS_abprob}
\begin{split}
P_{\underline{m}^a}&:=\mathbb{P} \big[ X^{s^{2,\ast}_{\underline{m}+1}}_t \le s^{1,\ast}_{\underline{m}} \forall t \big] = \lim_{d\downarrow\underline{d}}\mathbb{P}\big[X^{s^{2,\ast}_{\underline{m}+1}}_{\tau(s^{2,\ast}_{\underline{m}+1};d,s^{1,\ast}_{\underline{m}})}=d\big],\\ 
P_{\overline{m}^a}&:=\lim_{u\uparrow\overline{d}}\mathbb{P}\big[X^{s^{1,\ast}_{\overline{m}-1}}_{\tau(s^{1,\ast}_{\overline{m}-1};s^{2,\ast}_{\overline{m}},u)}=u\big],\quad 
\end{split}
\end{equation}
are strictly positive. Namely, when $M^\ast$ enters a boundary regime, there is a positive probability that $M^\ast$ will stay constant henceforth. To address this, we use the states $\{\underline{m}^a,\,\overline{m}^a\}$ of the extended $\cM$ that are entered from the regime adjacent to the corresponding boundary. For instance, three transitions are possible from  $\cM^\ast_{n-1}\in\{(\overline{m}-1)^-, (\overline{m}-1)^+\}$:
\begin{align}
\cM^\ast_{n}=\begin{cases} \text{up to}\quad
\overline{m}^a,&\text{ if $\tilde{X}^{(n)}_t$ hits $s^{1,\ast}_{\overline{m}-1}$ before $s^{2,\ast}_{\overline{m}-1}$ and $\cM^\ast$ \textit{gets absorbed}},\\
\text{up to}\quad \overline{m}^+,&\text{ if $\tilde{X}^{(n)}_t$ hits $s^{1,\ast}_{\overline{m}-1}$ before $s^{2,\ast}_{\overline{m}-1}$ and $\cM^\ast$ \textit{is not absorbed}},\\
\text{down to}\quad (\overline{m}-2)^-,&\text{ if $\tilde{X}^{(n)}_t$ hits $s^{2,\ast}_{\overline{m}-1}$ before $s^{1,\ast}_{\overline{m}-1}$ }.
\end{cases}
\end{align}
Probabilistically, we may interpret absorption as an independent ``coin toss'' at the transition out of $(\overline{m}-1)^{\pm}$, so that
using \eqref{eq:SMS_abprob}
\begin{align}
\bm{P}_{(\overline{m}-1)^+, \overline{m}^a} &= \mathbb{P}\big[X^{s^{1,\ast}_{\overline{m}-2}}_{\tau(s^{1,\ast}_{\overline{m}-2};s^{2,\ast}_{\overline{m}-1},s^{1,\ast}_{\overline{m}-1})}=s^{1,\ast}_{\overline{m}-1}\big]\times P_{\overline{m}^a},\\
\bm{P}_{(\overline{m}-1)^+, \overline{m}^+} &= \mathbb{P}\big[X^{s^{1,\ast}_{\overline{m}-2}}_{\tau(s^{1,\ast}_{\overline{m}-2};s^{2,\ast}_{\overline{m}-1},s^{1,\ast}_{\overline{m}-1})}=s^{1,\ast}_{\overline{m}-1}\big]\times (1-P_{\overline{m}^a}),\\
\bm{P}_{(\overline{m}-1)^+, (\overline{m}-2)^-} & = 1 - \bm{P}_{(\overline{m}-1)^+, \overline{m}^a} - \bm{P}_{(\overline{m}-1)^+, \overline{m}^+}.
\end{align}
Similar computations are used for $\bm{P}_{(\overline{m}-1)^-, \cdot}, \bm{P}_{(\underline{m}+1)^-, \cdot}, \bm{P}_{(\underline{m}-1)^+, \cdot}$.

\subsection{Average sojourn times of $\cM^\ast$}\label{AP_sojourntime}
The expected sojourn times $\vec{\xi}$ of $M^\ast$ in \eqref{eq:SMS_xi}, or equivalently expected inter-arrival times between jumps of $\cM^\ast$ correspond to the mean two-sided exit time, $\delta_{ab}(x):=\mathbb{E}\big[\tau(x; a, b)\big], x\in(a,b)$, namely
\begin{align}\label{eq:xi}
\xi_{m^-}:=\mathbb{E}\big[\tau(s^{2,\ast}_{m+1};s^{2,\ast}_m,s^{1,\ast}_m)\big],\quad  \xi_{m^+}:=\mathbb{E}\big[\tau(s^{1,\ast}_{m-1};s^{2,\ast}_m,s^{1,\ast}_m)\big],
\end{align}
for $m\notin\partial\mathcal{M}$. Applying Dynkin's formula, it is well known that $\delta_{ab}(\cdot)$ solves the ODE $\mathcal{L}\delta+1=0$ with boundary conditions $\delta_{ab}(a)=\delta_{ab}(b)=0$.

\textit{Geometric Brownian motion.} The expected exit time $\delta_{ab}(\cdot)$ is a solution to
\begin{align*}
\mu x\delta'_{ab}(x)+\frac{1}{2}\sigma^2 x^2\delta''_{ab}(x)+1=0,\qquad x\in(a,b),\quad\text{and }\quad \delta_{ab}(a)=\delta_{ab}(b)=0.
\end{align*}
Solving we obtain
\begin{align*}	\delta_{ab}(x)=\left(\frac{1}{2}\sigma^2-\mu\right)^{-1}\left\{  \ln\left(\frac{x}{a}\right)+\ln(\frac{a}{b}) \frac{\left(x^{1-2\mu/\sigma^2}-a^{1-2\mu/\sigma^2}\right)}{b^{1-2\mu/\sigma^2}-a^{1-2\mu/\sigma^2}} \right\}, \qquad x \in (a,b).
\end{align*}

\subsection{One-sided exit times and sojourn times in boundary regimes}
To compute mean sojourn times $\xi_{\underline{m}^-},\xi_{\overline{m}^+}$ we make use of the one-sided passage times \begin{align*}
\tau(x;s):=\inf\{t\geq0\,:\,X^x_t = s\}.
\end{align*}
If  the corresponding absorbing probability \eqref{eq:SMS_abprob} is zero, we have
\begin{align*}
\xi_{\underline{m}^-}:=\mathbb{E}\big[\tau(s^{2,\ast}_{\underline{m}+1};s^{1,\ast}_{\underline{m}})\big],\qquad\quad \xi_{\overline{m}^+}:=\mathbb{E}\big[\tau(s^{1,\ast}_{\overline{m}-1};s^{2,\ast}_{\overline{m}})\big].
\end{align*}
Otherwise,  we condition on the exit time $\tau$ being finite, denoting
$ \delta_s(x) = \E[ \tau(x;s)\mathds{1}_{\{\tau(x;s)<\infty\}}]$. Then, e.g.
\begin{align}\label{eq:AP_NRsoTime}
\xi_{\overline{m}^+} = \mathbb{E}\big[\tau(s^{1,\ast}_{\overline{m}-1};s^{2,\ast}_{\overline{m}})\big|\tau(s^{1,\ast}_{\overline{m}-1};s^{2,\ast}_{\overline{m}})<\infty\big]
=\frac{1}{1-P_{\overline{m}^a}} \delta_{s^{2,\ast}_{\overline{m}}}(s^{1,\ast}_{\overline{m}-1}).
\end{align}

\textbf{Computing $\delta_{s}(x)$.}
We re-write \begin{align*}
\delta_s(x)= -\frac{\partial}{\partial \rho}\mathbb{E}_x\Big[e^{-\rho\tau(x;s)}\mathds{1}_{\{\tau(x;s)<\infty\}}\Big]\Bigg|_{\rho=0},
\end{align*}
and use the  well-known result~\cite{darling1953first} about the Laplace transform of  $\tau(x;s)$,
\begin{equation}\label{eq:AP_soTimeRec}
\E_x\Big[e^{-\rho\tau(x;s)}\mathds{1}_{\{\tau(x;s)<\infty\}}\Big]=\begin{cases}
\frac{F(x;\rho)}{F(a;\rho)}, &\qquad \text{if } x\leq s,\\
\frac{G(x;\rho)}{G(a;\rho)}, &\qquad \text{if } x\geq s,
\end{cases}
\end{equation}
where $F(\cdot;\rho)$ and $G(\cdot;\rho)$ are solutions to $(\mathcal{L}-\rho)u=0$ (recall \cref{pro:MS_OSP}) and we emphasize their dependence on the Laplace parameter $\rho$.

\textit{Ornstein-Uhlenbeck process.} Following from \eqref{eq:AP_soTimeRec},  the expected first passage time $\delta_s(x)$ to a level $s$ is admitted as
\begin{align}\label{eq:passage-a}
\delta_s(x) =\frac{\sqrt{2\pi}}{\mu} \left\{ \left[\int^{(s-\theta)\sqrt{\frac{2\mu}{\sigma^2}}}_{(x-\theta)\sqrt{\frac{2\mu}{\sigma^2}}}\Phi\left(z\right)e^{\frac{1}{2}z^2}dz \right] 1_{\{ s \geq x\}} + \left[ \int^{(\theta-s)\sqrt{\frac{2\mu}{\sigma^2}}}_{(\theta-x)\sqrt{\frac{2\mu}{\sigma^2}}}\Phi\left(z\right)e^{\frac{1}{2}z^2}dz
 \right] 1_{\{ s < x\}} \right\},
\end{align}
where $\Phi$ is the standard Gaussian cumulative distribution function. The expected exit time from an interval $x \in (a,b)$, $\delta_{ab}(x)$ can then be obtained via
\begin{align}\label{eq:passage-ab}
\delta_{ab}(x)=\frac{\delta_a(x)\delta_b(a)+\delta_b(x)\delta_a(b)-\delta_a(b)\delta_b(a)}{\delta_b(a)+\delta_a(b)}.
\end{align}

\textit{Geometric Brownian motion.}
GBM is non-recurrent; suppose that $\mu-\frac{1}{2}\sigma^2>0$ so that	$\overline{m}^a$ is the absorbing regime. Then from \eqref{eq:AP_soTimeRec} we compute
\begin{align*}
\delta_{s^{2,\ast}_{\overline{m}}}(s^{1,\ast}_{\overline{m}-1}) =\mathbb{E}\Big[\tau(s^{1,\ast}_{\overline{m}-1};s^{2,\ast}_{\overline{m}})\mathds{1}_{\{\tau(s^{1,\ast}_{\overline{m}-1};s^{2,\ast}_{\overline{m}})<\infty\}}\Big]
= \frac{1}{\mu-\frac{1}{2}\sigma^2} \cdot  \ln\Big(\frac{s^{1,\ast}_{\overline{m}-1}}{s^{2,\ast}_{\overline{m}}}\Big)\cdot\Big(\frac{s^{1,\ast}_{\overline{m}-1}}{s^{2,\ast}_{\overline{m}}}\Big)^{1-\frac{2\mu}{\sigma}}.
\end{align*}

\subsection{Expected number of switches until absorption under non-recurrent $(X_t)$}
Without loss of generality, let us assume that $\overline{m}$ is the absorbing regime, so that $\lim_{t \to \infty}M^\ast_t=\overline{m}$. Define
\begin{align*}
\upsilon^{up}_e:=&\mathbb{E}\bigg[\text{\#\textit{up-moves} before }\cM^\ast\text{ hits }\overline{m}^a\,|\,\cM^\ast_0=e\bigg], \qquad e\in E\setminus\{\overline{m}^a\},\\
\upsilon^{dn}_e:=&\mathbb{E}\bigg[\text{\#\textit{down-moves} before }\cM^\ast\text{ hits }\overline{m}^a\,|\,\cM^\ast_0=e\bigg], \qquad e\in E\setminus\{\overline{m}^a\},
\end{align*}
where $E$ is the state space of $\cM$ from \eqref{eq:state-E} and $\bm{P}$ is the transition matrix of $\cM^\ast$. Let  $\bm{P}_{-a}$ be the sub-matrix with the row and column corresponding to $\overline{m}^a$ removed.  Define $\vec{\upsilon}^{up}:=[\upsilon^{up}_{\underline{m}^-}, \cdots, \upsilon^{up}_{\overline{m}^+}]^T$, $\vec{\upsilon}^{dn}:=[\upsilon^{dn}_{\underline{m}^-}, \cdots, \upsilon^{dn}_{\overline{m}^+}]^T$,   $\vec{P}^{up}:=[P^{up}_{\underline{m}^-},\cdots, P^{up}_{\overline{m}^+}]^T$ and $\vec{P}^{dn}:=[P^{dn}_{\underline{m}^-},\cdots, P^{dn}_{\overline{m}^+}]^T$, with
\begin{align*}
\begin{cases}
P^{up}_{m^{\pm}}:= \bm{P}_{m^{\pm}, (m+1)^+},\text{ for } m<\overline{m}-1,&  \qquad P^{up}_{\overline{m}^+}:= 0,\\
P^{dn}_{m^{\pm}}:= \bm{P}_{m^{\pm}, (m-1)^-},\text{ for } m>\underline{m}, &\qquad P^{dn}_{\underline{m}^-}:= 0,\\
\end{cases}
\end{align*}
and $P^{up}_{(\overline{m}-1)^{\pm}}:= \bm{P}_{(\overline{m}-1)^{\pm}, \overline{m}^+}+\bm{P}_{(\overline{m}-1)^{\pm}, \overline{m}^a}.$ Then we obtain
\begin{align}
\vec{\upsilon}^{up}=(\bm{I}-\bm{P_{-a}})^{-1}\vec{P}^{up}\quad\text{and}\qquad \vec{\upsilon}^{dn}=(\bm{I}-\bm{P_{-a}})^{-1}\vec{P}^{dn},
\end{align}
and after taking care of the initial condition $X_0 = x$ which leads to a non-standard first transition probability, obtain the expected number of switches defined in \eqref{eq:SMS_numsw}
\begin{align}\notag
\mathbb{N}^1_m(x)&=\mathbb{P}_{x, m}\big[\cM^\ast_1=(m+1)^+\big]\times \big(\upsilon^{up}_{(m+1)^+}+1\big) +\mathbb{P}_{x,m}\big[\cM^\ast_1=(m-1)^-\big]\times \upsilon^{up}_{(m-1)^-},\\
\mathbb{N}^2_m(x)&=\mathbb{P}_{x, m}\big[\cM^\ast_1=(m+1)^+\big]\times \upsilon^{dn}_{(m+1)^+} +\mathbb{P}_{x,m}\big[\cM^\ast_1=(m-1)^-\big]\times \big(\upsilon^{dn}_{(m-1)^-}+1\big).\label{eq:expected-nim}
\end{align}
In general $\mathbb{P}_{x,m}\big[\cM^\ast_1=(m+1)^+\big]=\mathbb{P}(X^x_{\tau(x;s^{2,\ast}_m,s^{1,\ast}_m)}=s^{1,\ast}_m)$; however one must also consider the situation when $M^\ast_0=\overline{m}-1$, so that
$\cM^\ast_1 = \overline{m}^a$ becomes possible, and also $M^\ast_0=\overline{m}$, in which case one must assign $\cM_0 = \overline{m}^a$ or $\cM_0 = \overline{m}^+$ according to the probability $P_{\overline{m}^a}(x) :=\lim_{u\uparrow\overline{d}}\mathbb{P}\big[X^{x}_{\tau(x;s^{2,\ast}_{\overline{m}},u)}=u\big].$

\subsection{Non-recurrent $(X_t)$: expected time until absorption} To begin with, we need the expected number of visits to each non-absorbing regime. Define
\begin{align*}
\mathcal{V}_{e_1,e_2}:=\mathbb{E}\bigg[\text{\#\textit{visits} to $e_2$ before }\cM^\ast\text{ reaches }\overline{m}^a\,\bigg|\,\cM^\ast_0=e_1\bigg], \quad\text{ for all }e_1, e_2\in E\setminus\{\overline{m}^a\},
\end{align*}
and let $\bm{\mathcal{V}}$ denote the matrix of $\mathcal{V}_{e_1, e_2}$ with rows
$\vec{\mathcal{V}}_{e_1,\cdot}:=\big[\mathcal{V}_{e_1,\underline{m}^-}, \ldots, \mathcal{V}_{e_1, \overline{m}^+}\big]$, for  $e_1\in E\setminus\{\overline{m}^a\}$.
Then from standard Markov chain arguments,
\begin{align*}
\bm{\mathcal{V}}=(\bm{I}-\bm{P_{-a}})^{-1},
\end{align*}
where $\bm{P}_{-a}$ is the transient transition sub-matrix defined in the preceding subsection. Multiplying by the respective sojourn times $\xi_m$, the expected absorption time starting from an arbitrary regime $e\in E\setminus\{\overline{m}^a\}$ is
$\widetilde{\mathbb{T}}_{e}:= \vec{\mathcal{V}}_{e, \cdot} \cdot \vec{\xi}_{-a}$,
where $\vec{\xi}_{-a}$ is the vector of expected sojourn times excluding $\xi_{\overline{m}^a}$. Finally, the expected time until $M^\ast$ gets absorbed, as defined in \eqref{eq:SMS_abTime}, is admitted as (cf.~\eqref{eq:expected-nim})
\begin{align*}
\mathbb{T}_m(x)=\mathbb{E}_x\big[\tau(x;s^{2,\ast}_m, s^{1,\ast}_m)\big] + \mathbb{P}_{x, m}\big[\cM^\ast_1=(m+1)^+\big]\times \widetilde{\mathbb{T}}_{(m+1)^+} +\mathbb{P}_{x, m}\big[\cM^\ast_1=(m-1)^-\big]\times \widetilde{\mathbb{T}}_{(m-1)^-}.
\end{align*}
Again further adjustments are needed when $M^\ast_0=\overline{m}-1$  or $M^\ast_0=\overline{m}$ as discussed in the preceding subsection.

\section{Proof of \cref{cor:11case}}\label{AP:prof11case}
\begin{proof}
From \eqref{eq:sys11b} and \eqref{eq:sys11c}, we write $V^1_{+1}(x)$ explicitly for $x\geq \ms$ by substituting in the respective expressions for $\nu^1_{+1}$ and $\omega^1_{-1}$:
\begin{align*}
	V^1_{+1}(x) & = D^1_{+1}(x)+\nu^1_{+1} G(x)\\
	&=D^1_{+1}(x)+\frac{V^1_{-1}(-\ms)-D^1_{+1}(-\ms)}{G(-\ms)}G(x)\\
	&=D^1_{+1}(x)+\frac{D^1_{-1}(-\ms)+\omega^1_{-1} F(-\ms)-D^1_{+1}(-\ms)}{G(-\ms)}G(x)\\
	&=D^1_{+1}(x)+\frac{G(x)}{G(-\ms)}\Big[(V^1_{+1}-D^1_{-1}-K^1_{-1})(\ms)\frac{F(-\ms)}{F(\ms)}+D^1_{-1}(-\ms)-D^1_{+1}(-\ms)\Big].
\end{align*}
The above gives an equation relating $V^1_{+1}(x)$ to $V^1_{+1}(\ms)$; therefore, if one defines
\begin{align} Q(s):=\frac{D^1_{+1}(s)-\frac{G(s)}{G(-s)}\Big[\big(D^1_{-1}(s)+K^1_{-1}(s)\big)\frac{F(-s)}{F(s)}+D^1_{+1}(-s)-D^1_{-1}(-s)\Big]}{1-\frac{G(s)}{G(-s)}\frac{F(-s)}{F(s)}},
\end{align}
and let $\ms$ be a solution to the system \eqref{eq:sys11}, then it holds $Q(\ms)=V^1_{+1}(\ms)$. Similarly, after differentiating with respect to $x$ (guaranteed by \cref{cor:11case} which requires smoothness of $D^1_m(\cdot)$ and $K^1_m(\cdot)$),
one can define
\begin{align}
	q(s):=D^{1'}_{+1}(s)+\frac{G'(s)}{G(-s)}\Big[(Q(s)-D^1_{-1}(s)-K^1_{-1}(s))\frac{F(-s)}{F(s)}+D^1_{-1}(-s)-D^1_{+1}(-s)\Big],
\end{align}
and conclude $q(\ms)=V^{1'}_{+1}(\ms)$. Then replacing $V^1_{+1}(x)$ by $Q(x)$ and $(V^{1}_{+1})'(x)$ by $q(x)$ in \eqref{eq:sys11a} we obtain that  solving the system \eqref{eq:sys11} is equivalent to finding the root(s) of
\begin{align}\label{eq:11_alt}
	\mathcal{Z}(s):=\big[Q(s)-D^1_{-1}(s)-K^1_{-1}(s)\big]F'(s)-\big[q(s)-(D^{1}_{-1})'(s)-(K^{1}_{-1})'(s)\big]F(s) = 0,
\end{align}
Since  $\ms>s^{2,\ast}_{+1}=-\ms\Longrightarrow \ms>0$ (otherwise the switching regions would overlap), we seek positive solutions to \eqref{eq:11_alt}. We shall show that $\mathcal{Z}(0) < 0$ and $\mathcal{Z}(s) > 0$ for $s$ large enough, which by continuity (as each term in \eqref{eq:11_alt} is continuous) implies the existence of a root.

{ On the one hand, the numerator of $Q(s)$ at $s=0$ is admitted as $$D^1_{+1}(0)-\frac{G(0)}{G(0)}\Big[\big(D^1_{-1}(0)+K^1_{-1}(0)\big)\frac{F(0)}{F(0)}+D^1_{+1}(0)-D^1_{-1}(0)\Big]=-K^1_{-1}(0)<0,$$
	while the denominator $1-\frac{G(s)}{G(-s)}\frac{F(-s)}{F(s)}=1-\big(\frac{F(-s)}{F(s)}\big)^2$ is strictly positive ($F(\cdot)$ is increasing) for $s>0$ and tends to zero as $s \downarrow 0$, so that $\lim_{s\downarrow 0 }Q(s) =-\infty$. Furthermore,
	\begin{align*}
		\lim_{s\downarrow 0 }q(s)= (D^{1}_{+1})'(0)+\frac{G'(0)}{G(0)}\Big[(\lim_{s\downarrow 0 }Q(s)-D^1_{-1}(0)-K^1_{-1}(0))\frac{F(0)}{F(0)}+D^1_{-1}(0)-D^1_{+1}(0)\Big]=+\infty,
	\end{align*}
	since $G(\cdot)$ is positive and decreasing ($G'(\cdot) < 0$) while all other terms beyond $\lim_{s\downarrow 0 }Q(s)$ are finite. Putting everything together,
	$$\lim_{s\downarrow 0 }\mathcal{Z}(s)=\big[\lim_{s\downarrow 0 }Q(s)-D^1_{-1}(0)-K^1_{-1}(0)\big]F'(0)-\big[\lim_{s\downarrow 0 }q(s)-(D^{1}_{-1}-K^{1}_{-1})'(0)\big]F(0) =-\infty,$$
	since $F(\cdot)$ is positive and increasing, and all other terms are finite.
	
	On the other hand, for $s$ large enough  and using that~\cite[Sec.~2]{borodin2012handbook}
	\begin{align*} \lim_{x\downarrow\underline{d}}F(x)=0,\qquad\lim_{x\downarrow\underline{d}}G(x)=+\infty,\qquad\lim_{x\uparrow\bar{d}}F(x)=+\infty,\qquad\lim_{x\uparrow\bar{d}}G(x)=0,
	\end{align*}
 we have $Q(s)\approx D^1_{+1}({s}), q({s})\approx (D^{1}_{+1})'({s})$ asymptotically as $s\uparrow \bar{d}$ and hence
	\begin{align*}
		\mathcal{Z}(s)&\approx \big[D^1_{+1}({s})-D^1_{-1}({s})-K^1_{-1}(\bar{s})\big]F'({s})-\big[D^{1'}_{+1}({s})-D^{1'}_{-1}({s})-K^{1'}_{-1}({s})\big]F({s}),\\
		&= \Big[-\big(\frac{D^1_{+1}({s})-D^1_{-1}({s})-K^1_{-1}({s})}{F({s})}\big)'\Big]\cdot F^2({s}) > 0.
	\end{align*}
	The last inequality follows from $\Delta D := D^1_{+1}-D^1_{-1}-K^1_{-1} \in  \mathcal{H}_{\text{inc}}$ (cf.~\cref{def:AP_Hinc}), thus
	\begin{align*}
		\begin{cases}
			\displaystyle\limsup_{s\uparrow\bar{d}} \frac{\Delta D(s)}{F(s)}=0,\\
			\Delta D(s) >0,\text{ for $s$ large},
		\end{cases}\quad\Longrightarrow\quad \bigg(\frac{\Delta D({s})}{F({s})}\bigg)'<0 \quad \text{as} \quad s \uparrow \bar{d}.
	\end{align*}
}  \end{proof}


\bibliographystyle{siamplain}
{\large \bibliography{mybib}}

\end{document}